%% file: main.tex
\documentclass{article}
\usepackage[margin=1in]{geometry}
\usepackage{amsmath, amssymb, amsfonts}
\usepackage{amsthm}
\usepackage{thmtools}
\usepackage{thm-restate}
\usepackage{authblk}

\usepackage{bbm}
\usepackage{diagbox}

\usepackage{graphicx}
\usepackage{mathtools} % For \coloneqq, etc.
\usepackage[dvipsnames]{xcolor}
\usepackage[hypertexnames=false]{hyperref}
\usepackage{cleveref}
\usepackage{svg}
\usepackage{enumitem}
\usepackage{quiver}
\usetikzlibrary{calc,shapes.geometric}
\input{macros.sty}
\usepackage{braket}

\hypersetup{
    colorlinks=true,
    linkcolor=NavyBlue,
    citecolor=Maroon,
    urlcolor=blue
}

\theoremstyle{plain}
\newtheorem{theorem}{Theorem}[section]
\newtheorem{proposition}[theorem]{Proposition}
\newtheorem{lemma}[theorem]{Lemma}
\newtheorem{corollary}[theorem]{Corollary}

\makeatletter
\newcommand{\BeginLinkedRestatable}[1]{%
    \def\linkedrestate@target{#1}%
    \let\linkedrestate@orig@begintheorem\@begintheorem
    \let\linkedrestate@orig@opargbegintheorem\@opargbegintheorem
    \def\@begintheorem##1##2{%
        \linkedrestate@orig@begintheorem{##1}{\hyperlink{\linkedrestate@target}{##2}}%
    }%
    \def\@opargbegintheorem##1##2##3{%
        \linkedrestate@orig@opargbegintheorem{##1}{\hyperlink{\linkedrestate@target}{##2}}{##3}%
    }%
}
\newcommand{\LinkedRestatableTarget}[1]{\hypertarget{#1}{}}
\newcommand{\EndLinkedRestatable}{%
    \let\@begintheorem\linkedrestate@orig@begintheorem
    \let\@opargbegintheorem\linkedrestate@orig@opargbegintheorem
}

\makeatother

\theoremstyle{definition}
\newtheorem{definition}[theorem]{Definition}

\theoremstyle{remark}
\newtheorem{remark}[theorem]{Remark} % Using "Note" as the environment name to match the source
\newtheorem{example}[theorem]{Example}

\numberwithin{equation}{section} %equation numbering with section numbers

\title{Bulk-boundary correspondence of (1+1)D symmetric gapped phases}

\date{\today}
\author[a]{Yizhou Ma}
\author[a]{Gen Yue}
\author[a,b]{Tian Lan\thanks{\href{mailto:tlan@cuhk.edu.hk}{tlan@cuhk.edu.hk}}}
\affil[a]{Department of Physics, The Chinese University of Hong Kong, Shatin, New Territories, Hong Kong, China}
\affil[b]{The State Key Laboratory of Quantum Information Technologies and Materials, The Chinese University of Hong Kong, Shatin, New Territories, Hong Kong, China}

\begin{document}

\maketitle

\begin{abstract}
We develop an operator-algebraic framework for boundary conditions and bulk-boundary correspondence in one-dimensional gapped phases with categorical symmetry. 
Working directly in the thermodynamic limit, we construct half-infinite fusion spin chains and commuting-projector boundary Hamiltonians from a unitary fusion category \(\C\), an indecomposable semisimple right \(\C\)-module category \(\M\), a Q-system \(Q\in\C\) specifying the bulk phase, and a right \(Q\)-module \(K\in\MQ\), regarded as an object of \(\MQ^\op\), specifying the boundary. 
We prove that these Hamiltonians have unique ground states and that the resulting realization functor \(\MQ^{\op}\to\BCond\) is an equivalence, so simple boundary conditions are classified by simple objects of \(\MQ\) and general boundary conditions by their finite direct sums.
We also give a microscopic formulation of the boundary symmetry topological field theory using DHR bimodules of the boundary quasi-local algebra. For a half-infinite fusion spin chain, the boundary DHR category is monoidally equivalent to \((\CMdual)^{\rev}\), and the canonical action of the bulk DHR category on it agrees with the categorical action of \(Z_1(\C^\rev)\). 
Finally, we identify the action of the boundary DHR category on boundary conditions with the categorical action of \((\CMdual)^{\rev}\) on \(\MQ^\op\). This yields a one-dimensional bulk-boundary correspondence: the enriched monoidal category describing the bulk is the enriched center of the enriched category describing the boundary.
\end{abstract}

\tableofcontents

\section{Introduction}

Topological phases of matter are a class of quantum matter beyond the Landau paradigm. They exhibit interesting behaviors, such as fractional braiding statistics and topologically protected quantum information. The study of topological phases has become one of the central themes of condensed matter physics in recent decades. 

The presence of symmetry makes the story richer. Even in one spatial dimension, where intrinsic topological order is absent for ordinary bosonic systems, gapped phases with symmetry can carry protected or symmetry-enriched data. 
For finite group symmetries, the classification of one-dimensional SPT phases by \(H^2(G,U(1))\) is well understood, both from matrix product state (MPS) techniques \cite{Chen_Gu_Wen_2011_SPTMPS,Cirac_2011_1DPhaseMPS} and from operator-algebraic methods that do not assume an MPS form \cite{Kapustin_Sopenko_2021_SPTClassification,Ogata_2021_ClassificationSPT}. 
More recent work has studied defects and symmetry fractionalization in related operator-algebraic settings \cite{Ogata_2024_SET,Kawagoe_Vadnerkar_2025_SET}.

\paragraph{Bulk-boundary correspondence and symmetry}
One of the most intriguing phenomena in topological physics is the \emph{bulk-boundary correspondence}: the boundaries uniquely determine the bulk. This phenomenon was first noticed in the context of quantum field theories \cite{Witten_1988_QFTJones} and later discovered in condensed matter systems. As a well-known example, let $\cA$ be the unitary fusion category of particles living on a certain gapped boundary of a (2+1)D topological order described by the unitary modular tensor category $\C$, then we have the bulk-boundary relation $\C=Z_1(\cA)$, where $Z_1$ here is the Drinfeld center \cite{KitaevKong_2012_DomainWall}; the relation can be verified in string-net models~\cite{Levin_2005}, and more generally, systems satisfying the local topological order (LTO) axioms~\cite{Penneys_Add_2025,Jones_Naaijkens_Penneys_Wallick_Izumi_2025}.

The bulk-boundary correspondence holds for a wide range of topological phases, though specific statements vary. In higher dimensions, the bulk-boundary relation can be formulated using the theory of higher fusion categories \cite{Kong_2022_QL1,Kong_2024_QL2,Johnson_Freyd_2022_OnClassificationTO,Lan_2024_CatSET}.
There is also a proposal of bulk-boundary correspondence for gapless boundaries of (2+1)D topological orders. The macroscopic observables on gapless boundaries of (2+1)D topological orders are proposed to form an \emph{enriched fusion category} ${}^{\Mod_V}\cS$ \cite{Kong_2018_GaplessEnrichedCat,Chen_2020_TopologicalPhaseTransition,Kong_2020_GaplessEdge}, where $\cS$ is the fusion category of topological defect lines on the (1+1)D world sheet, $V$ is the local quantum symmetry emergent in the IR limit that is described by a vertex operator algebra, and $\Mod_V$ is the category of spaces of $V$-symmetric local fields. It is proposed that the bulk topological order $\cC$ is the enriched center of the gapless boundary, $\C=Z_1({}^{\Mod_V}\cS)$.

It is proposed that all these bulk-boundary relations have closely related origins. In Ref.~\cite{Kong_2015_BbyBulk,Kong_2017_BdyBulk}, the authors argue that under certain physically reasonable assumptions such as uniqueness of the bulk and topological invariance, (not necessarily gapped) quantum liquid phases satisfy the bulk boundary relation: the bulk phase is the \emph{center} of the boundary phase, if we find their proper mathematical descriptions.

In the presence of symmetry, we can also investigate the bulk-boundary correspondence of symmetry-enriched topological phases. 
In~\cite{Kong_2022_1DEnrichedCat, Xu_2024_1DPhaseAbelianSym, LanZhou_2024_QuantumCurrent}, the authors studied (1+1)D gapped phases with symmetry and their bulk-boundary relations. These works use enriched monoidal categories ${}^{Z_1(\Rep(G))}\cS$ to describe the bulk (1+1)D gapped phases with finite group symmetry $G$, where $\cS$ is the category of topological defects that are transparent to symmetry defects (for example $\cS=\Rep(G)$ for $G$-SPT phases), and the enriched background $Z_1(\Rep(G))$ is interpreted as the category of sectors of certain non-local operators. The enriched background $Z_1(\Rep(G))$ is completely determined by the symmetry and doesn't depend on the phase under study. 

Note that $Z_1(\Rep(G))$ (or more generally $Z_1(\cC)$ if we consider generalized symmetries) is also the mathematical description of the (2+1)D topological order that can be realized by string-net model with input $\Rep(G)$ (or more generally $\cC$). This duality between symmetry and topological order in one higher dimension is known as \emph{topological holography}~\cite{KLW+2005.14178,Kong_2020Classi,Ji_2020_CategoricalSym, Kong_2020_GaplessEdge, 2021Gaiotto,Lichtman_2021, 
 Kong_2022_QL1, 
Chatterjee_2023,Moradi_2023, 2023SymTFTfromString,freed2024, Bhardwaj2024,Lan2412.07198,2026arXiv260506661Y}. The (2+1)D topological order $Z_1(\cC)$ is named as \emph{symmetry topological field theory} (SymTFT) in the literature. In this language, the bulk-boundary correspondence is expected to be a statement about centers of enriched categories: boundaries together with the boundary SymTFT form an enriched category, whose center is the category of bulk defects enriched in the bulk SymTFT.

However, there are gaps between the categorical description and the microscopic theory. In one-dimensional examples, the SymTFT has often been described in terms of ``symmetric non-local operators.'' This description is physically illuminating in simple examples, but awkward when we consider non-abelian or non-invertible symmetries. Nor is it conceptually satisfactory: non-local operators are not elements of the quasi-local algebra, and it is unknown how to understand non-invertible objects of the SymTFT or their direct sums.
Moreover, no general microscopic definition of boundary conditions was given in these works. Consequently, discussions of symmetries beyond nonabelian groups were avoided, as they necessarily involve \emph{composite} boundary conditions, whose description and formulation was unclear.

These developments suggest several basic questions to be answered: What is the correct microscopic notion of a boundary condition for a topological phase with symmetry? How should such boundary conditions be classified? What is the symmetry topological field theory, or SymTFT, seen by a boundary? Finally, what is the precise bulk-boundary correspondence in this one-dimensional symmetric setting?

All these suggest that a new framework is needed. In this paper, we propose a framework that bridges categorical description with microscopic theory and answer these basic questions.

\paragraph{The thermodynamic limit and operator algebras}

Our framework is built entirely on \emph{the thermodynamic limit}, i.e. with infinite-size spin systems, using the language of \emph{operator algebras}.

Why must we work in the thermodynamic limit? The answer is that many structures don't exist in finite systems and only emerge when the system size is infinite. 
First of all, phases and phase transitions only make sense in the thermodynamic limit. 
Also, as the dichotomy of ``local'' and ``non-local'' is only clear in infinite systems, superselection sectors are only well-defined in the thermodynamic limit as modules of the algebra of (quasi-) local operators.

If the reader still feels uneasy, we suggest thinking about Newton's first law. In reality there is no object free from any external force, but Newton's laws are nevertheless most naturally formulated by considering the ideal situation. We believe that the thermodynamic limit is the ideal situation to study topological phases.

Operator algebras are the natural language to work in the thermodynamic limit.
Finite spin chains have ordinary finite-dimensional Hilbert spaces, but there is no reasonable infinite tensor product of on-site Hilbert spaces for an infinite chain. In contrast, the algebra of quasi-local observables
\[
    \Loc=\varinjlim_I \Loc_I
\]
is naturally defined as a unital approximately finite (AF) \(C^*\)-algebra. A state is then just a positive unital functional on \(\Loc\).

A key concept in studying topological phases is superselection sectors. They turn out to be formulated most naturally in the operator algebraic language: A superselection sector is a representation or Hilbert module of \(\Loc\). This captures the physical idea that states differing by a local operator belong to the same sector. Also, direct sums of charges, whose interpretation was a constant source of confusion, obviously correspond to direct sums of modules.

This operator-algebraic viewpoint has already proved useful in topological phases. Naaijkens' analysis of Kitaev's toric code describes excitations as localized endomorphisms, or equivalently representations, of the observable algebra and computes their fusion and braiding within a DHR/BF-type superselection theory \cite{Naaijkens_2011_ToricCode}. Related sector theories for non-abelian quantum double models were studied in \cite{Naaijkens_2025_SectorsQD}. Cha, Naaijkens, and Nachtergaele proved stability of the resulting braided tensor \(C^*\)-category of charges under gapped perturbations by replacing strict localization with almost-localization in cone-like regions \cite{Naaijkens_Cha_2019_StabilityCharges}. Ogata derived braided \(C^*\)-tensor categories from two-dimensional gapped ground states satisfying approximate Haag duality and showed that the resulting anyon category is a phase invariant \cite{Ogata_2022_BraidedCStarTensorCat}; a related operator-algebraic framework has also been proposed for mixed-state topological order \cite{Ogata_2025_MixedState}. Our work follows this philosophy in the one-dimensional symmetric setting, using quasi-local algebras and Hilbert bimodules as the microscopic foundation for boundary conditions, SymTFTs, and their bulk-boundary relation.

\paragraph{Summary of results}
This paper classifies boundary conditions and establishes bulk-boundary correspondence for one-dimensional phases with categorical symmetry. 
Before stating our results, we would like to remind the reader that conventions in this paper are different in various ways from conventions in the literature; these differences are discussed in Appendix~\ref{app:orientation}.

First, we construct the corresponding half-infinite fusion spin chain and commuting-projector boundary Hamiltonian directly in the thermodynamic limit, from the following data: a unitary fusion category \(\C\), interpreted as the category of local symmetry charges; an indecomposable semisimple right \(\C\)-module category \(\M\), interpreted as the category of boundary local quantum charges; a bulk gapped phase is specified by a Q-system \(Q\) in \(\C\), and a microscopic boundary condition is specified categorically by a right \(Q\)-module \(K\in \M_Q\). We prove that the Hamiltonian has a unique ground state (Theorem~\ref{thm:boundary_unique_GS}).

Furthermore, we proved that all possible boundaries are realized by our construction:
\BeginLinkedRestatable{target:bcondclassification}
\begin{restatable*}[Classification of boundary conditions]{theorem}{bcondclassification}
\label{thm:intro-bcond-classification}
The realization functor $\ReaBCond: \MQ^\op \to \BCond$ is an equivalence of categories.
\end{restatable*}
\EndLinkedRestatable
Therefore, simple boundary conditions are classified by the simple objects of \(\M_Q\), and non-simple boundary conditions are their finite direct sums. This statement is stronger than merely constructing examples: every microscopic boundary condition satisfying the definition in this paper is obtained from a right \(Q\)-module in \(\M\). However, we find it mysterious that one must take the \emph{opposite} category in microscopic constructions to obtain reasonable morphisms between boundary conditions, whose explanation is left to future work.

The boundary SymTFT is formulated using DHR bimodules of fusion spin chains with boundary. 
Jones showed that for a fusion spin chain without boundary, the DHR category of the quasi-local algebra recovers the Drinfeld center \(Z_1(\C^\rev)\)
\footnote{The original reference proves that the DHR category is equivalent to $Z_1(\C)$, as their convention is different from ours. See Appendix~\ref{app:orientation:dhr} for a detailed discussion.}, 
giving a rigorous operator-algebraic formulation of the bulk SymTFT \cite{Jones_2024_DHR}. 
We extend this viewpoint to half-infinite fusion spin chains, and it turns out the boundary DHR category recovers \((\CMdual)^\rev\), the monoidal reverse of the dual tensor category of $\C$ with respect to $\M$:
\BeginLinkedRestatable{target:boundarysymtft}
\begin{restatable*}{theorem}{boundarysymtft}
\label{thm:intro-boundary-symtft}
The standard action functor $\ReaDHR \colon (\CMdual)^{\rev} \to \DHR(\Loc^{\bdy}_\bullet)$ is a monoidal equivalence.
\end{restatable*}
\EndLinkedRestatable
Thus the boundary SymTFT is realized microscopically as a DHR category of the boundary observable algebra. We have also found that there is a canonical action of the bulk SymTFT on the boundary SymTFT, which agrees with the canonical action of $Z_1(\C^\rev)$ on $(\CMdual)^{\rev}$:
\BeginLinkedRestatable{target:symtftbulkbdy}
\begin{restatable*}{theorem}{symtftbulkbdy}
\label{thm:SymTFT-bulk-bdy}
    The action of $\DHR(\Loc^{\bulk}_{\bullet})$ on $\DHR(\Loc^{\bdy}_\bullet)$ is canonically equivalent to the canonical action of $Z_1(\C^\rev)$ on $(\CMdual)^{\rev}$.
\end{restatable*}
\EndLinkedRestatable

The action of the boundary SymTFT on boundary conditions is realized rigorously using the language of operator algebras:
\BeginLinkedRestatable{target:bulkbdyaction}
\begin{restatable*}{theorem}{bulkbdyaction}
\label{thm:intro-bulk-boundary-action}
The action of the boundary DHR category on boundary conditions by relative tensor product of Hilbert bimodules,
\[
    \DHR(\Loc^{\bdy}_\bullet) \times \BCond \longrightarrow \BCond
\]
is identified with the categorical action of \((\CMdual)^{\rev}\) on \(\MQ^\op\) under the equivalences of Theorems~\ref{thm:intro-bcond-classification} and~\ref{thm:intro-boundary-symtft}, with the duality/bending isomorphism accounting for the contravariance of \(\ReaBCond\).
\end{restatable*}
\EndLinkedRestatable
Consequently, the full boundary datum is not merely the linear category of boundary conditions; it is the category of boundary conditions equipped with its boundary SymTFT action. 

This implies the main theorem of the paper, the bulk-boundary correspondence of 1D phases with symmetry:
\BeginLinkedRestatable{target:bulkbdycorrespondence}
\begin{restatable*}[Bulk-boundary correspondence]{theorem}{bulkbdycorrespondence}
\label{thm:bulk-boundary-correspondence}
The enriched monoidal category describing the bulk of a (1+1)D symmetric gapped phase is the enriched center of the enriched category describing the boundary,
\[
     Z_0({}^{\DHR(\Loc_\bullet^\bdy)}\BC) \simeq {}^{\DHR(\Loc_\bullet^\bulk)}\mathrm{TopDef}
\]
\end{restatable*}
\EndLinkedRestatable

\paragraph{Organization.}
The paper is organized as follows. In Section~\ref{section: Preliminaries}, we review the operator-algebraic background used throughout the paper, including states, representations, AF algebras, and commuting-projector interactions in the thermodynamic limit. In Section~\ref{section: fusion spin chains with bdy}, we introduce fusion spin chains with and without boundary, explain how they generalize spin chains with finite group symmetry, and establish structural properties such as covering and algebraic Haag duality. In Section~\ref{Q-system}, we construct Q-system commuting-projector models for one-dimensional phases and their boundary analogues. In Section~\ref{subsec: Q-system models with bdy} and the following section on boundary conditions, we construct boundary Hamiltonians and prove the classification theorem of boundaries. In Section~\ref{ComputationDHRB}, we define the boundary DHR category and show that it is equivalent to \((\CMdual)^{\rev}\), and study the action of the bulk DHR category on the boundary DHR category. Finally, in Section~\ref{section: bulk-boundary correspondence}, we identify the DHR action on boundary conditions with the canonical categorical action and establish the bulk-boundary correspondence as an enriched-center statement.

% ----------------------------------------

\section{Preliminaries}
\label{section: Preliminaries}
In this section, we review the operator-algebraic and categorical preliminaries used later in the paper, emphasizing the physical meaning of the constructions while keeping the mathematical statements precise.

\begin{table}
    \centering
    \small
    \begin{tabular}{|p{0.25\textwidth}|p{0.39\textwidth}|p{0.25\textwidth}|}\hline
        Notation & Meaning & Reference\\\hline

        $\C$ & The category of local quantum charges, which is a unitary fusion category & Definition~\ref{def:fusion-spin-chain}\\\hline
        $\M$ & The category of boundary local quantum charges, which is a \emph{right} module category of $\C$ & Definition~\ref{def:fusion-spin-chain-with-boundary}\\\hline
        $x \in \C$ & Degree of freedom on bulk sites& \\\hline
        $l \in \M$ & Degree of freedom on the boundary site& \\\hline
        
        $Q$ & The Q-system labelling the bulk phase& Definition~\ref{def:q-system-model}\\ \hline
        $m: Q \otimes Q \to Q$ & The multiplication morphism of the Q-system& Definition~\ref{def:q-system-model}\\\hline
        $\mu: K \otimes Q \to K$ & The action morphism of the right $Q$-module $K$& Definition~\ref{def:q-system-model-with-boundary}; Eq.~\eqref{eq.BdyFrob}\\\hline

        $\Loc_\bullet$ & The net of local algebras& Definition~\ref{def:abstract-spin-chain}\\\hline
        $\Loc$ & The quasi-local algebra& Definition~\ref{def:abstract-spin-chain}\\\hline
        $\DHR(\Loc^{\bdy}_\bullet)$ & The boundary DHR category& Definition~\ref{def:boundary-dhr-bimodule}\\\hline

        $\BCond$ & The category of boundary conditions& Definition~\ref{def:boundary-condition}\\\hline
        $\CMdual$ & The Morita dual $\mathrm{Fun}_{\C^\rev}(\M,\M)$, using right module categories& Eq.~\eqref{eq:morita-dual}; Appendix~\ref{app:orientation:right-modules}\\\hline
        $\ReaDHR: (\CMdual)^{\rev} \to \DHR(\Loc^{\bdy}_\bullet)$ & The realization functor for DHR bimodules& Eq.~\eqref{eq:std-finite-piece}; Section~\ref{ComputationDHRB}; Appendix~\ref{app:orientation:dhr}\\\hline
        $\ReaBCond: \MQ^{\op} \to \BCond$ & The realization functor for boundary conditions& Definition~\ref{Def. The realization functor for boundary conditions}; Appendix~\ref{app:orientation:defects}\\\hline
    \end{tabular}
    \caption{Table of notations}
\end{table}

\subsection[C*-algebras and states]{\(C^*\)-algebras and states}
It has long been known that there are at least two formulations of quantum mechanics: The formulation via Hilbert spaces, and the formulation via operator algebras.
The latter one is known as the algebraic approach to quantum theory, which is the cornerstone of algebraic quantum field theory (AQFT); in this paper we exclusively use this approach.

The reader might ask: What's the point of abandoning familiar Hilbert spaces and fumbling with \(C^*\)-algebras? The answer is, as the notion of phases only makes sense in the thermodynamic limit, we hope to work with infinite lattice systems.
The Hilbert space approach would encounter some problems in this setting. For example, if we try to study the toric code on an infinite plane using the Hilbert space approach, we will find that the ground state cannot be normalized: the number of closed-string configurations is infinite, therefore the amplitude of each configuration has to vanish.

Mathematically, the problem is related to the definition of infinite products of Hilbert spaces. The main issue is that there is no canonical linear map $V_i \hookrightarrow V_i \otimes V_j$, which is required to define an inductive limit. One approach is to specify a unit vector $w_j \in V_j$ for each $j$ and define $V_i \hookrightarrow V_i \otimes V_j$ by $v \mapsto v \otimes w_j$ to define the infinite product of local Hilbert spaces as the inductive limit, however the resulting space could be strange; for example, different choices of special vectors $w_j$ could give rise to very different results \cite{Naaijkens2017_book}.

On the other hand, there is a good definition of the algebra of operators on infinite lattice systems: One first defines the algebra of operators on finite regions as usual, then defines the \emph{quasi-local algebra} as the inductive limit of algebras of operators on finite regions.
The algebraic approach also provides natural definitions of states and superselection sectors.
Therefore, we would briefly review basic notions of the algebraic approach in this subsection.
We refer the reader to \cite{Blackadar_OperatorAlg_2017,Murphy90} for detailed accounts of the theory of operator algebras.

The basic concept is \emph{$C^*$-algebras}, which are complete normed algebras equipped with a *-operation satisfying the $C^*$ axiom:
\begin{definition}[{\cite[Definition II.1.1.1]{Blackadar_OperatorAlg_2017}}]
A \emph{Banach algebra} is a complex algebra equipped with a submultiplicative norm (which means $\|xy\| \leq \|x\| \|y\|$), and the norm makes the algebra a Banach space.
An \emph{involution} on a Banach algebra $A$ is an anti-linear isometric anti-automorphism of order two, usually denoted $x \mapsto x^*$.
In other words, $(x+y)^*= x^*+y^*$, $(xy)^*=y^*x^*$, $(\lambda x)^*=\bar{\lambda} x^*$, $(x^*)^*=x$, $\|x^*\|=\|x\|$ for all $x, y \in A$, $\lambda \in \mathbb{C}$.
Finally, a \emph{\(C^*\)-algebra} is a Banach algebra $A$ equipped with an involution satisfying the \(C^*\)-axiom
\[
    \left\|x^* x\right\|=\|x\|^2.
\]
\end{definition}

\begin{example}[Equivalence of the algebraic approach and the Hilbert-space approach]
Given any Hilbert space $\mathcal{H}$, the algebra of bounded operators $\mathcal{L}(\mathcal{H})$ is a \(C^*\)-algebra, with the norm given by the usual operator norm ($\|T\| := \sup_{v \neq 0} \{ \|T(v)\| / \|v\| \}$) and the involution given by taking adjoint.
When $\mathcal{H}$ is $n$-dimensional, the set of $n \times n$ matrices form a \(C^*\)-algebra with the usual involution (conjugate transpose) and operator norm.
\end{example}

\subsection{States and representations}\label{subsec:states-and-reps}

Density matrices are translated to linear functionals in the algebraic formulation of quantum mechanics:

\begin{definition}
A linear functional $\phi$ on a \(C^*\)-algebra $A$ is \emph{positive} if for any positive element $a \in A$, $\phi(a) \geq 0$.
A \emph{state} is a positive linear functional satisfying $\phi(1_A) = 1$ where $1_A$ is the unit of $A$.
$\phi$ is a \emph{mixed state} if it can be written as a proper convex combination, i.e., $\phi = t \phi_1 + (1-t) \phi_2$, where $\phi_1, \phi_2$ are two different states and $t \in (0,1)$.
$\phi$ is a \emph{pure state} if it is not a mixed state.
\end{definition}

The physical meaning is, a state is a function sending operators to their expectation values.
For example, given a Hilbert space $\mathcal{H}$ and a unit vector $| \psi \rangle \in \mathcal{H}$, the linear functional $\phi \colon \mathcal{L}(\mathcal{H}) \to \mathbb{C}$ defined by $O \mapsto \langle \psi| O |\psi \rangle$ is a state.

We note that this formulation is equivalent to the density matrix formulation in the case of finite-dimensional Hilbert spaces.
However, when we are dealing with infinite-dimensional Hilbert spaces or spin systems with anomalous symmetries, the operator algebraic formulation is a lot more useful.

\begin{definition}
A \emph{representation} of a \(C^*\)-algebra $A$ is a *-homomorphism from $A$ to $\mathcal{L}(\mathcal{H})$ for some Hilbert space $\mathcal{H}$.
We denote the category of representations of $A$ and intertwiners as $\Mod(A)$.
\end{definition}

The physical meaning for a representation is a superselection sector.
Naturally, a simple representation is a simple superselection sector, and direct sums of representations correspond to direct sums of superselection sectors.

\begin{remark}
This provides a clarification of the long-standing confusion about the physical meaning of direct sums of anyons: An anyon should really be understood as a superselection sector, or a representation of the quasi-local algebra, rather than a quantum state.
\end{remark}

\begin{lemma}
Given a representation $\pi: A \to \mathcal{L}(\mathcal{H})$, if $V \subset \mathcal{H}$ is a closed invariant subspace, then its orthogonal complement $V^\perp$ is also an invariant subspace.
\end{lemma}

There is a canonical way to obtain a representation from a state, called the GNS construction:

\begin{definition}[GNS construction]
Given a \(C^*\)-algebra $A$ and a state $\phi: A \to \mathbb{C}$, we define a representation $\pi_\phi: A \to \mathcal{L}(\mathcal{H}_\phi)$ as follows:

Step 1.
Take $A$ as a vector space and assign the pre-inner product $\langle a,b \rangle = \phi(a^*b)$ on it.
Note that it is not positive definite; the null subspace, $N_\phi = \{ x \in A: \phi(x^* x) = 0 \}$, is nonzero in general.

Step 2.
Perform Cauchy completion on the quotient space $A/N_\phi$ to obtain the Hilbert space $\mathcal{H}_\phi$.
(It is often denoted as $L^2(A,\phi)$ in the literature.) $1 + N_\phi$ is called the distinguished cyclic vector in $\mathcal{H}_\phi$.

Step 3.
Let $\pi_\phi(a)$ be the left-multiplication operator on $A/N_\phi$, i.e., $\pi_\phi(a)\left(x+N_\phi\right)=a x+N_\phi$.
Note that $\pi_\phi(a)$ is a bounded operator, so it extends to a bounded operator on $\mathcal{H}_\phi$.
In other words, $\pi_\phi$ defines a representation $A \to \mathcal{L}(\mathcal{H}_\phi)$.
\end{definition}

We would use $\GNS(\phi)$ to denote the module obtained by the GNS construction. The physical meaning of $\GNS(\phi)$ is the superselection sector that the state $\phi$ lives in.

\begin{example}[GNS constructions in the Ising chain]
    Consider the infinite Ising chain with $Z_2$ symmetry imposed. Let $A$ be the algebra of symmetric quasi-local operators, which is generated by $X_i$ and $Z_i Z_j$.

    Take the ground state $\rho_1$ of the trivial phase, which satisfies $\rho_1(X_i) = 1, \rho_1(Z_i Z_j) = 0$.
    Note that $\rho_1((X_i - 1)^* (X_i - 1)) = 0$, thus $X_i - 1 \in N_{\rho_1}$. This means the vector $X_i$ is identified with the vector $1$ in $\mathcal{H}_{\rho_1}$, which means the distinguished cyclic vector is in the $+1$ eigenspace of $X_i$. On the other hand, $Z_i Z_j$ is identified with neither $0$ nor $1$.

    Note that every $a \in A$ is a linear combination of products of $X_i$ and $Z_i Z_j$, and we could always commute all $X$ to the rightmost by the anti-commutation relation.
    On the rightmost, $X_i$ is identified with $1$.
    Therefore, we have found a basis of $\mathcal{H}_{\rho_1}$: $\{ 1, Z_i Z_j, Z_i Z_j Z_k Z_l, \dots \}$.
    Intuitively, it is the Hilbert space spanned by vector states with even numbers of spin flips.
    
    Similarly, if we take the ground state $\rho_2$ of the symmetry breaking phase, which satisfies $\rho_2(X_i) = 0, \rho_2(Z_i Z_j) = 1$, we could find a basis of $\mathcal{H}_{\rho_2}$: $\{ 1, X_i, X_i X_j, \dots \}$.
    Intuitively, it is the Hilbert space spanned by vector states with even numbers of domain walls.

    Finally, we may also take states with defects. 
    For example, if we create a single spin flip at site $0$ on the trivial phase ground state and denote the resulting state as $\phi_1$, it can be shown $\mathcal{H}_{\phi_1}$ again has a basis $\{ 1, Z_i Z_j, Z_i Z_j Z_k Z_l, \dots \}$, but now $X_0$ is identified with $-1$ instead of $1$.
    Using techniques in Section~\ref{subsec: ReaBCond is an equivalence of categories}, we could show that $\mathcal{H}_{\rho_1}$ and $\mathcal{H}_{\phi_1}$ are non-isomorphic simple modules, which confirms the physical intuition that a single spin flip in the trivial phase is a simple topological defect.
\end{example}

\begin{theorem}[{\cite[II.6.4.8.]{Blackadar_OperatorAlg_2017}}]
If $\phi$ is a state of $A$, then $\pi_\phi$ is irreducible if and only if $\phi$ is a pure state.
\end{theorem}

We refer the reader to \cite[II.6.]{Blackadar_OperatorAlg_2017} for more details about states and representations of \(C^*\)-algebras.

\subsection{Hilbert modules and bimodules}
Representations are left modules of $C^*$ algebras. Right modules and bimodules are also important in the study of spin systems.

\begin{definition}
A \emph{right Hilbert $A$-module} is a vector space $X$ which is algebraically a right $A$-module, equipped with a sesquilinear map $\langle\cdot \mid \cdot\rangle: X \times X \rightarrow A$, satisfying the following conditions:
    \begin{enumerate}
        \item $\langle x \mid y\rangle^*=\langle y \mid x\rangle$;
        \item (Linearity) $\langle x \mid y a\rangle=\langle x \mid y\rangle a$;
        \item (Positivity) $\langle x \mid x\rangle \geq 0$, with equality if and only if $x=0$. Here $x \geq 0$ means $x$ is a positive operator;
        \item (Completeness) The norm $\|x\|:=\|\langle x \mid x\rangle\|^{\frac{1}{2}}$ is Cauchy complete.
    \end{enumerate}
\end{definition}

Morphisms of right Hilbert modules are called adjointable operators:
\begin{definition}
    Let $X$ and $Y$ be two Hilbert $A$-modules. An \emph{adjointable operator} from $X$ to $Y$ is an $A$-module intertwiner $T: X \rightarrow Y$ such that there is another $A$-module intertwiner $T^*: Y \rightarrow X$, called the adjoint of $T$, satisfying $\langle T(x) \mid y\rangle_Y=\left\langle x \mid T^*(y)\right\rangle_X$.
    The space of adjointable operators from $X$ to $Y$ is denoted $\mathcal{L}(X, Y)$.
\end{definition}
It is straightforward to see $\mathcal{L}(X, X)$ has a canonical structure of a unital $\mathrm{C}^*$-algebra. 

Finally, we have Hilbert bimodules, which are bimodules with appropriate structures compatible with $C^*$-algebras:
\begin{definition}
    A \emph{Hilbert $(A,B)$-bimodule} is a right Hilbert $B$-module $X$, together with a unital *-homomorphism $A \rightarrow \mathcal{L}(X, X)$.
\end{definition}

An important notion of Hilbert modules is projective basis:
\begin{definition}
Let $X$ be a right Hilbert $A$-module. A \emph{projective basis} is a finite subset $\left\{s_i\right\}_{i=1}^n \subseteq X$ such that
\[
    \sum_i s_i\left\langle s_i \mid x \right\rangle = x \quad \forall x \in X.
\]
A Hilbert (bi-)module is said to be \emph{right finite} if there exists a projective basis. 
\end{definition}
Note that a right Hilbert module has a projective basis if and only if it is projective and finitely generated.

\subsection{Toolkit for AF algebras}
In this section, we will introduce the language of \emph{approximately finite-dimensional \(C^*\)-algebras} (or \emph{AF algebras} for brevity) and basic tools for dealing with them.
This language is indispensable in describing physical systems at the thermodynamic limit, as the Hilbert space language becomes unusable in this setting.

Though the definition looks fancy, the physical intuition behind the definition is clear: As it doesn't make sense to define a \emph{global} Hilbert space of a physical system at the thermodynamic limit, we study information encoded in arbitrarily large \emph{finite} regions.

\begin{definition}
An \emph{inductive system} of \(C^*\)-algebras is a collection of data $\{\left(A_i, \phi_{i j}\right): i, j \in I, i \leq j\}$, where $I$ is a directed set, the $A_i$ are $\mathrm{C}^*$-algebras, and $\phi_{i j}: A_i \to A_j$ is a $*$-homomorphism, satisfying the consistency relation $\phi_{i k}=\phi_{j k} \circ \phi_{i j}$ for $i \leq j \leq k$.
The \emph{inductive limit} $\underset{\longrightarrow}{\lim }\left(A_i, \phi_{i j}\right)$ is the norm completion of the algebraic inductive limit (with zero-norm elements quotiented out), where the \(C^*\)-norm is defined by
\begin{equation}
\|a\|=\lim _{j>i}\left\|\phi_{i j}(a)\right\|=\inf _{j>i}\left\|\phi_{i j}(a)\right\|.
\end{equation}
\end{definition}

\begin{remark}
When all $\phi_{ij}$ are inclusion maps $\underset{\longrightarrow}{\lim }\left(A_i, \phi_{i j}\right)$ is the norm completion of the union $\bigcup_i A_i$; this is the case for most physical systems of interest.
\end{remark}

\begin{definition}[AF algebra]
An \emph{approximately finite-dimensional $C^*$-algebra}, or \emph{AF algebra}, is an inductive limit of finite-dimensional \(C^*\)-algebras.
\end{definition}

\begin{example}[Quasi-local algebra]
Any infinite spin system, regardless of its dimension, naturally gives rise to an AF algebra as follows.
We take $I$ to be the set of all finite regions in the system, with $X \leq Y$ if $X$ is a subregion of $Y$.
For any finite region $X$, define $A_X$ to be the algebra of observables supported on the region $X$.
$\phi_{X Y}: A_X \to A_Y$ is defined as tensoring identity operators on all sites in $Y \setminus X$.
The \emph{quasi-local algebra} is defined as the inductive limit $\underset{\longrightarrow}{\lim }\left(A_X, \phi_{X Y}\right)$.
It is the foundation of many works on the operator algebraic aspects of topological orders; see the introduction and references therein.
\end{example}

To supply some physical intuition, we show that a state defined on the quasi-local algebra is precisely a collection of compatible reduced density matrices defined on all finite regions:

\begin{lemma}
Given an AF algebra $A = \underset{\longrightarrow}{\lim }\left(A_i, \phi_{i j}\right)$, a state $\rho: A \to \mathbb{C}$ determines a set of states $\{\rho_i: A_i \to \mathbb{C}\}_{i \in I}$, which satisfies the compatibility condition $\rho_i = \rho_j \circ \phi_{ij} \quad \forall i \leq j$.
Conversely, any set of states $\{\rho_i: A_i \to \mathbb{C}\}_{i \in I}$ satisfying the compatibility condition uniquely determines a state $\rho: A \to \mathbb{C}$.
\end{lemma}

\begin{proof}
For the first part, note that there is a canonical $*$-homomorphism $\iota_i: A_i \to A$ in the definition of the inductive limit, therefore $\rho \circ \iota_i: A_i \to \mathbb{C}$ is a state on $A_i$; obviously this gives a set of states $\{\rho_i: A_i \to \mathbb{C}\}_{i \in I}$ which satisfies the compatibility condition.
For the second part, note that $\{\rho_i: A_i \to \mathbb{C}\}_{i \in I}$ naturally induces a bounded linear functional on the algebraic inductive limit.
On the other hand, the algebraic inductive limit (with zero-norm elements quotiented out) is dense in $A$, thus any linear bounded functional on the algebraic inductive limit extends uniquely to $A$.
It is straightforward to verify that this extension defines a state on $A$.
\end{proof}

\paragraph{Descending to finite pieces} An important strategy for dealing with AF algebras is \emph{descending to finite pieces}.
Namely, any state (resp. representation, Hilbert bimodule, etc.) on an AF algebra $A = \underset{\longrightarrow}{\lim} \left(A_i, \phi_{i j}\right)$ could be understood as the result of \emph{gluing} a collection of states (resp. representations, Hilbert bimodules, etc.) defined on the finite pieces $A_i$.
The problem is often quite easy on $A_i$, which are multi-matrix algebras.

Here is one example of the strategy: Given any state $\rho: A \to \mathbb{C}$, we would expect $\operatorname{GNS} (\rho)$ could also be obtained by gluing finite pieces $\operatorname{GNS} (\rho_i)$ together.
This is indeed the case:

\begin{lemma}
Consider an AF algebra $A = \underset{\longrightarrow}{\lim} \left(A_i, \phi_{i j}\right)$, a state $\rho: A \to \mathbb{C}$, and the set of states $\{\rho_i: A_i \to \mathbb{C}\}_{i \in I}$ defined by $\rho_i = \rho \circ \iota_i$, as in the previous lemma.
Then each $\phi_{ij}: A_i \to A_j$ induces an isometry $\tilde{\phi}_{ij}: \mathcal{H}_{\rho_i} \to \mathcal{H}_{\rho_j}$.
Moreover, $\tilde{\phi}_{ij}$ is compatible with ${\phi}_{ij}$, expressed by the following commutative diagram:
\[\input{diagrams/gns-isometry-0479.tex}\]

\end{lemma}

\begin{proof}
First, as $\rho_i(a) = \rho_j(\phi_{ij}(a))$, $\tilde{\phi}_{ij}: A_i/N_i \to A_j/N_j$ is a well-defined isometry; in particular it is a bounded map.
As $A_i/N_i$ is dense in $\mathcal{H}_{\rho_i}$, $A_j/N_j$ is dense in $\mathcal{H}_{\rho_j}$, $\tilde{\phi}_{ij}$ could be extended to $\mathcal{H}_{\rho_i} \to \mathcal{H}_{\rho_j}$.
Compatibility of $\tilde{\phi}_{ij}$ with ${\phi}_{ij}$ is straightforward to verify.
\end{proof}

\begin{theorem}[GNS construction and inductive limit commutes]
We have $\mathcal{H}_\rho = \underset{\longrightarrow}{\lim} (\mathcal{H}_{\rho_i}, \tilde{\phi}_{ij})$.
Moreover, the inductive limit of Hilbert spaces is compatible with the inductive limit of algebras.
\end{theorem}

\begin{proof}
For the first part, our strategy is to prove that both $\underset{\longrightarrow}{\lim} (\mathcal{H}_{\rho_i}, \tilde{\phi}_{ij})$ and $\mathcal{H}_\rho$ are the norm completion of $\bigcup_i (A_i/N_i)$.
On the one hand, since each $\tilde{\phi}_{ij}: \mathcal{H}_{\rho_i} \to \mathcal{H}_{\rho_j}$ is an isometry (thus an embedding), $\underset{\longrightarrow}{\lim} (\mathcal{H}_{\rho_i}, \tilde{\phi}_{ij})$ is the completion of $\bigcup_i \mathcal{H}_{\rho_i}$.
It implies $\underset{\longrightarrow}{\lim} (\mathcal{H}_{\rho_i}, \tilde{\phi}_{ij})$ is the completion of $\bigcup_i (A_i/N_i)$, as $A_i/N_i$ is dense in $\mathcal{H}_{\rho_i}$.
On the other hand, note that there is an isometry $\tilde{\iota}_i: A_i/N_i \to A/N$ defined as $a_i + N_i \mapsto \iota(a_i) + N$.
It follows $\bigcup_i (A_i/N_i)$ is dense in $A/N$, as $A$ is the inductive limit of $\{A_i\}$.
Since $\mathcal{H}_\rho$ is the completion of $A/N$, density of $\bigcup_i (A_i/N_i)$ in $A/N$ implies $\mathcal{H}_\rho$ is the completion of $\bigcup_i (A_i/N_i)$.
This finishes the proof of the first part.
The second part follows the fact that $\tilde{\phi}_{ij}$ is compatible with ${\phi}_{ij}$.
\end{proof}

Also, tensor products of Hilbert bimodules between AF algebras could be obtained as the inductive limit of tensor products of finite pieces:

\begin{theorem}\label{thm:tensor-product-inductive-limit}
Let $A = \varinjlim A_i$ be an AF algebra, and let $M$, $N$ be Hilbert $A$-bimodules obtained as inductive limits $M = \varinjlim M_i$, $N = \varinjlim N_i$ of Hilbert $A_i$-bimodules.
Then the relative tensor product $M \otimes_A N$ is canonically isomorphic to $\varinjlim (M_i \otimes_{A_i} N_i)$.
\end{theorem}

The proof is tedious but straightforward. This technique would be essential for establishing the bulk-boundary correspondence in Section~\ref{section: bulk-boundary correspondence}.

\subsection{Interactions and ground states}
We explain how to define an interaction for an infinite system.
On an infinite system, we cannot define the Hamiltonian as an ordinary operator, as the sum of infinitely many local terms is not defined.
Therefore, we adopt the following definition:

\begin{definition}
Given a net of local algebras $\Loc_\bullet$, an \emph{interaction} is a map $\Phi$, sending a finite interval $I \subset \mathbb{Z}$ to a self-adjoint operator $\Phi(I) \in \Loc_I$.
We define the local Hamiltonian associated to $\Phi$ on $I$ to be $H_\Phi(I) := \sum_{J \subset I} \Phi(J)$.
\end{definition}

We would use \emph{interaction} and \emph{Hamiltonian} interchangeably, but the reader should always remember there is no honest Hamiltonian on infinite systems. In this paper we would only deal with commuting-projector interactions, which are defined as follows:

\begin{definition}
A \emph{commuting-projector interaction} is an interaction $\Phi$ such that $\Phi(I)^2 = \Phi(I)$ for all $I$, and $\Phi(I) \Phi(J) = \Phi(J) \Phi(I)$ for any $I,J$.
\end{definition}

\begin{definition}
A \emph{ground state} of a commuting-projector interaction $\Phi$ is a state $\rho \colon \Loc \to \mathbb{C}$, satisfying the condition that for any interval $I$, $\rho(H_\Phi(I)) = 0$.
\end{definition}

\begin{example}
The commuting-projector interaction for the trivial phase of the Ising chain is defined by $\Phi([i,i]) = 1-X_i$ and $\Phi([i,j]) = 0$ if $j-i > 0$; the ground state $\rho$ is the unique state satisfying $\rho(X_n) = 1$ for all $n \in \mathbb{Z}$.
On the other hand, the commuting-projector interaction for the symmetry-breaking phase of the Ising chain is defined by $\Phi([i,i+1]) = 1- Z_i Z_{i+1}$ and $\Phi([i,j]) = 0$ if $j-i \neq 1$; the ground state $\rho$ is the unique state satisfying $\rho(Z_n Z_{n+1}) = 1$ for all $n \in \mathbb{Z}$.
\end{example}

We note that this definition of ground state, which is dubbed \emph{H-type ground state} in \cite{MYLG25}, could not be generalized to arbitrary interactions; the usual definition of ground state in the operator algebraic theory of quantum spin systems, dubbed \emph{O-type ground state} in \cite{MYLG25}, is the state with the lowest energy in a simple superselection sector.

We note that the ground state of a commuting-projector interaction is indeed an eigenstate of all local projectors, in the following sense:

\begin{proposition}\label{prop:ground_state_projector_eigenstate}
The ground state $\rho$ of a commuting-projector interaction $\Phi$ satisfies
\[
    \rho(a \Phi(I)) = \rho(\Phi(I) a) = 0
\]
for any interval $I$ and any operator $a \in \Loc$.
\end{proposition}
\begin{proof}
The Cauchy-Schwarz inequality for states says that $|\phi(x^*y)|^2 \leq \phi(x^*x) \phi(y^*y)$ for any state $\phi$ and any $x, y$ in the algebra \cite[II.6.2.4]{Blackadar_OperatorAlg_2017}.
Let $x=a^*$ and $y=\Phi(I)$.
Then $|\rho(a \Phi(I))|^2 \leq \rho(aa^*) \rho(\Phi(I)^2)$.
Since $\Phi(I)$ is a projector, $\Phi(I)^2 = \Phi(I)$, and since $\rho$ is a ground state, $\rho(\Phi(I))=0$.
Thus, $|\rho(a \Phi(I))|^2 \leq \rho(aa^*) \cdot 0 = 0$, which implies $\rho(a \Phi(I)) = 0$.
The other equality follows by letting $x=\Phi(I)$ and $y=a^*$.
\end{proof}

\section{Fusion spin chains with boundary}\label{section: fusion spin chains with bdy}
\subsection{From group symmetry to fusion spin chains}
\label{subsection: group sym to fusion spin chain}

In this section, we briefly review the operator algebraic description of spin chains with on-site group symmetry found in the literature \cite{Kapustin_Sopenko_2021_SPTClassification, Ogata_2021_ClassificationSPT}.
We then explain how fusion spin chains not only generalize this description, but also provide a framework that is arguably more intrinsic and natural.

In the operator algebraic framework, a one-dimensional quantum lattice system without global symmetry is defined by its net of local operator algebras $\Loc_\bullet$.
In the bosonic case, this is a net of \(C^*\)-algebras.
Explicitly, we consider a finite-dimensional Hilbert space $V$ describing the degrees of freedom on a single site (e.g., $V = \mathbb{C}^d$ for a spin-$s$ chain where $d=2s+1$).
The algebra of observables on the finite interval $[-n, n]$ is given by
\begin{equation}
  \Loc_{[-n,n]} = \bigotimes_{j \in [-n, n]} \operatorname{End}(V)_j \cong \mathbb{M}_{d^{2n+1}}(\mathbb{C}).
\end{equation}
The quasi-local algebra is the norm completion of the union of these finite pieces, $\Loc = \overline{\bigcup_n \Loc_{[-n,n]}}$.

An on-site symmetry given by a finite group $G$ is defined as a homomorphism from $G$ to the group of automorphisms of the quasi-local algebra.
Let $u: G \to U(V)$ be a unitary representation of $G$ on the single-site Hilbert space $V$.
The global symmetry action is defined by conjugation: $g \cdot a := \operatorname{Ad}_{U(g)}(a) = U(g) a U(g)^{-1}$, where formally $U(g)$ is the infinite tensor product of on-site representations,
\begin{equation}
  U(g) = \bigotimes_{j \in \mathbb{Z}} u_j(g), \quad g \in G .
\end{equation}
While the infinite product $U(g)$ is not an element of the algebra $\Loc$, the automorphism $\operatorname{Ad}_{U(g)}$ is well-defined on local operators.

A special class of interesting operators is the subalgebra of symmetric operators.
An operator $a \in \Loc$ is \emph{symmetric} if it is invariant under the group action, i.e., $g \cdot a = a$ for all $g \in G$.
This condition is equivalent to saying that $a$ commutes with the representation of the symmetry group on the relevant Hilbert space: $[a, U(g)] = 0$.
In the language of representation theory, symmetric operators are precisely the \emph{intertwiners} of the symmetry representations.

This observation provides the bridge to fusion spin chains.
Recall that the category of representations of a finite group, $\Rep(G)$, is a unitary fusion category.
The single-site Hilbert space $V$ equipped with the action $u$ is simply an object $x \in \Rep(G)$.
The Hilbert space on an interval of length $L$ corresponds to the tensor product object $x^{\otimes L}$.
Consequently, the algebra of symmetric operators supported on this interval is exactly the space of endomorphisms of this object in the category $\Rep(G)$:
\begin{equation}
    \Loc^{\text{sym}}_{[a, b]} \cong \operatorname{End}_{\Rep(G)}\left(x^{\otimes (b-a+1)}\right) \equiv \C(x^{\otimes (b-a+1)}, x^{\otimes (b-a+1)}).
\end{equation}
This is precisely the definition of the local algebras $\Loc_{[1,L]}$ for a fusion spin chain constructed from the category $\C = \Rep(G)$ and the object $x$, see the next subsection for details.

The discussion above encapsulates the key intuition behind our formalism: \emph{elements in the quasi-local algebra of a fusion spin chain are intertwiners of symmetry representations.}
Once this identification is made, it is natural to generalize the input category $\C$ from $\Rep(G)$ to a general unitary fusion category.
This allows us to describe systems with generalized symmetries on the same footing as group symmetries.

\begin{remark}
A crucial feature of the fusion spin chain is that it does not know how the symmetry acts on the underlying vector space.
In categorical terms, it does not depend on the choice of a fiber functor $F: \C \to \ve$.
This is physically significant because the fiber functor represents the "charged" Hilbert space, which is often unobservable in the context of topological phases where we only have access to symmetric observables.

To illustrate this, consider the phenomenon of \emph{isocategorical groups} \cite{Etingof_2000_IsocategoricalGroups}.
There exist distinct finite groups $G$ and $H$ whose representation categories are equivalent as tensor categories, $\Rep(G) \simeq \Rep(H)$, but whose underlying fiber functors are different.
If we construct fusion spin chains based on these two groups, the resulting algebras of symmetric observables are completely identical.

One might ask: does the fusion spin chain formalism miss physical information by ignoring the symmetry action?
We advocate the view that the traditional Hilbert space description contains redundancy, at least for the purpose of classifying and studying topological phases.
Recall that the equivalence relation of symmetry-protected and symmetry-enriched topological phases is defined by symmetric local unitary transformations \cite{Chen_Wen_2010_LocalUnitary}, which only concerns symmetric operators. In other words, we do not need to know the symmetry action to classify topological phases! 
In this aspect, fusion spin chains capture exactly the intrinsic data of phases defined on it.
\end{remark}

\subsection{Fusion spin chains}
In this section, we first give the general operator algebraic description of spin chains, then specialize to spin chains constructed from unitary fusion categories.
Physically, the unitary fusion category encodes local quantum charges, or charges (representations) of the global symmetry.
However, this approach immediately generalizes group symmetries or even Hopf-algebra symmetries, since we may use unitary fusion categories without fiber functors; it is known that a fusion category is equivalent to the representation category of some Hopf algebra if and only if a fiber functor exists \cite[Section 5.3]{EGNO_2015}.

Let $\mathcal{I}$ denote the category defined as follows: The objects are discrete intervals in $\mathbb{Z}$, and the morphisms are inclusions.
For $F \subset \mathbb{Z}$, we will denote its complement by $F^c$.

\begin{definition}[{\cite[Definition 2.1]{Hataishi_2025_DHRAbstractChain}}]\label{def:abstract-spin-chain}
An \emph{abstract spin chain without boundary} is a functor $\Loc_{\bullet} \colon \mathcal{I} \to C^*\text{-alg}$, satisfying the following conditions:
\begin{enumerate}
    \item Local finiteness: $\Loc_I$ is a finite-dimensional \(C^*\)-algebra for any $I \in \mathcal{I}$;
    \item Locality: If $I \cap J = \varnothing$, then $\Loc_I$ and $\Loc_J$ commute.
\end{enumerate}
The \emph{quasi-local algebra} associated to the abstract spin chain is $\Loc := \varinjlim_{I\in\mathcal I}\Loc_I$, which is a unital AF \(C^*\)-algebra.
$\Loc$ may be understood to be the norm completion of $\cup_I \Loc_I$.
Physically, $\Loc_I$ is the algebra of operators supported in $I$.
\end{definition}

\begin{definition}[{Fusion spin chain \cite{Jones_2024_DHR}}]\label{def:fusion-spin-chain}
    Given a unitary fusion category $\C$, we define an abstract spin chain, called the \emph{fusion spin chain}, as follows: 
    First, we need to pick an object $x \in \C$.
    $x$ should be a \emph{strong tensor generator}, in the sense that there is some $n$ such that $x^n$ contains all simple objects of $\C$ as summands.
    For all $i \in \mathbb{Z}$, we put $x$ on site $i$.
    Local operators supported in some $I \in \mathcal{I}$ are just endomorphisms of the tensor product of all objects in $I$; that is, we set $\Loc_{[i,j]} := \C(x^{j-i+1},x^{j-i+1})$.
    Note that $\Loc_{[i,j]}$ is a multi-matrix $^*$-algebra, so it is a \(C^*\)-algebra with the \(C^*\)-norm given by the operator norm.

    We still need to define inclusions of algebras induced by inclusions of regions.
    Obviously, the inclusions are defined by tensoring with identity morphisms; that is, for $I=[i,j]$ and $J=[k, l]$ with $I \subseteq J$ ($k \leq i$ and $j \leq l$), we define the algebra homomorphism $i_{I,J} \colon \Loc_I \to \Loc_J$ by $i_{I,J}(f) = 1_{x^{i-k}} \otimes f \otimes 1_{x^{l-j}}$.
    It is obvious $i_{I,J}$ respects both the involution and the \(C^*\)-norm.
\end{definition}

\subsection{Boundary charge category and fusion spin chain with boundary}
\label{BFusionSpinChain}
In this paper we would construct boundary Hamiltonians for 1D phases.
In order to do this, we first introduce how to define degrees of freedom for 1D systems with boundary.
The following definition is analogous to the fusion spin chain; the only difference is that we have a boundary and use a module category to encode boundary quantum charges.

Let $\mathcal{I}^+$ denote the category defined as follows: The objects are discrete intervals in $\mathbb{Z}_{\geq 0}$, and the morphisms are inclusions.
For $F \subset \mathbb{Z}_{\geq 0}$, we will denote its complement by $F^c$.

\begin{definition}\label{def:abstract-spin-chain-with-boundary}
    An \emph{abstract spin chain with boundary} is a functor $\Loc_{\bullet} \colon \mathcal{I}^+ \to C^*\text{-alg}$, satisfying the following conditions:
    \begin{enumerate}
        \item Local finiteness: $\Loc_I$ is a finite-dimensional \(C^*\)-algebra for any $I \in \mathcal{I}^+$;
        \item Locality: If $I \cap J = \varnothing$, then $\Loc_I$ and $\Loc_J$ commute.
    \end{enumerate}
    The \emph{quasi-local algebra} associated to the abstract spin chain is $\Loc := \varinjlim_{\mathcal{I}^+} \Loc_I$, which is a unital AF \(C^*\)-algebra.
    $\Loc$ may be understood to be the norm completion of $\cup_I \Loc_I$.
\end{definition}

Fusion spin chains with boundary are abstract chains constructed from a fusion category $\C$ and its module category\footnote{Note that by default, module category means \emph{right} module category in this paper. This convention is different from many references.} $\M$.
Fusion spin chains are of particular physical interest, since $\C$ encodes data about the bulk symmetry charge and $\M$ encodes data about the boundary local quantum charges available before a bulk phase is chosen.
After choosing a Q-system $Q\in\C$, the corresponding categorical boundary conditions are right $Q$-modules in $\M$, namely objects of $\MQ$; their microscopic realization is discussed in Section~\ref{Def. The realization functor for boundary conditions}.

\begin{definition}\label{def:fusion-spin-chain-with-boundary}
    Given a unitary fusion category $\C$ and an indecomposable semisimple unitary module category $\M$, we define an abstract spin chain with boundary, called the \emph{fusion spin chain with boundary}, as follows: 
    First, pick a strong tensor generator $x \in \C$ and any nonzero $l \in \M$.
    Put $x$ on site $i$ for all $i > 0$, and put $l$ on site 0. 

    Local operators supported in $I \in \mathcal{I}^+$ are just endomorphisms of the tensor product of all objects in $I$; that is, we set $\Loc_{[i,j]} := \C(x^{j-i+1},x^{j-i+1})$ when $i>0$, and $\Loc_{[0,j]} := \M(l \otimes x^j, l \otimes x^j)$.
    Note that $\Loc_{[i,j]}$ is a multi-matrix $^*$-algebra, so it is a \(C^*\)-algebra with the \(C^*\)-norm given by the operator norm.
    To define the algebra on the half-infinite chain, we need to define inclusions of algebras when there is inclusion between two regions.
    The inclusion \(i_{I,J}\) is similarly defined by tensoring with identity morphisms in \(J \setminus I\).
    It is obvious $i_{I,J}$ respects both the involution and the \(C^*\)-norm.
\end{definition}

Note that when $\M = \C$, $\Loc$ is just the half-line algebra $\Loc_+ = \Loc_{[0, \infty)}$ of a fusion spin chain \emph{without} boundary constructed from $\C$ \cite{Jones_2024_DHR,Hataishi_2025_DHRAbstractChain}.

\begin{remark}
    Why should boundaries be described by module categories? 
    The answer is we must have some \emph{coarse-graining} or \emph{renormalization} operation, which groups multiple sites into one and obtain a coarser lattice. 
    In the bulk, this means local quantum charges must be able to fuse, which implies the category of bulk local quantum charges $\C$ has a monoidal structure.
    At the boundary, existence of coarse-graining means the boundary must be able to absorb bulk local quantum charges associatively. This implies the category of boundary local quantum charges is a module category of $\C$.
\end{remark}

\begin{example}
When $\C = \mathrm{Rep}(G)$ for some finite group $G$, indecomposable module categories of $\C$ are all of the form $\mathrm{Rep}_\psi(H)$, the linear category of representations of the twisted group algebra $\mathbb{C}^\psi[L]$ (or the category of projective representations of $L$ with the 2-cocycle $\psi$), where $L$ is a subgroup of $G$ and $\psi \in H^2(L, U(1))$ \cite[Corollary 7.12.20]{EGNO_2015}.
In particular, $\M = \mathrm{Rep}(G)$ is the usual fully symmetric boundary, and $\M = \mathrm{Rep}(\{ e \}) = \ve$ is the fully symmetry breaking boundary.
\end{example}

\paragraph{Remark.} One useful viewpoint is to regard $\M$ as ${}_W \C$, where $W$ is an algebra internal to $\C$ and ${}_W \C$ is the category of left $W$-modules internal to $\C$, equipped with the right $\C$-module structure by tensoring on the right; this can be done by, for example, taking $W = [l,l]$, where $[-,-] \colon \M^{\op} \times \M \to \C$ is the internal hom functor \cite[Theorem 7.10.1]{EGNO_2015}.

\subsection{Toolkit for fusion spin chains}
Here we review some definitions, tools and results useful for dealing with fusion spin chains.
We mainly follow \cite{Chen_2024_InductiveLimitAF, Chen_2024_QSysCompletion, Jones_2024_DHR}, with some natural generalizations.
Before diving into technical details, we emphasize that the most important feature about the algebra $\Loc$ is that it is an \emph{AF-algebra}, which means the colimit of finite-dimensional $\text{C}^*$-algebras.
Therefore, although $\Loc$ is infinite-dimensional, we could study it by working on its finite-dimensional building blocks.

% everything in this subsection follows the spirit.

\subsubsection{Actions of fusion categories on \(C^*\)-algebras}\label{AFAction}
\begin{definition}
If $\C$ is a unitary fusion category and $A$ is a (unital) $\text{C}^*$-algebra, an \emph{action} of $\C$ on $A$ is a $\text{C}^*$-tensor functor $F \colon \C \to \mathrm{Bim}(A)$.
That is, we specify for every object $x \in \C$ a bimodule $F(x) \in \mathrm{Bim}(A)$, for every morphism $f \in \C(x,y)$ a bimodule intertwiner $F(f) \colon F(x) \to F(y)$, and for every pair of objects $x,y \in \C$ a unitary isomorphism $F^{2}_{x,y} \colon F(x) \boxtimes_{A} F(y) \to F(x \otimes y)$ satisfying the pentagon equation.
\end{definition}

\begin{definition}[AF-action \cite{Chen_2024_InductiveLimitAF}] 

Let $A,B$ be unital C*-algebras and $\phi:A\rightarrow B$ a unital $*$-homomorphism.
Let $\C$ be a unitary fusion category and suppose we have actions $F:\C\rightarrow \textbf{Bim}(A)$ and $G:\C\rightarrow \textbf{Bim}(B)$.
An equivariant structure on $\phi$ with respect to $F$ and $G$ is a family of linear maps $\{k^{x}: F(x)\rightarrow G(x):\ x\in \C\}$ satisfying the following conditions

\begin{enumerate}
\item 
For $a,b\in A$, and $\xi\in F(x)$, 

$k^{x}(a\triangleright \xi\triangleleft b)=\phi(a)\triangleright k^{x}(\xi)\triangleleft \phi(b).$
\item 
For $f\in \C(x,y)$, $\xi\in F(x)$
$k^{y}\circ F(f)(\xi)=G(f)\circ k^{x}(\xi).$
\item
$\langle k^{x}(\xi)\ |\ k^{x}(\eta)\rangle_{B}=\phi(\langle \xi\ |\ \eta\rangle_{A}).$
\item
$G(x)=k^{x}(F(x))\triangleleft B$.
\item 
The following diagram commutes:

\[
\begin{tikzcd}[column sep=large]
F(x) \otimes F(y)
\arrow[d, swap, "{k^{x} \otimes k^{y}}"]
\arrow[r]
& F(x) \boxtimes_A F(y)
\arrow[r, "{F^2_{x,y}}"]
& F(x \otimes y)
\arrow[d, "{k^{x \otimes y}}"]
\\
G(x) \otimes G(y)
\arrow[r]
& G(x) \boxtimes_B G(y)
\arrow[r, "{G^2_{x,y}}"]
& G(x \otimes y)
\end{tikzcd}
\]

\end{enumerate}

Let $\C$ be a (strict) unitary fusion category.
Suppose that we have:

\begin{enumerate}
    \item 
A sequence of finite dimensional C*-algebras $A_{n}$ and unital, injective $*$-inclusions $\iota_{n}: A_{n}\rightarrow A_{n+1}$.
\item 
A sequence of actions $F_{n}:\C\rightarrow \textbf{Bim}(A_{n})$.
\item 
A family $\{k^{x}_{n}\}$ of equivariant structures on $\iota_{n}$ with respect to $F_{n}$ and $F_{n+1}$.
\end{enumerate}

Then if we let $A:=\varinjlim A_{n}$ be the inductive limit of the sequence $A_{n}$ in the category of C*-algebras, there exists a canonical action $F:\C\rightarrow \textbf{Bim}(A)$ called the inductive limit action of the $F_{n}$ (for a detailed construction, see \cite[Proposition 4.4]{Chen_2024_InductiveLimitAF}).
Any action of $\C$ on an AF C*-algebra equivalent (in the sense of \cite[Definition 3.9]{Chen_2024_InductiveLimitAF}) to one constructed as above is called an \textit{AF-action}.

\end{definition}

Actions of fusion categories on $\text{C}^*$-algebras serve as a bridge between category theory and operator algebra.
We will mostly use this bridge to translate categorical constructions to $\text{C}^*$-algebraic objects.

\subsubsection{The standard action}\label{StandardAction}

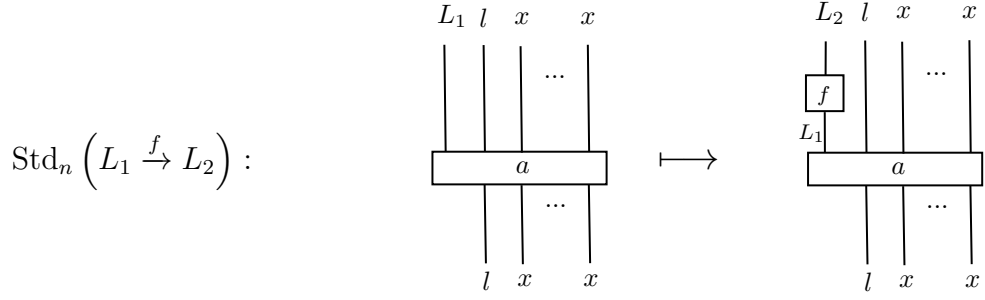
\begin{figure}
    \centering

\input{diagrams/standard-action-0788.tex}
    \caption{Definition of \(\Std\), which will be called \(\ReaDHR\) later.}
    \label{fig:Std and ReaDHR}
\end{figure}

Here we introduce the most important usage of our bridge.
Given a unitary fusion category $\C$, an indecomposable semisimple unitary module category $\M$, we construct an AF-action, called the \emph{standard action}, as follows.

Fix a self-dual object $x \in \C$ containing every isomorphism class of simple object in $\C$, and an object $l \in \M$ containing every isomorphism class of simple object in $\M$.
We consider the category of right module functors:
\begin{equation}\label{eq:morita-dual}
    \C^\vee_{\M} := \mathrm{Fun}_{\C^\rev}(\M,\M)
\end{equation}
We define an action $\Std \colon (\C^\vee_{\M})^{\rev} \to \mathrm{Bim}(\Loc^{\bdy})$, where $\Loc^{\bdy}$ is precisely the quasi-local algebra of the fusion spin chain with boundary defined above.

The typing convention is as follows. 
An object $L\in \C_\M^\vee$ is a right-module endofunctor of $\M$, but in string diagrams we draw its action as a new strand on the \emph{left} of the boundary object. 
Thus a vector in the finite piece $\Std_n(L)$ is a morphism from the physical boundary Hilbert space $l\otimes x^n$ to the same boundary Hilbert space with an extra $L$-strand on the left. 
Composition of two such vectors stacks these extra strands, and the outer strand is produced by the second tensor factor. 
This is why the source of the action is the monoidal reverse $(\C_\M^\vee)^{\rev}$.

We denote $\Loc^{\bdy}_{n}:=\M(l \otimes x^{n}, l \otimes x^{n})$, which is a finite dimensional C*-algebra.
The inclusions $\iota_{n} \colon \Loc^{\bdy}_{n} \to \Loc^{\bdy}_{n+1}$ are given by
\[
a\mapsto a \otimes 1_x.
\]
First, we define the finite pieces of the AF action, $\Std_n \colon (\C^\vee_{\M})^{\rev} \to \mathrm{Bim}(\Loc^{\bdy}_n)$.
On object level, for $L \in \C_\M^\vee$, we define
\begin{equation}\label{eq:std-finite-piece}
    \Std_{n}(L):=\M(l \otimes x^{n}, L \otimes l \otimes x^{n})
\end{equation}
with the $\Loc^{\bdy}_{n}$-bimodule structure given by
\[
    a\triangleright \xi \triangleleft b:= (1_L \otimes a) \circ \xi \circ b, \quad a,b \in \Loc^{\bdy}_{n}, \xi\in \Std_{n}(L)
\]
Throughout the fusion-spin-chain construction, we denote the action of the module functor $L$ on objects and morphisms as tensoring on the left, i.e. by $L \otimes l \otimes x^n$ and $1_L \otimes a$ for $a \in \Loc^{\bdy}$, respectively.
We employ this notation to emphasize intuitions from the graphical calculus. 
Indeed, it is the graphical intuition that the authors keep in mind while writing this paper; we find that this intuition is much more convenient than checking lengthy algebraic expressions.

The bimodule structure is graphically represented as
\[
\input{diagrams/bimodule-action-0917.tex}
\]

The $\Loc^{\bdy}_{n}$-valued sesquilinear form is $\langle \xi\ |\ \eta\rangle := \xi^{*}\circ \eta$, which is graphically represented as

\[
\input{diagrams/inner-product-0978.tex}
\]

On morphism level, we define $\Std_{n}(f)(\xi):=(f \otimes 1_{l \otimes x^{n}}) \circ \xi$ for $f \in \C_\M^\vee(L_1, L_2)$, which is obviously a bimodule intertwiner.

We define the monoidal structure of the functor out of $(\C_\M^\vee)^{\rev}$ by giving its action on elementary tensors,
\[
(\Std_{n})^{2}_{L_1, L_2}: \Std_{n}(L_1) \boxtimes_{\Loc^{\bdy}_n} \Std_{n}(L_2) \to \Std_{n}(L_2 \otimes L_1), \qquad
(\Std_{n})^{2}_{L_1, L_2}(\xi\otimes \eta):=(1_{L_2} \otimes \xi) \circ \eta.
\]
Indeed, if $\xi: l\otimes x^n\to L_1\otimes l\otimes x^n$ and $\eta: l\otimes x^n\to L_2\otimes l\otimes x^n$, then $(1_{L_2}\otimes \xi)\circ \eta$ is a morphism $l\otimes x^n\to L_2\otimes L_1\otimes l\otimes x^n$, hence an element of $\Std_n(L_2\otimes L_1)$. This reversed order is exactly the monoidal product in $(\C_\M^\vee)^{\rev}$.
It is easy to check that these extend to unitary natural isomorphisms satisfying the pentagon equation.

It remains to define the equivariant structure.
We define $k^{L}_{n} \colon \Std_{n}(L) \to \Std_{n+1}(L)$ by
\[
k^{L}_{n}(\xi):=\xi \otimes 1_{x}.
\]
It is straightforward to verify that this defines an equivariant structure on $\iota_{n} \text{ with respect to } \Std_{n} \text{ and } \Std_{n+1}$.

Taking the limit, we obtain an AF-action $\Std \colon (\C^\vee_{\M})^{\rev}\to \mathrm{Bim}(\Loc^{\bdy})$ which we call a \emph{standard AF-action}.

Note that the standard action in \cite[Example 4.6]{Jones_2024_DHR} is just the special case $\M = \C$, where $\C_{\M}^\vee \simeq \C$.
This convention differs from \cite{Jones_2024_DHR}: we put the module functor on the left, so the order of the tensor product is reversed. This is the monoidal reverse convention summarized in Appendix~\ref{app:orientation:dhr}.

\subsubsection{The realization construction}\label{RealizationConstruction}
We further demonstrate how to use our bridge to translate algebras (more precisely, Q-systems; see appendix for details) internal to $\C$ to finite-index extensions of the $\text{C}^*$-algebra $A$.

\begin{theorem}[{\cite[Construction 4.3]{Chen_2024_QSysCompletion}}]
Given a fully faithful action $F \colon \C \to \mathrm{Bim}(A)$ and a Q-system $Q$ in $\C$, $F(Q)$ is canonically a unital $\text{C}^*$-algebra; we denote the $\text{C}^*$-algebra as $|Q|$.
Furthermore, $|Q|$ is equipped with (1) a canonical inclusion $i \colon A \hookrightarrow |Q|$; (2) a canonical faithful conditional expectation $E_{A} \colon |Q| \to A$ with finite Watatani index.
\end{theorem}

For the convenience of the reader, we briefly review the construction of $|Q|$ here.
We first show that $|Q|$ is a $\text{C}^*$-algebra.
The multiplication map is $|Q| \otimes |Q| \xrightarrow{\pi} |Q| \otimes_{A} |Q| \xrightarrow{F(m)} |Q|$, where $m \in \C(Q \otimes Q, Q)$ is the multiplication morphism of the $Q$-system.
The inclusion $i \colon A \hookrightarrow |Q|$ is given by $F(u)$, where $u \in \C(1_\C, Q)$ is the unit morphism of the Q-system; note that $F(1_\C) = 1_{\mathrm{Bim}(A)} = A$.
The unit of $|Q|$ is $i(1_{A})$.
It is straightforward to verify the definition satisfies associativity and unitality.

The conditional expectation is $F(\varepsilon)$, where $\varepsilon \in \C(Q,1)$ is the counit morphism.
To define the involution on $|Q|$, we note $F(Q) \simeq \mathrm{Bim}(A)(A, F(Q))$, therefore $x \in F(Q)$ could be identified with a bimodule intertwiner $f_x \colon A \to F(Q)$.
Then we define $x^*$ to be the element corresponding to the bimodule intertwiner
\[
    A \xrightarrow{F(\Delta \circ u)} F(Q) \boxtimes_A F(Q) \xrightarrow{f_x^\dag \boxtimes_A \mathrm{id}_{F(Q)}} F(Q),
\]
where $u \in \C(1_\C,Q)$ is the unit morphism and $\Delta \in \C(Q, Q \otimes Q)$ is the comultiplication morphism.

Furthermore, bimodules of $Q$ internal to $\C$ could be translated to Hilbert bimodules of $|Q|$.

\begin{theorem}[{\cite[Construction 4.8]{Chen_2024_QSysCompletion}}]\label{thm:functor-G-fully-faithful}
    There is a canonical fully faithful functor $G \colon {}_Q \C_Q \to \mathrm{Bim}(|Q|)$.
\end{theorem}

Again we briefly review the construction of the functor $G$.
Given some $M \in {}_Q \C_Q$, we define the Hilbert space of $G(M)$ to be $F(M)$; the right action is given by $F(M) \otimes |Q| \xrightarrow{\pi} F(M) \otimes_{A} |Q| \xrightarrow{F(r)} F(M)$, and the left action is given by $|Q| \otimes F(M)  \xrightarrow{\pi} |Q| \otimes_{A} F(M) \xrightarrow{F(l)} F(M)$, where $l$ and $r$ are left and right action morphisms of $M$.

Finally, we review a characterization of the replete image of $G$, which would be used in showing the essential surjectivity of $\ReaDHR \colon (\CMdual)^{\rev} \to \DHR(\Loc^{\bdy}_\bullet)$ in Section~\ref{ComputationDHRB}.

\begin{lemma}[{\cite[Section 4.1.3]{Jones_2024_DHR}}]
There is a canonical linear restriction functor $\operatorname{Res} \colon \mathrm{Bim}(|Q|) \to \mathrm{Bim}(A)$ defined as follows:
\begin{enumerate}
    \item For $M \in \mathrm{Bim}(|Q|)$, define $\operatorname{Res}(M)$ to have the same vector space as $M$, and define the left and right $A$-actions to be restrictions of $|Q|$ actions.
The $A$-valued sesquilinear form is given by $\langle\xi \mid \nu\rangle_{A}:= E_{A}\left(\langle\xi \mid \nu\rangle_{|Q|}\right)$.
    \item For a bimodule intertwiner $f$, just define $\operatorname{Res}(f)=f$, which is automatically an $A$-bimodule intertwiner.
\end{enumerate}
\end{lemma}

\begin{theorem}[{\cite[Theorem 4.12]{Jones_2024_DHR}}]\label{ImageOfG}
Let $\mathrm{Bim}(|Q|, \C)$ be the full subcategory of $\mathrm{Bim}(|Q|)$ spanned by objects $M$ such that $\operatorname{Res}(M) \cong F(Y)$ for some $Y \in \C$.
Then $G \colon {}_Q \C_Q \to \mathrm{Bim}(|Q|, \C)$ is an equivalence of $\text{C}^*$-tensor categories.
\end{theorem}

\subsection{Properties of fusion spin chains}\label{subsec:properties-fusion-spin-chains}
In this section, we first review some useful properties of fusion spin chains without boundary.
Then we generalize these properties to fusion spin chains with boundary.
We fix a unitary fusion category \(\C\), a strong tensor generator \(x \in \C\), an indecomposable module category \(\M\), and a nonzero object \(l \in \M\).
Let \(\Loc^{\bulk}_\bullet\) be the fusion spin chain without boundary and \(\Loc^{\bdy}_\bullet\) be the fusion spin chain with boundary constructed from these data.
We denote their quasi-local algebras by \(\Loc^{\bulk}\) and \(\Loc^{\bdy}\), respectively.

\begin{lemma}
\(\Loc^{\bdy}\) is a simple AF-algebra with a unique tracial state.
In particular, the tracial state could be constructed from the categorical trace.
\end{lemma}

\begin{proof}
Note that \(\Loc^{\bdy}\) has a simple stationary Bratteli diagram corresponding to fusion rules of the module action \(\M \times \C \to \M\).
This implies that \(\Loc^{\bdy}\) is simple with a unique tracial state (the reader might refer to \cite{Bra72} or \cite[Chapter 6]{Effros_1981_DimensionsAC}).
\end{proof}

\begin{corollary}
\(\Loc^{\bulk}_{(-\infty, a)} \otimes \Loc^{\bulk}_{(b, \infty)}\) is a simple \(C^*\)-algebra.
\end{corollary}

\begin{proof}
First, note that the tensor product in \(\Loc^{\bulk}_{(-\infty, a)} \otimes \Loc^{\bulk}_{(b, \infty)}\) is unambiguous, as these two algebras are AF, thus nuclear.
Therefore the tensor product coincides with the minimal tensor product.
Note that \(\Loc^{\bulk}_{(b, \infty)}\) is simple since it coincides with \(\Loc^{\bdy}\) by taking \(\M = \C\); similarly, \(\Loc^{\bulk}_{(-\infty, a)}\) is simple.
Finally, minimal tensor products of simple \(C^*\)-algebras are simple \cite[II.9.5.3]{Blackadar_OperatorAlg_2017}.
\end{proof}

\begin{theorem}
\label{thm:covering-property}
(Covering property) \cite[Proposition 4.14]{Jones_2024_DHR}
There is some \(L>0\), such that if \(I, J\) are intervals with an overlap of length greater than \(L\), then \(\Loc^{\bulk}_{I \cup J}\) is generated by \(\Loc^{\bulk}_I\) and \(\Loc^{\bulk}_J\).
\end{theorem}

\begin{proof}
To prove the theorem, we take \(L\) to be the smallest positive integer such that \(x^n\) contains all simple objects as direct summand.
It suffices to show that if \(k \geq L+1\), then \(\Loc^{\bulk}_{[a, a+k]}\) is generated by the subalgebras \(\Loc^{\bulk}_{[a, a+k-1]}\) and \(\Loc^{\bulk}_{[a+1, a+k]}\).
The idea is to take a basis for \(\C(x^{k+1}, x^{k+1})\) and show each basis element is a product \((1_x \otimes f)(g \otimes 1_x)\) for \(f,g \in \C(x^k, x^k)\), which is straightforward but rather tedious.
We refer the reader to the reference for details.
\end{proof}

\begin{theorem}
\label{thm:algebraic-haag-duality}
(Algebraic Haag duality) \cite[Proposition 4.14]{Jones_2024_DHR}
There is a constant \(R > 0\) such that for any \(b-a>R\), we have \((\Loc^{\bulk}_{(-\infty, a)} \otimes \Loc^{\bulk}_{(b, \infty)})' = \Loc^{\bulk}_{[a,b]}\).
\end{theorem}

An vector $v$ in a $\Loc^{\bdy}$-module $M$ is said to be \emph{central} if $a \vartriangleright v = v \vartriangleleft a$ for any $a \in \Loc^{\bdy}$. The idea is to regard \(\Loc^{\bulk}\) as a bimodule of \(\Loc^{\bulk}_{(-\infty, a)} \otimes \Loc^{\bulk}_{(b, \infty)}\), then the commutant consists of central vectors in the bimodule. 
To implement the idea, we need some technical lemmas.
First, we should give the bimodule a categorical characterization.
Let \(m = x^{b-a+1}\), which is the object in the interval \([a,b]\).
Note that \(\C\) is canonically a module category of \(\C \boxtimes \C^{\rev}\), therefore we can consider the internal hom \([m,m]\), which has the structure of a Q-system \cite[Theorem 4.6]{Chen_2024_InductiveLimitAF}\footnote{Some non-canonical choice must be made to equip \([m,m]\) with the structure of a Q-system. However, what we care is the bimodule structure of \(\Loc^{\bulk}\), which doesn't depend on the choice. See \cite[Remark 4.9]{Jones_2024_DHR} for more detailed discussion.}.
Let's denote \(Q_{a,b} = [m,m]\).

\begin{lemma}
\(\Loc^{\bulk} \simeq (L^a \boxtimes R^b)(Q_{a,b})\) as bimodules.
\end{lemma}

\begin{proof}
It suffices to prove the isomorphism for finite pieces, i.e. \((L^a_k \boxtimes R_k^b)(Q_{a,b}) \simeq \Loc^{\bulk}_{[a-k,b+k]}\).
Since \(\C\) is semisimple, we have \(Q_{a,b} \simeq \oplus_{y,z \in \mathrm{Irr}(\C)} \C(y \otimes x^{b-a+1} \otimes z, x^{b-a+1}) \otimes (y \boxtimes z)\).
It follows that \((L^a_k \boxtimes R_k^b)(Q_{a,b}) \simeq \oplus_{y,z \in \mathrm{Irr}(\C)} \C(y \otimes x^{b-a+1} \otimes z, x^{b-a+1}) \otimes \C(x^k, y \otimes x^k) \otimes \C(x^k, x^k \otimes z)\).
Using the same technique as in Theorem~\ref{thm:covering-property}, it's straightforward to verify \((L^a_k \boxtimes R_k^b)(Q_{a,b}) \simeq \Loc^{\bulk}_{[a-k,b+k]}\).
\end{proof}

\begin{lemma}
If \(M\) is a simple \(\Loc^{\bulk}_{(-\infty, a)} \otimes \Loc^{\bulk}_{(b, \infty)}\)-bimodule with a nonzero central vector \(v\), then \(M\) must be isomorphic to \(\Loc^{\bulk}_{(-\infty, a)} \otimes \Loc^{\bulk}_{(b, \infty)}\) as a bimodule.
\end{lemma}

\begin{proof}
First, note that \(M\) could be generated by the single vector \(v\).
We claim that if \(a \cdot v = 0\) for some \(a \in \Loc^{\bulk}_{(-\infty, a)} \otimes \Loc^{\bulk}_{(b, \infty)}\), then \(a=0\).
The reason is that all such \(a\) form a two-sided closed ideal in \(\Loc^{\bulk}_{(-\infty, a)} \otimes \Loc^{\bulk}_{(b, \infty)}\); however, as the algebra is simple, this ideal must be zero.
Finally, we obtain a bimodule isomorphism \(M \simeq \Loc^{\bulk}_{(-\infty, a)} \otimes \Loc^{\bulk}_{(b, \infty)}\) by \(a \cdot v \mapsto a\).
This finishes the proof.
\end{proof}

\emph{Proof of Theorem~\ref{thm:algebraic-haag-duality}.}
First, note that \(\Loc^{\bulk}\) is a semisimple bimodule.
Then it suffices to find central vectors in each simple summand of \(\Loc^{\bulk}\).
We know that only simple summands isomorphic to \(\Loc^{\bulk}_{(-\infty, a)} \otimes \Loc^{\bulk}_{(b, \infty)}\) as a bimodule could contain nonzero central vectors.
As \(L^a \boxtimes R^b\) is fully faithful, \(L^a \boxtimes R^b(y \boxtimes z) \simeq \Loc^{\bulk}_{(-\infty, a)} \otimes \Loc^{\bulk}_{(b, \infty)}\) iff \(y \simeq z \simeq \mathbbm{1}_\C\).
Therefore, simple summands isomorphic to \(\Loc^{\bulk}_{(-\infty, a)} \otimes \Loc^{\bulk}_{(b, \infty)}\) in \(\Loc^{\bulk} \simeq (L^a \boxtimes R^b)(Q_{a,b})\) must correspond to tensor unit summands in \(Q_{a,b}\).
Let's denote the subobject of \(Q_{a,b}\) containing all tensor unit summands as \(U\).
Recall that
\[
    (L^a_k \boxtimes R_k^b)(Q_{a,b}) \simeq \oplus_{y,z \in \mathrm{Irr}(\C)} \C(y \otimes x^{b-a+1} \otimes z, x^{b-a+1}) \otimes \C(x^k, y \otimes x^k) \otimes \C(x^k, x^k \otimes z).
\]
By taking \(y=z=1_{\C}\), we find 
\[
    (L^a_k \boxtimes R_k^b)(U) = \C(x^{b-a+1}, x^{b-a+1}) \otimes \C(x^k, x^k) \otimes \C(x^k, x^k) = \Loc^{\bulk}_{[a-k,a)} \otimes \Loc^{\bulk}_{[a,b]} \otimes \Loc^{\bulk}_{(b,b+k]}.
\]
By taking the limit \(k \to \infty\), we find that \((L^a \boxtimes R^b)(U) = \Loc^{\bulk}_{(-\infty, a)} \otimes \Loc^{\bulk}_{[a,b]} \otimes \Loc^{\bulk}_{(b, \infty)}\).
The center of \(\Loc^{\bulk}_{(-\infty, a)} \otimes \Loc^{\bulk}_{(b, \infty)}\) is trivial since any simple unital \(C^*\)-algebra must have trivial center.
Thus the set of central vectors in \((L^a \boxtimes R^b)(U)\) is precisely \(\Loc^{\bulk}_{[a,b]}\).
Finally, the previous lemma shows that \((L^a \boxtimes R^b)(U)\) must contain all central vectors in \(\Loc^{\bulk}\).
This finishes the proof.

Similarly, we can show that the covering property and algebraic Haag duality hold for fusion spin chains with boundary:

\begin{theorem}
\label{thm:covering-property-boundary}
(Covering property)
There is an \(L>0\) such that if \(I, J\) are intervals whose intersection contains an interval of length \(L\), then \(\Loc^{\bdy}_{I \cup J}\) is generated by \(\Loc^{\bdy}_I\) and \(\Loc^{\bdy}_J\).
\end{theorem}

\begin{proof}
Again we take a basis for \(\M(l \otimes x^{k}, l \otimes x^{k})\) and show each basis element is a product \((1_x \otimes f)(g \otimes 1_x)\) for \(f \in \C(x^k, x^k), g \in \M(l \otimes x^{k-1}, l \otimes x^{k-1})\).
The proof is completely analogous to before.
\end{proof}

\begin{theorem}
\label{thm:algebraic-haag-duality-boundary}
(Algebraic Haag duality)
There is a constant \(K > 0\) such that for any \(N>K\), we have \((\Loc^{\bdy}_{(N, \infty)})' = \Loc^{\bdy}_{[0,N]}\).
\end{theorem}

\begin{proof}
Again, we view \(\Loc^{\bdy}\) as a bimodule of \(\Loc^{\bdy}_{(N, \infty)}\) and try to find the central vectors.
The rest is completely analogous to before: \(\Loc^{\bdy} = L^a(Q_N)\) where \(Q_N = [l \otimes x^{N}, l \otimes x^{N}]\), and central vectors must lie in the image of tensor unit summands of \(Q_N\), which turns out to be \(\Loc^{\bdy}_{[0,N]}\).
\end{proof}

The next theorem is fundamental for the theory of DHR bimodules. We slightly generalize the statement in \cite[Remark 4.7]{Jones_2024_DHR} and improve the proof presented there, so that the theorem holds for any module category $\M$ and the proof does not rely on techniques of subfactor theory:
\begin{theorem}[{\cite[Remark 4.7]{Jones_2024_DHR}}]\label{FullyFaithful}
Standard actions are fully faithful, that is, for any objects $L_1, L_2\in \C_\M^\vee$, the map $\CMdual(L_1, L_2) \to \mathrm{Bim}(\Loc^{\bdy})(\Std(L_1),\Std(L_2))$ induced by $\Std$ is a bijection.
\end{theorem}
\begin{proof}
Let $L_1^\vee$ be a left dual of $L_1$ in $\CMdual$, and let $L=L_2\otimes L_1^\vee$.
Then $L_1^\vee$ is a right dual of $L_1$ in $(\CMdual)^{\rev}$.
Since $\Std$ is monoidal, rigidity gives natural isomorphisms
\[
\begin{aligned}
\mathrm{Bim}(\Loc^{\bdy})(\Std(L_1),\Std(L_2))
&\simeq \mathrm{Bim}(\Loc^{\bdy})(\Loc^{\bdy}, \Std(L_1^\vee) \boxtimes_{\Loc^{\bdy}} \Std(L_2))\\
&\simeq \mathrm{Bim}(\Loc^{\bdy})(\Loc^{\bdy}, \Std(L)).
\end{aligned}
\]

Note that $\Loc^{\bdy}$, the unit bimodule, is generated by a single vector $1$ (the unit of the algebra).
Therefore, the image of $1$ completely determines a bimodule morphism $\Loc^{\bdy} \to \Std(L_1^\vee) \boxtimes_{\Loc^{\bdy}} \Std(L_2)$.
Note that the image of a central vector under any bimodule morphism $f: M \to M'$ must also be a central vector, as $a \vartriangleright f(v) = f(a \vartriangleright v) = f(v \vartriangleleft a) = f(v) \vartriangleleft a$. Since $1 \in \Loc^{\bdy}$ is a central vector, its image must also be central. Therefore, $\mathrm{Bim}(\Loc^{\bdy})(\Loc^{\bdy}, \Std(L))$ identifies with the set of central vectors in $\Std(L)$. Next we prove that a central vectors in $\Std(L)$ is strictly localized in some finite interval.

Choose $r$ greater than the constant in Theorem~\ref{thm:algebraic-haag-duality-boundary}.
Since $l\otimes x^r$ contains every simple object of $\M$, semisimplicity gives a finite collection $\{s_i\}\subseteq \Std_r(L)$ such that
\[
    \sum_i s_i\circ s_i^*=1_{L\otimes l\otimes x^r}.
\]
Their images form a projective basis of $\Std(L)$, and for every $c\in \Loc^{\bdy}_{(r,\infty)}$, we have $c\triangleright s_i=s_i\triangleleft c$.
Let $v\in \Std(L)$ be a central vector and set $a_i:=\langle s_i\mid v\rangle\in \Loc^{\bdy}$.
For $c\in \Loc^{\bdy}_{(r,\infty)}$, we have
\[
\begin{aligned}
c a_i
&=\langle s_i\triangleleft c^*\mid v\rangle
=\langle c^*\triangleright s_i\mid v\rangle \\
&=\langle s_i\mid c\triangleright v\rangle
=\langle s_i\mid v\triangleleft c\rangle
=a_i c.
\end{aligned}
\]
Thus $a_i\in (\Loc^{\bdy}_{(r,\infty)})'=\Loc^{\bdy}_{[0,r]}$ by Theorem~\ref{thm:algebraic-haag-duality-boundary}.
It follows that
\[
    v=\sum_i s_i\triangleleft a_i\in \Std_r(L).
\]

We can now identify the central vectors in $\Std(L)$ with module natural transformations from the identity functor to $L$.
Regard $v$ as a morphism
\[
    v_r\colon l\otimes x^r\longrightarrow L\otimes l\otimes x^r,
\]
and set $v_n:=v_r\otimes 1_{x^{n-r}}$ for $n\geq r$.
Centrality gives
\[
    (1_L\otimes a)\circ v_n=v_n\circ a,\qquad a\in \Loc^{\bdy}_n.
\]
For each simple object $y\in \M$, choose an isometry $i_y\colon y\to l\otimes x^n$, and define
\[
    f_y:=(1_L\otimes i_y^*)\circ v_n\circ i_y\colon y\longrightarrow L\otimes y.
\]
Taking $a=i_z\circ g\circ i_y^*$ in the centrality condition of $v_n$, one finds
\[
    (1_L\otimes g)\circ f_y=f_z\circ g
\]
for every $g\colon y\to z$. In particular, the above implies that $f_y$ is independent of the choice of $i_y$ (taking $a=i_y'\circ i_y^*$ for two choices $i_y'$ and $i_y$).
These maps $\{f_y\}$ therefore form a natural transformation $f\colon \operatorname{id}_{\M}\Rightarrow L$.
Since every object of $\C$ is a direct summand of some tensor power of $x$, the compatibility $v_{n+q}=v_n\otimes 1_{x^q}$ implies, under the right module structure isomorphisms, that
\[
    f_{y\otimes c}=f_y\otimes 1_c,
\]
for every $c\in \C$.
Thus $f$ is a $\C^\rev$-module natural transformation.

Conversely, a $\C^\rev$-module natural transformation $f\colon \operatorname{id}_{\M}\Rightarrow L$ gives compatible vectors
\[
    v_n:=f_{l\otimes x^n}\in \Std_n(L).
\]
Ordinary naturality implies $(1_L \otimes a) \circ v_n=v_n\circ a$ for every $a\in \Loc^{\bdy}_n$, so the resulting vector $v\in \Std(L)$ is central.
We have therefore proved that
\[
\mathrm{Bim}(\Loc^{\bdy})(\Loc^{\bdy},\Std(L))
\simeq \operatorname{Nat}_{\C^\rev}(\operatorname{id}_{\M},L)
=\CMdual(\operatorname{id}_{\M},L)
\simeq \CMdual(L_1,L_2).
\]
All the isomorphisms above are natural, and their composite is the map induced by $\Std$.
\end{proof}

%--------------------------------------------

\section{Q-system model for 1D phases}
\label{Q-system}
In this section, we first introduce how to define commuting-projector models for 1D phases on fusion spin chains using Q-systems \cite{MYLG25} and prove they have a unique ground state, then generalize the construction and the result to fusion spin chains with boundaries.

Q-systems are special Frobenius algebras such that the comultiplication morphism is the Hermitian conjugate of the multiplication morphism, and the counit is the Hermitian conjugate of the unit \cite{Chen_2024_QSysCompletion}.
Before diving into details, we would like to clarify two important points.
First, the models are defined on \emph{infinite} spin chains, which are the appropriate setting for studying topological phases.
Second, we consider only \emph{symmetric} operators and \emph{symmetry-invariant} states.

\begin{remark}
It is proposed in \cite{Gaiotto_Johnson-Freyd_2025_CondensationsInHigherCats} that given a unitary fusion category $\C$ as the category of local quantum charges, one-dimensional gapped phases correspond to separable algebras (condensation 1-monads in their terminology) in $\C$.
However, our construction requires algebras that are not only separable, but also \emph{special $C^*$ Frobenius}. Is there any discrepancy between our construction and the proposal in \cite{Gaiotto_Johnson-Freyd_2025_CondensationsInHigherCats}?
The answer is no:
It's shown in \cite{Giorgetti_2024} that an algebra in $\C$ is isomorphic to a special $C^\ast$-Frobenius algebra if and only if it's separable.

Similarly, when we construct Hamiltonians with boundaries, we need \emph{special $C^*$ Frobenius modules}.
It's shown in \cite{Giorgetti_2023} that any right $Q$-module is isomorphic to a special one. Furthermore, it is shown in \cite{LanZhou_2024_QuantumCurrent} that if a module $(K,\mu)$ over a $C^\ast$-special Frobenius algebra $(Q,m,\eta,m^\dagger,\eta^\dagger)$ is special, namely if
\begin{equation}
\mu\mu^\dagger = \id_K, 
\end{equation}
 then it is automatically a special $C^*$ Frobenius module, i.e. 
\begin{equation}
\label{eq.BdyFrob}
    (\id_K\ot m)(\mu^\dagger \ot \id_Q) = \mu^\dagger \mu = (\mu\ot id_Q)(\id_K\ot m^\dagger).
\end{equation}

To summarize, unitarity and Frobenius conditions are automatic in (1+1)D, and the choice is unique.
However, it is unknown that whether the same holds true in higher dimensions. We conjecture that unitarity and Frobenius conditions are not automatic in higher dimensions, and the same separable algebra might be endowed with inequivalent unitarity and Frobenius structures.
\end{remark}

\subsection{Q-system models without boundary}
\begin{definition}[Q-system model]\label{def:q-system-model}
    Fix a Q-system $(Q, m, \eta, \Delta, \varepsilon)$ in a unitary fusion category $\C$. The \emph{Q-system model} is a spin chain model defined as follows: 
    
    We define local degrees of freedom by choosing $x = Q$ and defining a fusion spin chain, i.e. placing the object $Q$ on each site. We then define a commuting-projector interaction by $\Phi_Q([i,i+1]) = 1 - m_{i, i+1}^\dagger m_{i,i+1}$, where $m_{i,i+1} \colon Q \otimes Q \to Q$ is the multiplication morphism of the two sites $i$ and $i+1$, and $\Phi_Q([i,j]) = 0$ whenever $j-i \neq 1$. In other words, the Hamiltonian is defined as $H = \sum_i 1 - m_{i, i+1}^\dagger m_{i,i+1}$.
\end{definition}

\begin{remark}\label{rmk: projections}
    Actually we should put a sufficiently large object $x$ on each site, and choose an orthogonal projection $p: x \to Q$. Then the interaction should be defined as $\Phi_Q([i,i+1]) = (p_i p_{i+1})^\dagger (1 - m_{i, i+1}^\dagger m_{i,i+1}) (p_i p_{i+1})$. We omit the projections in formulas for the sake of simplicity, but we retain the projections and their Hermitian conjugates in the diagrams and represent them as triangles.
\end{remark}

To prove the model has a unique ground state, we need a fact about the fixed-point algebra.

Given an interval $[i,j] \subset \mathbb{Z}$, we define the composition of multiplication morphisms in $[i,j]$ as $m_{[i,j]} := m_{i,i+1} m_{i+1,i+2} \dots m_{j-1,j}$.
By Proposition~\ref{prop:ground_state_projector_eigenstate}, any ground state $\rho_Q$ of $\Phi_Q$ should satisfy
\begin{equation}
    \rho_Q((m^\dag_{i,i+1} m_{i,i+1}) a (m^\dag_{i,i+1} m_{i,i+1})) = \rho_Q(a),
\end{equation}
which implies $\rho_Q((m^\dag_I m_I) a (m^\dag_I m_I)) = \rho_Q(a)$.
Therefore, we could define the fixed-point algebra as in \cite{MYLG25}:

\begin{definition}\label{def:screening-map}
The \emph{screening map} is a linear map $\operatorname{scr} \colon A \to \operatorname{End}(Q)$ defined as follows: Given an operator $a \in A_{[i,j]}$, we define $\operatorname{scr} (a) = m_{[i-1,j+1]} (\mathrm{id}_Q \otimes a \otimes \mathrm{id}_Q) m^\dag_{[i-1,j+1]}$.
The \emph{fixed-point algebra} is defined as the image of $\operatorname{scr}$.

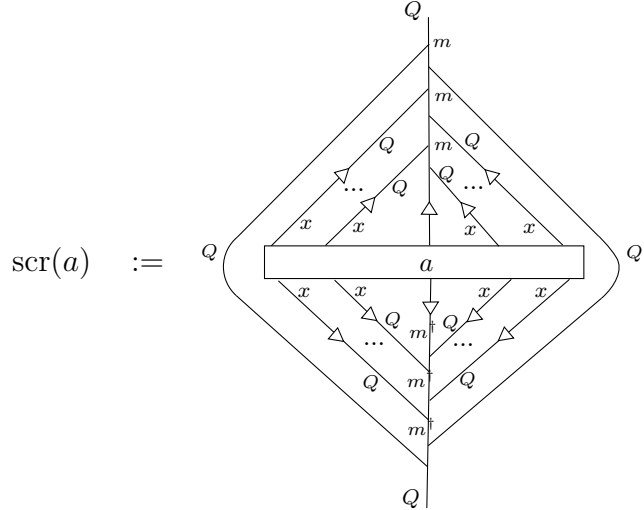
\[
\input{diagrams/screening-map-1293.tex}
\]

\end{definition}

\begin{lemma}\label{lemma:fixed_point_algebra}
    When $Q$ is a simple Q-system, the fixed-point algebra is the one-dimensional algebra $\mathbb{C} \cdot \mathrm{id}_Q$.
\end{lemma}
\begin{proof}
Note that the Frobenius condition implies that any $\operatorname{scr} (a) \colon Q \to Q$ is a $Q$-$Q$ bimodule map, i.e. $m \circ (\mathrm{id}_Q \otimes \operatorname{scr} (a)) = \operatorname{scr} (a) \circ m$ and $m \circ (\operatorname{scr} (a) \otimes \mathrm{id}_Q) = \operatorname{scr} (a) \circ m$.
Since $Q$ is simple, any $Q$-$Q$ bimodule map is proportional to $\mathrm{id}_Q$.
\end{proof}

\begin{figure}
    \centering

    \resizebox{1.05\textwidth}{!}{
\input{diagrams/frobenius-module-1436.tex}

    }

    \caption{The Frobenius conditions ensure that \(\mathrm{scr}(a)\) is a right \(Q\)-module map. Similarly, it is also a left \(Q\)-module map.}

\end{figure}
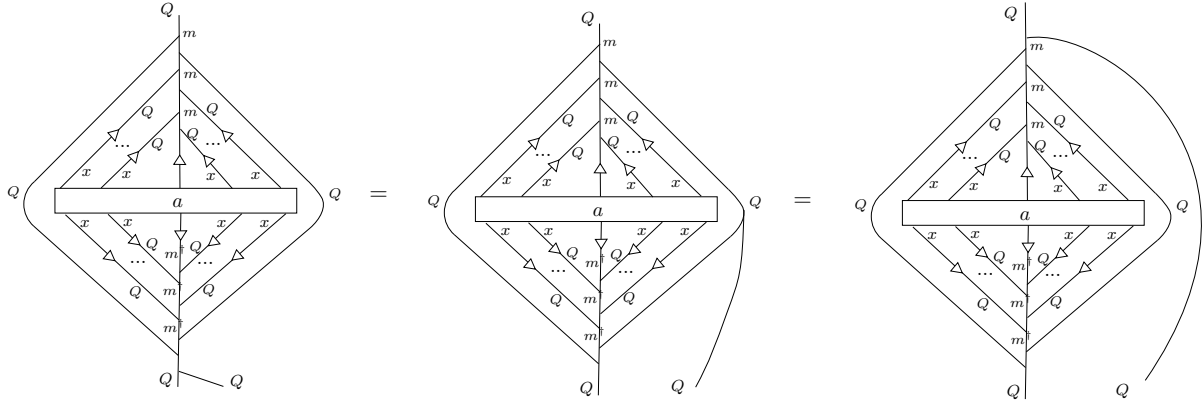

\begin{theorem}
The interaction $\Phi_Q$ has a unique ground state $\rho_Q$ for any simple Q-system $Q$.
\end{theorem}
\begin{proof}
We would show the existence and uniqueness of the ground state by explicit construction.
To determine the expectation value $\rho_Q(a)$ of a local operator $a$, note that any ground state $\rho_Q$ of $\Phi_Q$ should satisfy $\rho_Q(a) = \rho_Q((m^\dag_{i,i+1} m_{i,i+1}) a (m^\dag_{i,i+1} m_{i,i+1})) = \rho_Q(m^\dag_{i,i+1} (\operatorname{scr} (a)) m_{i,i+1})$.
However, Lemma~\ref{lemma:fixed_point_algebra} tells us that $\operatorname{scr} (a) = \lambda_a \mathrm{id}_Q$ for some $\lambda_a \in \mathbb{C}$, therefore
\[
    \rho_Q(a) = \rho_Q(m^\dag_{i,i+1} \lambda_a \mathrm{id}_Q m_{i,i+1}) = \lambda_a \rho_Q(m^\dag_{i,i+1} m_{i,i+1}) = \lambda_a \rho_Q(1) = \lambda_a,
\]
which finishes the construction of the ground state.
\end{proof}

\begin{remark}
    The ground state of $\Phi_Q$ is defined by the condition that all local terms $m_{i,i+1}^\dagger m_{i,i+1}$ have the maximum expectation value $+1$. We avoid defining ground states of generic Hamiltonians, as expectation values of Hamiltonians are ill-defined in the thermodynamic limit. Although there is an operator algebraic definition of ground states in the thermodynamic limit (see, for example, \cite{ogata_2021_classificationgappedgroundstate}), it doesn't quite match the physical intuition of a ground state. See \cite[Section IV]{MYLG25} for a discussion on this issue.
\end{remark}

\begin{remark}
This theorem holds true even when $Q$ corresponds to a 1D phase with spontaneous symmetry breaking.
Why do Hamiltonians with spontaneous symmetry breaking still have a unique ground state in our framework? The reason is that the operator algebra we consider only contains \textbf{symmetric} operators, therefore all states we consider are necessarily symmetry-invariant. 
It is easy to verify that in the conventional setup of symmetry breaking phases, the only symmetry-invariant state in the ground state space (which is not a vector space, but a convex set of positive linear functionals on the operator algebra) is the maximally mixed state.

Actually, in our framework, the notion of ``non-symmetric operators'' is absent without further structures specified.
This stands in contrast to the traditional approach, where one begins with a theory without an explicit symmetry, identifies symmetry actions afterward, and then defines symmetric operators as those commuting with the symmetry. 
In that case, symmetric operators arise as a subset of the larger algebra of all quasi-local operators. 
Since non-symmetric operators exist naturally in this setup, one can then consider non-symmetric perturbations and observables, allowing one to meaningfully ask whether a phase has symmetry-breaking.

If the symmetric charge category $\C$ is anomaly-free, that is, if there exists fiber functors 
$F:\C\rightarrow \ve$ to the category of vector spaces, then we may choose a fiber functor to identify objects in 
$\C$ with vector spaces, and non-symmetric operators become well-defined. 
Note that the fiber functor is not unique in general, so the identification is not canonical.
Actually, choosing a fiber functor is equivalent to choosing an identification of the symmetry action as a Hopf algebra by Tannaka duality.

However, if $\C$ is anomalous (i.e. no such fiber functor exists), there is no canonical way to define non-symmetric operators, and it would be difficult to work with the traditional approach. In such cases, notions like ``SPT phase'' or ``symmetry-breaking phase'' cease to be meaningful within this framework.

\end{remark}

\begin{example}\label{eg. Z2 gapped phases}
1D phases with $\Z_2$ global symmetry: Up to Morita equivalence, there are two Q-systems in $\mathrm{Rep}(\Z_2)$.
Denote the tensor unit in $\mathrm{Rep}(\Z_2)$ as $\one$ and the other simple object as $e$.
The first one is $1$ on object level, and the multiplication morphism $m_1 \colon 1 \otimes 1 \to 1$ is just the identity; obviously it gives the trivial phase, since any state defined on the fusion spin chain with $x=1$ is a product state.
The second one is $1 \oplus e$ on object level, and the multiplication morphism $m_2 \colon (1 \oplus e) \otimes (1 \oplus e) \to 1 \oplus e$ is given by the following matrix:
\[
    \begin{array}{ c|cccc }
     & \one \otimes \one & \one \otimes e & e \otimes \one & e\otimes e\\
    \hline
    \one & 1/\sqrt{2} & 0 & 0 & 1/\sqrt{2}\\
    e & 0 & 1/\sqrt{2} & 1/\sqrt{2} & 0
    \end{array}
\]
It is easy to verify that $m_2^\dagger m_2 = (1+Z_i Z_{i+1})/2$, corresponding to the symmetry-breaking phase.
\end{example}

\subsection{Q-system models with boundary}
\label{subsec: Q-system models with bdy}

The construction of Q-system models with boundary is almost identical to those without boundary, with the difference that we place a right $Q$-module on the boundary.

\begin{definition}\label{def:q-system-model-with-boundary}
    Fix a Q-system $(Q, m, \eta, \Delta, \varepsilon)$ in a unitary fusion category $\C$, a unitary $\C$-module category $\M$, and a special $C^*$ Frobenius right $Q$-module $K \in \M_Q$ with action $\mu \colon K \otimes Q \to K$. The \emph{Q-system model with boundary} is defined as follows:
    
    We define the local degrees of freedom by choosing the bulk object to be $x = Q$ and the boundary object to be $l = K$\footnote{Again we should actually choose sufficiently large objects and orthogonal projections, see Remark~\ref{rmk: projections}.}, and defining a fusion spin chain with boundary, i.e. placing the object $Q$ on each bulk site $i > 0$ and the object $K$ on the boundary site $0$.
    
    We then define a commuting-projector interaction by $\Phi_{K,Q}([0,1]) = 1 - \mu^\dagger \mu$, $\Phi_{K,Q}([i,i+1]) = 1 - m_{i, i+1}^\dagger m_{i,i+1}$ for $i > 0$, and $\Phi_{K,Q}([i,j]) = 0$ whenever $j-i \neq 1$. In other words, the Hamiltonian is defined as $H = (1 - \mu^\dagger \mu)+ \sum_{i > 0} (1 - m_{i, i+1}^\dagger m_{i,i+1})$.
\end{definition}

Given an interval $[0,j] \subset \mathbb{Z}_{\geq 0}$, we define the composition of action and multiplication morphisms in $[0,j]$ as $\mu_j := \mu \circ (\mathrm{id}_K \otimes m_{[1,j]})$.
By Proposition~\ref{prop:ground_state_projector_eigenstate}, any ground state $\phi$ of $\Phi_{K,Q}$ should satisfy
\begin{equation}
    \phi((\mu^\dag \mu) a (\mu^\dag \mu)) = \phi(a),
\end{equation}
and similarly for the bulk terms, which implies $\phi((\mu^\dag_j \mu_j) a (\mu^\dag_j \mu_j)) = \phi(a)$ for any interval $[0,j]$.
Therefore, we could define the boundary fixed-point algebra:

\begin{definition}\label{def:boundary-screening-map}
The \emph{boundary screening map} is a linear map $\operatorname{scr} \colon \Loc^{\bdy} \to \operatorname{End}_{\MQ}(K)$ defined as follows: Given an operator $a \in \Loc^{\bdy}_{[0,j]}$, we define $\operatorname{scr} (a) = \mu_{j+1} (a \otimes \mathrm{id}_Q) \mu^\dag_{j+1}$.
The \emph{boundary fixed-point algebra} is defined as the image of $\operatorname{scr}$.

\begin{equation}
\begin{split}
    \operatorname{scr}:\Loc^{\bdy}&\rightarrow \End_{\MQ}(K)\\
    a&\mapsto\operatorname{scr}(a):=
\input{diagrams/boundary-screening-1906.tex}
\end{split}
\end{equation}

\end{definition}

\begin{lemma}\label{lemma:boundary_fixed_point_algebra}
          $\operatorname{im}(\operatorname{scr}) = \End_{\MQ}(K)$. In particular, when $K$ is a simple right $Q$-module in $\M$, the boundary fixed-point algebra is the one-dimensional algebra $\mathbb{C} \cdot \mathrm{id}_K$.
\end{lemma}
\begin{proof}
First, note that the Frobenius condition implies that any $\operatorname{scr} (a) \colon K \to K$ is a right $Q$-module map, i.e. $\mu \circ (\operatorname{scr} (a) \otimes \mathrm{id}_Q) = \operatorname{scr} (a) \circ \mu$, which means $\operatorname{im}(\operatorname{scr}) \subset \End_{\MQ}(K)$.
For the other direction, note that there is a canonical inclusion $\End_{\MQ}(K) \hookrightarrow \Loc$. It's straightforward to verify that composition $\End_{\MQ}(K) \hookrightarrow \Loc \xrightarrow{\mathrm{scr}} \End_{\MQ}(K)$ is the identity map, which implies $\operatorname{im}(\operatorname{scr}) \supset \End_{\MQ}(K)$.

 Finally, note that when $K$ is simple, any right $Q$-module map is proportional to $\mathrm{id}_K$.
\end{proof}

\begin{theorem}
The interaction $\Phi_{K,Q}$ has a unique ground state $\phi$ for any simple right $Q$-module $K$ in $\M$.
\end{theorem}
\begin{proof}
Again we show the existence and uniqueness of the ground state by explicit construction.
To determine the expectation value $\phi(a)$ of a local operator $a \in \Loc^{\bdy}_{[0,j]}$, note that any ground state $\phi$ of $\Phi_{K,Q}$ should satisfy $\phi(a) = \phi((\mu^\dag_{j+1} \mu_{j+1}) a (\mu^\dag_{j+1} \mu_{j+1})) = \phi(\mu^\dag_{j+1} (\operatorname{scr} (a)) \mu_{j+1})$.
However, Lemma~\ref{lemma:boundary_fixed_point_algebra} tells us that $\operatorname{scr} (a) = \lambda_a \mathrm{id}_K$ for some $\lambda_a \in \mathbb{C}$, therefore
\[
    \phi(a) = \phi(\mu^\dag_{j+1} \lambda_a \mathrm{id}_K \mu_{j+1}) = \lambda_a \phi(\mu^\dag_{j+1} \mu_{j+1}) = \lambda_a \phi(1) = \lambda_a,
\]
which finishes the construction of the ground state.
\end{proof}

\begin{theorem}\label{thm:boundary_unique_GS}
Let \(Q\) be a simple Q-system in \(\C\). The associated infinite-chain Q-system interaction has a unique ground state. More generally, if \(K\in \M_Q\) is a simple right \(Q\)-module, then the boundary interaction \(\Phi_{K,Q}\) on the half-infinite chain has a unique ground state.
\end{theorem}
\begin{proof}
This follows from the uniqueness theorem for \(\Phi_Q\) above and the uniqueness theorem for \(\Phi_{K,Q}\) above.
\end{proof}

\begin{remark}
   We only consider the boundary condition $K$ being a simple object in $\MQ$ in the above derivation. A non-simple object $K$ is a finite direct sum of simple objects $K=\oplus_{(i,n)}K_i^{\oplus n}$, since $\MQ$ is finite semisimple. $\End_{\MQ}(K)$ is then a multi-matrix algebra with multiplication the composition of right $Q$-module maps, and Lemma~\ref{lemma:boundary_fixed_point_algebra} still holds.
\end{remark}

%---------------------------------

\section{Classification of boundary conditions and construction of Hamiltonians}

\subsection{Boundary conditions and realization functors}

Before presenting rigorous definitions, we would like to ask the reader to reflect on the following question: What is a topological charge? Is it a single Hamiltonian, or a single quantum state, that is different from the ground state somewhere?

Our answer is, a topological charge should be understood as an equivalence class of quantum states, i.e. a superselection sector.
In other words, two quantum states that can be transformed to each other by a (quasi-) local operator should be viewed as topologically equivalent; the philosophy is, though they could differ in some finite region, viewed from a distance they are really the same.

The mathematical notion that captures the intuition is representations, or modules, of the quasi-local algebra.
It has been realized before that 0D topological excitations should really be understood as modules of the quasi-local algebra \cite{Wen_lecture_2016}; see also \cite[Section 3.2]{Kong_2022_Invitation}.
Given a quantum state, there is a canonical representation associated with it via the GNS construction of Section~\ref{subsec:states-and-reps}, which should be interpreted as the superselection sector that the quantum state lives in. 

Based on these observations, we propose that objects of the category of boundary conditions, $\BCond$, are modules of the quasi-local algebra that correspond to GNS constructions of quantum states with a topological charge.
We first give an explicit description of these GNS constructions.

Consider a simple right $Q$-module $K \in \MQ$, and denote the ground state of the corresponding boundary interaction $\Phi_{K,Q}$ from Section~\ref{subsec: Q-system models with bdy} by $\rho_K \colon \Loc^{\bdy} \to \mathbb{C}$. As this section only involves fusion spin chains with boundary, we will abbreviate $\Loc^{\bdy}$ as $\Loc$.
Denote the restriction of $\rho_K$ to $\Loc_{[0,n]}$ as $\rho_{K,n}$.
Recall that $\operatorname{GNS}(\rho_{K,n})$ is defined as $\Loc_{[0,n]} / N_{K,n}$, where $N_{K,n} = \{ a \in \Loc_{[0,n]} \mid \rho_{K,n}(a^* a) = 0 \}$.
There is a simple characterization of $N_{K,n}$:

\begin{lemma}
$N_{K,n} = \{ a \in \Loc_{[0,n]} \mid (a \otimes \mathrm{id}_x)\mu^\dag_{n+1} = 0 \}$.
\end{lemma}
\begin{proof}
Note that $\rho_{K,n}(a^* a) = \rho(a^* a)$, and $\rho(a^* a) = 0$ if and only if 
\[\operatorname{scr} (a^* a) = \mu_{n+1} (a^* a) \mu^\dag_{n+1} = (a \mu^\dag_{n+1})^*(a \mu^\dag_{n+1}) = 0.\]
The last equality holds if and only if $a \mu^\dag_{n+1} = 0$.
\end{proof}

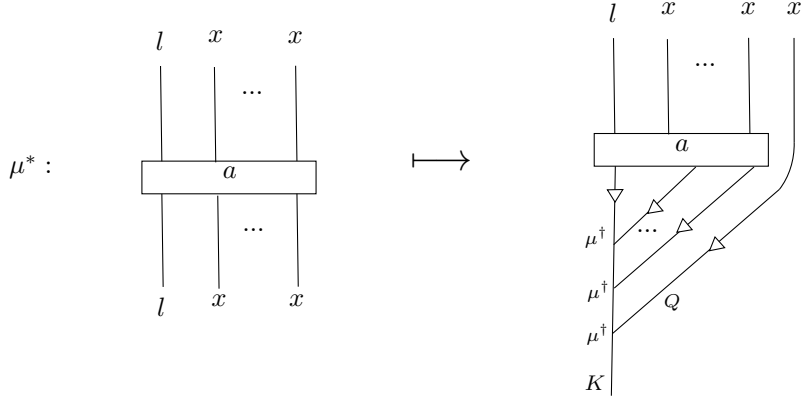
\[
\input{diagrams/gns-kernel-2056.tex}
\]

This leads us to consider the map $\mu_n^* \colon \Loc_{[0,n]} \to \M(K, l \otimes x^{n+1})$ defined by $\mu_n^*(a) = (a \otimes \mathrm{id}_x)\mu^\dag_{n+1}$.
Immediately, we arrive at an equivalent characterization of $\operatorname{GNS} (\rho_{K,n})$:

\begin{corollary}\label{cor:GNS-im-m-star}
There is a canonical isomorphism
\[
    \GNS(\rho_{K,n}) \simeq \operatorname{im}(\mu_n^*).
\]
\end{corollary}
\begin{proof}
From the first isomorphism theorem of modules, we have a canonical isomorphism
\[
    \Loc_{[0,n]} / \operatorname{ker}(\mu_n^*) \simeq \operatorname{im}(\mu_n^*)
\]
and $\operatorname{ker}(\mu_n^*) = N_{K,n}$ from the previous lemma, so $\Loc_{[0,n]} / \operatorname{ker}(\mu_n^*)$ is by definition $\operatorname{GNS}(\rho_{K,n})$.
\end{proof}

The GNS construction captures exactly how the IR description (the simple right $Q$-module $K \in \MQ$) relates to the complete UV description (the object $l \otimes x^n$).

Then we should find out how inclusions of local operator algebras induce inclusions of GNS modules.
For $r\geq 0$, define
\[
    j_{r,r+1} \colon \M(K, l \otimes x^{r}) \to \M(K, l \otimes x^{r+1}), \qquad
    j_{r,r+1}(f) = (f \otimes p_Q^\dagger) \mu_K^\dagger,
\]
where $\mu_K: K \otimes Q \to K$ is the action morphism of the right $Q$-module $K$ and $p_Q^\dagger: Q \to x$ is the inclusion of $Q$ into $x$. We will omit the subscript $K$ if it is clear from the context.
Obviously $j_{r,r+1}$ is injective.

For $a\in \Loc_{[0,n]}$, the inclusion of local algebras sends $a$ to $a\otimes \mathrm{id}_x\in \Loc_{[0,n+1]}$, and the corresponding GNS inclusion is expressed by
\[
    j_{n+1,n+2}(\mu_n^*(a)) = \mu_{n+1}^*(a \otimes \mathrm{id}_x).
\]
Thus the maps $j_{r,r+1}$ identify the finite GNS pieces under enlargement of the interval. Furthermore, $\operatorname{im} j_{r,r+1} \subset \operatorname{im}(\mu_r^*)\simeq \GNS(\rho_{K,r})$, which implies all morphisms $K \to l \otimes x^r$ are contained in the GNS construction after applying the inclusion $j_{r,r+1}$.

There are some drawbacks in the discussion above: We only define how the simple objects of  $\MQ$ are realized as modules of the quasi-local algebra; neither non-simple objects nor morphisms in $\MQ$ are dealt with.
Fortunately, there is an obvious fix of these drawbacks: extend the realization of boundary conditions as modules to a functor.

Again, we first define finite pieces of the realization functor, then glue them into an AF action functor:
\begin{definition}[The realization functor for boundary conditions]\label{Def. The realization functor for boundary conditions}

We first define $\Loc_n:=\Loc_{[0,n]}$ and $\ReaBCond_n: \MQ^{\op} \to \Mod(\Loc_n)$.
As before, we define $\ReaBCond_n(K) := \M(K, l \otimes x^n)$.
The inner product on $\ReaBCond_n(K)$ is defined as $\left\langle h|k \right\rangle = \operatorname{tr}(h^\dagger k)$.
Note that $h^\dagger k \in \MQ(K,K)$, and $\MQ(K,K)$ is a multi-matrix algebra, and the trace on the multi-matrix algebra is well-defined. 

On morphism level, given $f: K_1 \to K_2$, we define $\ReaBCond_n(f) = f^*$, where $f^*: \ReaBCond_n(K_2) \to \ReaBCond_n (K_1)$ is the map of pre-composition with $f$, i.e. $h \mapsto h \circ f$. 

Inclusions of small pieces into large pieces are defined by $j_{n,n+1}: \ReaBCond_n(K) \to \ReaBCond_{n+1}(K)$.
The \emph{realization functor for boundary conditions} $\ReaBCond: \MQ^{\op} \to \Mod(\Loc)$ is defined as the inductive limit $\varinjlim_n \ReaBCond_n$.
\end{definition}

\begin{remark}
    It is not just a matter of convention that the source of $\ReaBCond$ is $\MQ^\op$ instead of $\MQ$. Taking the opposite category of $\MQ$ is necessary for at least two reasons: we want the realization functor to produce left modules of $\Loc$; consistency with the action of the boundary SymTFT (see Section~\ref{section: bulk-boundary correspondence}). We believe that there are deeper reasons for taking the opposite category in the macroscopic description of topological phases, which is a subject for future investigation.
\end{remark}

\begin{figure}
    \centering
    \input{diagrams/reabcond-functor-1362.tex}
    \caption{Definition of \(\ReaBCond\).}
\end{figure}

\subsection{Classification of boundary conditions}

In this section, we prove that simple boundary conditions of the Q-system model are in one-to-one correspondence with simple objects in \(\MQ\), equivalently in \(\MQ^\op\).

First of all, we need to define the microscopic notion of simple boundary conditions in terms of operator algebras and states. 
The physical intuition is, a simple boundary condition is a pure state that is indistinguishable from the reference state in the bulk, but could host interesting features on the boundary. 
To be precise, we would like to study pure states $\sigma \in S(\mathsf{Loc})$, such that there is some $n>0$, with $\sigma(a) = \rho(a)$ for all $a \in \mathsf{Loc}_{[n, \infty)}$, where $\rho$ is the ground state of the Q-system model defined on an infinite fusion spin chain.

This idea could be further formulated in terms of modules of the quasi-local algebra, as a pure state could be translated by the GNS construction into a simple module; conversely, given a simple module $M$ and any unit vector $v \in M$, $\langle v | - \cdot v\rangle: \Loc \to \mathbb{C}$ is a pure state. 
Therefore, a \emph{simple boundary condition} of a Q-system model may be defined as a simple object $M \in \Mod(\Loc)$ such that there is a unit vector $v \in M$ satisfying $\langle v | a \cdot v\rangle = \rho(a)$ for all $a \in \mathsf{Loc}_{[n, \infty)}$, where $\rho$ is the ground state of the Q-system model.

However, it is necessary to take non-simple boundary conditions into consideration to formulate the bulk-boundary correspondence. We find that a natural condition on boundary conditions is being generated by a finite set localized in some finite region, whose spirit is quite similar to the definition of DHR bimodules \cite{Jones_2024_DHR} (also see Section~\ref{sec:DHR theory}).

\begin{definition}\label{def:boundary-condition}
    Given a Q-system model, $M \in \Mod(\Loc)$, and a finite interval $I$, a unit vector $v \in M$ is said to be localized in $I$ if $\langle v | a \cdot v\rangle = \rho(a)$ for any $a$ supported outside $I$. 
    A \emph{boundary condition} of a Q-system model is an object $M \in \Mod(\Loc)$, such that there is a finite set of unit vectors $\{v_1, \dots, v_r\}$ which generates $M$, and there is some $N>0$ such that each $v_i$ is localized in $[0,N]$. We denote the full subcategory of $\Mod(\Loc)$ consisting of boundary conditions as $\BCond$.
\end{definition}

\begin{remark}
This definition is manifestly invariant under symmetric finite-depth unitary circuits.
Therefore, though all our technical manipulations are based on the Q-system model, the classification result in this section automatically holds for any system in the same phase.
\end{remark}

\begin{lemma}
$\sigma(a) = \rho(a)$ for all $a \in \mathsf{Loc}_{[n, \infty)}$ if and only if $\sigma(m_i^\dag m_i) = 1$ for all $i \in [n, \infty)$.
\end{lemma}

\begin{proof}
We only need to prove the ``if'' part.
Let $a$ be supported in $[n,r]$. Since $\sigma(m_i^\dag m_i)=1$ for all $i\in[n,r]$, Proposition~\ref{prop:ground_state_projector_eigenstate} implies that inserting the corresponding commuting projectors does not change the expectation value of $a$ in the state $\sigma$.
Thus we may screen the interval $[n,r]$ exactly as in the ground state computation, and obtain
\[
    \sigma(a)=\sigma(\operatorname{scr}(a)).
\]
The screened operator is a scalar $\lambda_a\cdot \mathrm{id}$ by Lemma~\ref{lemma:fixed_point_algebra}. The ground state $\rho$ gives the same scalar after screening, hence
\[
    \sigma(a) = \sigma(\lambda_a \cdot \mathrm{id}) = \lambda_a = \rho(a).
\]
\end{proof}

\begin{remark}
    This lemma essentially says that local indistinguishability implies global indistinguishability in the Q-system model. One important reason for the lemma to hold is that there is no nontrivial topology in one dimension. Indeed, local indistinguishability \emph{does not} imply global indistinguishability in two and higher spatial dimensions; this is precisely the core idea of the entanglement bootstrap program \cite{EB1,EB2}.
\end{remark}

Since $\sigma(m_i^\dagger m_i) = 1$ does not hold for some boundary sites, we cannot screen a generic operator to a single-site module morphism.
However, we could still screen a generic operator to a module morphism localized near the boundary:

\begin{definition}\label{def:partial-screening-map}
We define the partial screening map $\operatorname{scr}_n: \mathsf{Loc} \to \mathsf{Loc}_{[0,n+1]}$ as follows: Given some operator $a \in \mathsf{Loc}_{[0,r]}$ where $r>n$, $\operatorname{scr}_n (a) = m_{[n,r+1]} a m_{[n,r+1]}^\dag$.
We define the fixed-point algebra $A_n^{\mathrm{fp}} := \operatorname{im} (\operatorname{scr}_n) \subset \mathsf{Loc}_{[0,n+1]}$.
\end{definition}

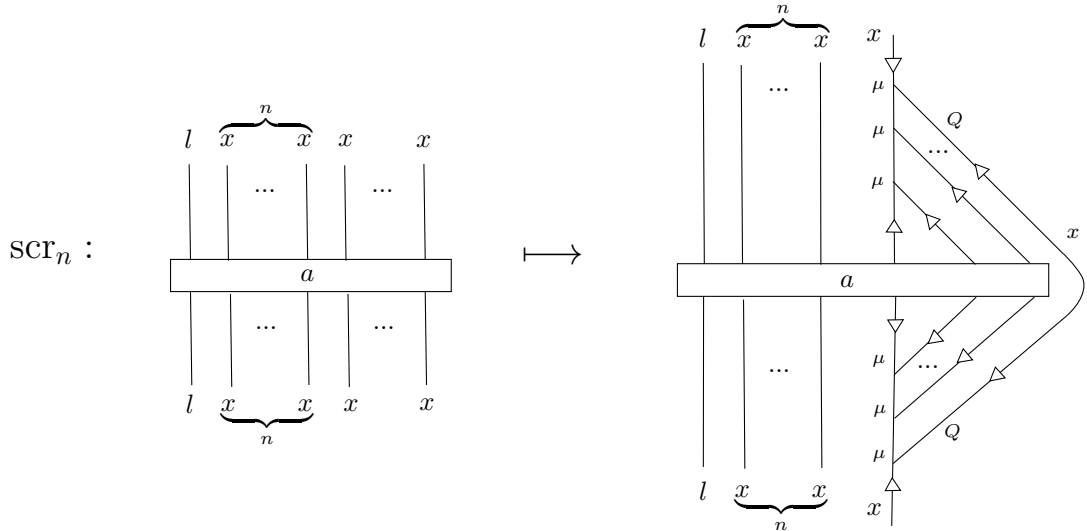
\[
\input{diagrams/partial-screening-2281.tex}
\]

\begin{proposition}
$A_n^{\mathrm{fp}}$ is the subalgebra of operators that commute with the right $Q$-action, i.e. $A_n^{\mathrm{fp}} = \MQ(l \otimes x^{n+1}, l \otimes x^{n+1}) \subset \M(l \otimes x^{n+1}, l \otimes x^{n+1}) = \mathsf{Loc}_{[0,n+1]}$. Note that the right $Q$-module structure of $l \otimes x^{n+1}$ is just the right $Q$-module structure of the rightmost $x$.
\end{proposition}

\begin{proof}
$A_n^{\mathrm{fp}} \subset \MQ (l \otimes x^{n+1}, l \otimes x^{n+1})$ is easily verified by using Frobenius conditions of the Q-system.
To prove $\MQ (l \otimes x^{n+1}, l \otimes x^{n+1}) \subset A_n^{\mathrm{fp}}$, note that given any $f \in \MQ (l \otimes x^{n+1}, l \otimes x^{n+1})$, we have $\operatorname{scr}_n (f) = f$.
\end{proof}

\begin{remark}
It is straightforward to verify that $\operatorname{scr}_n$ is also a conditional expectation.
\end{remark}

We will see that the fixed-point algebra does encode all information of a boundary state, i.e. a boundary state can be completely determined by its action on the fixed-point algebra, which is encapsulated in the next lemma.

\begin{lemma}
Define the set of boundary states to be $S_n^{\mathrm{bdy}} := \{ \phi \in S(\mathsf{Loc}) \big| \phi(m_i^\dag m_i) = 1 \quad \forall i \in [n, \infty) \}$.
Then there is a canonical convex isomorphism $S_n^{\mathrm{bdy}} \simeq S(A_n^{\mathrm{fp}})$.
\end{lemma}
\begin{proof}
    First, we define the convex map $\operatorname{Res} : S_n^{\mathrm{bdy}} \to S(A_n^{\mathrm{fp}})$ as the restriction map: Note that $A_n^{\mathrm{fp}} \subset \mathsf{Loc}_{[0,n+1]} \subset \mathsf{Loc}$, therefore given any $\phi \in S^{\mathrm{bdy}}_n$, we could define $\operatorname{Res} (\phi) = \phi \big|_{A_n^{\mathrm{fp}}}$.

    The inverse map $\operatorname{Ind} : S(A_n^{\mathrm{fp}}) \to S_n^{\mathrm{bdy}}$ is defined in a similar way to the ground state of the Q-system model in Section~\ref{Q-system}: Given any $\psi \in S(A_n^{\mathrm{fp}})$, we define $\operatorname{Ind} (\psi)(a) = \psi (\operatorname{scr}_n(a))$.

    Indeed, if $\phi\in S_n^{\mathrm{bdy}}$, then the defining projector conditions allow us to screen every local operator outside $[0,n+1]$, so $\phi(a)=\phi(\operatorname{scr}_n(a))$ for all local $a$. This gives $\operatorname{Ind}(\operatorname{Res}(\phi))=\phi$.
    Conversely, if $\psi\in S(A_n^{\mathrm{fp}})$ and $b\in A_n^{\mathrm{fp}}$, then $\operatorname{scr}_n(b)=b$, so $\operatorname{Res}(\operatorname{Ind}(\psi))(b)=\psi(b)$.
    Hence $\operatorname{Res}$ and $\operatorname{Ind}$ are inverse to each other.
\end{proof}

\begin{remark}
Physically, the screening map is the renormalization map, and thanks to good properties of the Q-system model, no information is lost during renormalization. In other words, invariance under the screening map could be understood as a fixed-point condition.
\end{remark}

Note that convex isomorphisms preserve extreme points, i.e. pure states, therefore $\operatorname{Res} (\sigma)$ is a pure state defined on $A_n^{\mathrm{fp}}$. 
On the other hand, as $\MQ$ is semisimple, $\MQ (l \otimes x^{n+1}, l \otimes x^{n+1})$ is a multi-matrix algebra, with each matrix component labeled by a simple object $K \in \MQ$. 
Note that a pure state on a multi-matrix algebra must be supported in a single block, therefore there is a simple object $K \in \MQ$ associated with $\operatorname{Res} (\sigma)$.
It turns out we could reconstruct the boundary condition from this data:

\begin{theorem}[Classification of simple boundary conditions]
Suppose $\operatorname{Res} (\sigma)$ is a pure state supported on the block $K$ of $A_n^{\mathrm{fp}}$.
Then $\operatorname{GNS} (\sigma) \simeq \ReaBCond(K)$.
\end{theorem}

\begin{proof}
    Recall that $\ReaBCond(K)$ could be identified with $\operatorname{GNS} (\rho_K)$ by the construction of $\ReaBCond$ in Definition~\ref{Def. The realization functor for boundary conditions}, where $\rho_K$ is the ground state of the boundary interaction $\Phi_{K,Q}$.
Obviously $\rho_K \in S^{\mathrm{bdy}}_n$, thus $\operatorname{Res} (\rho_K)$ is also a pure state defined on $A^{\mathrm{fp}}_n$.
Therefore, to show that $\operatorname{GNS} (\sigma) \simeq \operatorname{GNS} (\rho_K)$, it suffices to show that $\operatorname{GNS} (\operatorname{Res} (\sigma)) \simeq \operatorname{GNS} (\operatorname{Res} (\rho_K))$ as modules of $A^{\mathrm{fp}}_n$.

    Note that given any two pure states on a multi-matrix algebra, they have isomorphic GNS constructions if and only if they are supported on the same matrix block.
Therefore it suffices to show $\operatorname{Res} (\rho_K)$ is also supported on the block $K$ of $A_n^{\mathrm{fp}}$.
Under the identification $A_n^{\mathrm{fp}}=\MQ(l\otimes x^{n+1},l\otimes x^{n+1})$, the simple blocks are labelled by the simple right $Q$-modules appearing in $l\otimes x^{n+1}$. 
The state $\operatorname{Res}(\rho_K)$ is obtained from the vector $\mu_{n+1}^\dagger:K\to l\otimes x^{n+1}$, which lies in the summand labelled by $K$. Therefore $\operatorname{Res} (\rho_K)$ is supported on the block $K$.
\end{proof}

\begin{corollary}
    $\BCond$ is finite semisimple.
\end{corollary}
\begin{proof}
    Consider some $M \in \BCond$, with a finite set of generators $\{v_1, \dots, v_n\}$ localized in $[0,n]$. 
    It suffices to prove that the submodule generated by each $v_i$ is finite semisimple.

    Using the convex isomorphism $S_n^{\mathrm{bdy}} \simeq S(A_n^{\mathrm{fp}})$, $\langle v_i | - \cdot v_i \rangle$ is a state on $A_n^{\mathrm{fp}}$, which must be a finite convex combination of pure states on $A_n^{\mathrm{fp}}$. Therefore, the submodule generated by $v_i$ is isomorphic to a finite direct sum $\oplus_j \ReaBCond(K_{i,j})$, which is finite semisimple.
\end{proof}

Obviously the image of $\ReaBCond$ lies in $\BCond$. From now on, the target of $\ReaBCond$ would always be $\BCond$. As a general boundary condition is a direct sum of simple boundary conditions, we arrive at the following corollary:

\begin{corollary}
    $\ReaBCond: \MQ^{\op} \to \BCond$ is essentially surjective.
\end{corollary}

\begin{remark}
    The method in this section could be applied to show that bulk topological charges in the Q-system model correspond to objects in $\QCQ$.
\end{remark}

\subsection[ReaBCond is an equivalence of categories]{$\ReaBCond$ is an equivalence of categories}\label{subsec: ReaBCond is an equivalence of categories}
In this section we would prove $\ReaBCond$ is also fully faithful, which would imply $\ReaBCond$ is an equivalence between $\MQ^{\op}$ and $\BCond$.
In this way, we give a satisfactory characterization of the category of boundary conditions.

First, since both $\BCond$ and $\MQ$ are semisimple, it suffices to consider simple right $Q$-modules $K_1, K_2 \in \MQ$ and prove that $\ReaBCond$ induces a bijection between $\MQ(K_1, K_2)$ and $\Mod(\mathsf{Loc})(\ReaBCond(K_2), \ReaBCond(K_1))$ (remember that $\ReaBCond$ is contravariant).

When $K_1 \simeq K_2$, we just need to prove $\Mod(\mathsf{Loc})(\ReaBCond(K_2), \ReaBCond(K_1)) \simeq \mathbb{C}$.
Then since $\ReaBCond$ sends isomorphisms to isomorphisms, it must induce a bijection between $\MQ(K_1, K_2) \simeq \mathbb{C}$ and $\Mod(\mathsf{Loc})(\ReaBCond(K_2), \ReaBCond(K_1)) \simeq \mathbb{C}$.

\begin{proposition}
$\Mod(\mathsf{Loc})(\ReaBCond(K_2), \ReaBCond(K_1)) \simeq \mathbb{C}$ when $K_1 \simeq K_2$.
\end{proposition}

\begin{proof}
Recall that given a simple right $Q$-module $K \in \MQ$, $\ReaBCond(K) \simeq \operatorname{GNS} (\rho_K)$, where $\rho_K$ is a pure state by Definition~\ref{Def. The realization functor for boundary conditions}.
The GNS construction of a pure state is a simple module \cite[II.6.4.8]{Blackadar_OperatorAlg_2017}, and the morphism space between two simple objects must be zero or $\mathbb{C}$ by Schur's lemma.
Since $K_1 \simeq K_2$, there is an isomorphism between them.
As $\ReaBCond$ is a functor, it sends isomorphisms to isomorphisms.
thus there is an isomorphism between $\ReaBCond(K_2)$ and $\ReaBCond(K_1)$.
It follows that $\Mod(\mathsf{Loc})(\ReaBCond(K_2), \ReaBCond(K_1)) \simeq \mathbb{C}$.
\end{proof}

When $K_1$ is not isomorphic to $K_2$, we need to prove $\Mod(\mathsf{Loc})(\ReaBCond(K_2), \ReaBCond(K_1)) = 0$.
This could be proved by (again) reducing the problem to fixed-point algebras.

\begin{proposition}
When $K_1 \not \simeq K_2$, $\Mod(\mathsf{Loc})(\ReaBCond(K_2), \ReaBCond(K_1)) = 0$.
\end{proposition}

\begin{proof}
As $\ReaBCond(K) \simeq \operatorname{GNS} (\rho_K)$, we will prove $\Mod(\mathsf{Loc})(\operatorname{GNS} (\rho_{K_2}), \operatorname{GNS} (\rho_{K_1})) = 0$. 
Note that there is a canonical cyclic vector $\xi \in \operatorname{GNS} (\rho_{K_2})$ corresponding to the ground state. 
Our strategy is to prove that given any module morphism $\mu: \operatorname{GNS} (\rho_{K_2}) \to \operatorname{GNS} (\rho_{K_1})$, $\mu(\xi)$ must be a zero vector, by a variant of the screening map.

The cyclic vector satisfies $m_i^\dag m_i \cdot \xi = \xi$ for all $i>0$, and $({m^{K_2}_0})^\dag m^{K_2}_0 \cdot \xi = \xi$.
Therefore, given any module morphism $\mu: \operatorname{GNS} (\rho_{K_2}) \to \operatorname{GNS} (\rho_{K_1})$, we must have $m_i^\dag m_i \cdot \mu(\xi) = \mu(\xi)$ for all $i>0$, and $({m^{K_2}_0})^\dag m^{K_2}_0 \cdot \mu(\xi) = \mu(\xi)$. 
As shown in Corollary~\ref{cor:GNS-im-m-star}, vectors in $\operatorname{GNS} (\rho_{K_1})$ could be identified as morphisms $K \to l \otimes x^{n+1}$ or their infinite sums (with converging norm).

Case 1. $\mu(\xi)$ is some morphism $f: K \to l \otimes x^{n+1}$:
The condition $m_i^\dag m_i \cdot \mu(\xi) = \mu(\xi)$ for all $i>0$ is translated into $m_i^\dag m_i f= f$ for all $i>0$.
Therefore, we have $(\prod_{i=1}^{n} m_i^\dag m_i) ({m^{K_2}_0})^\dag m^{K_2}_0 f = f = (\prod_{i=1}^{n} m_i^\dag)({m^{K_2}_0})^\dag (\prod_{i=1}^{n} m_i) m^{K_2}_0 f$. 
However, $(\prod_{i=1}^{n} m_i) m^{K_2}_0 f \in \MQ(K_1, K_2)$, which must be zero. This implies $f$ is a zero morphism, thus $\mu(\xi)$ must be a zero vector.

\begin{equation}
    \begin{split}
\input{diagrams/module-morphism-2555.tex}
& \quad \in \MQ(K_1,K_2) \quad = \quad 0
\end{split}
\end{equation}

Case 2. $\mu(\xi) = \sum_{k=0}^{\infty} f_k$ is an convergent infinite sum of morphisms $f_k: K \to l \otimes x^{k+1}$: As shown in case 1, we have $(\prod_{i=1}^{n} m_i) m^{K_2}_0 f_i \in \MQ(K_1, K_2) = 0$, then for any $n > 0$,
\[
    \mu(\xi) = (\prod_{i=1}^{n} m_i^\dagger m_i) ({m^{K_2}_0})^\dagger m^{K_2}_0 \mu(\xi) = (\prod_{i=1}^{n} m_i^\dagger m_i) ({m^{K_2}_0})^\dagger m^{K_2}_0 \sum_{k=n}^{\infty} f_k.
\]
Therefore, 
\[
    \|\mu(\xi)\| = \|(\prod_{i=1}^{n} m_i^\dagger m_i) ({m^{K_2}_0})^\dagger m^{K_2}_0 \sum_{k=n}^{\infty} f_k\| \leq \|\sum_{k=n}^{\infty} f_k\|
\]
However, as the infinite sum is convergent, $\lim_{n \to \infty} \|\sum_{k=n}^{\infty} f_k\| = 0$, which implies $\|\mu(\xi)\| = 0$, thus $\mu(\xi) = 0$.
This finishes the proof.
\end{proof}

These two propositions imply that $\ReaBCond$ is fully faithful. Thus we arrive at the theorem:

\LinkedRestatableTarget{target:bcondclassification}\bcondclassification

\begin{remark}\label{rmk:TopDef}
  Using the same technique, one can easily show that the category of bulk topological defects in the Q-system model \emph{without boundary}, denoted by $\mathrm{TopDef}$, is equivalent to $(\QCQ)^{\op}$ \emph{as a linear category}. It seems natural to expect that this equivalence can be promoted to an equivalence \emph{as monoidal categories}. However, at present, we do not yet know how to define the fusion of defects in the operator-algebraic language, nor how to formulate the corresponding monoidal structure on the microscopic side.

Note that the direction of fusion of topological defects is a matter of choice. In other words, one may freely choose the category of bulk topological defects to be monoidally equivalent to either $(\QCQ)^{\op}$ or $(\QCQ)^{\rev,\op}$; the first choice means that defects fuse from left to right, and the second choice means defects fuse from right to left. However, categorical computations in Section~\ref{section: bulk-boundary correspondence} indicate that we should choose $(\QCQ)^{\rev,\op}$ for consistency with conventions of Ref.~\cite{Kong_Yuan_2024_EnrichedMonoidalCat}. 
Also note that our convention differs from those in the literature~\cite{Kong_2022_1DEnrichedCat,Xu_2024_1DPhaseAbelianSym,LanZhou_2024_QuantumCurrent}, where the category of bulk topological defects is taken to be $(\QCQ)^\rev$; this discrepancy is due to different choices of the direction of time. We will return to this orientation issue later in discussions of the bulk SymTFT and the bulk-boundary correspondence.
\end{remark}

\subsection{Examples}
\subsubsection{The Ising chain}

In this section, we study gapped boundaries of gapped phases on the Ising chain.
We will present all categorical data and show how these data can be used to construct concrete Hamiltonians.

As shown in our construction, we need to specify both the gapped phase and the boundary local quantum charge category.
Recall that there are two gapped phases corresponding to the two (Morita classes of) separable algebras \(\one\) and \(\one \oplus e\), as presented in Example~\ref{eg. Z2 gapped phases}.
So our first job would be to find out possible boundary local quantum charge categories and construct fusion spin chains with boundary.

\paragraph{Boundary local quantum charge categories and fusion spin chains with boundary}
Boundary local quantum charge categories are classified by module categories of the bulk charge category.
It is known there are exactly two indecomposable module categories of \(\mathrm{Rep}(\Z_2)\).
The first one is \(\M_1 = \mathrm{Rep}(\Z_2)\) itself (as a linear category), with the module action given by the tensor product in \(\mathrm{Rep}(\Z_2)\); it would be useful later that \(\M_1\) is equivalent to \({}_{\one}\mathrm{Rep}(\Z_2)\), the category of modules of the separable algebra \(\one\).
The second one is \(\M_2 = \ve\) (as a linear category), with the right module action given as follows: \(a \otimes x := a \otimes U(x)\), where \(U: \mathrm{Rep}(\Z_2) \to \ve\) is the forgetful functor, and the tensor product on the right-hand side is the tensor product in \(\ve\).
\(\M_2\) is equivalent to \({}_{\one \oplus e}\mathrm{Rep}(\Z_2)\), the category of modules of the separable algebra \(\one \oplus e\).

Now there are two gapped phases and two boundary local quantum charge categories, so four combinations in total.
We would analyze them one by one.

\paragraph{Trivial phase with symmetric boundary condition}
In this case, the Q-system used to construct the bulk Hamiltonian is \(Q = \one\), and the boundary local quantum charge category is \(\M = \mathrm{Rep}(\Z_2)\).
We take the boundary degree of freedom to be also  \(l=\one \oplus e\), which is a single qubit.

The right \(Q\)-module category is \(\MQ = \mathrm{Rep}(\Z_2)\), so the category of boundary conditions is equivalent to \(\MQ^\op\). Since \(\MQ\) and \(\MQ^\op\) have the same objects, there are two simple boundary conditions given by two simple special Frobenius \(Q\)-modules in \(\M\).

The first one is \(K_1 = \one\), the projection from $l$ to $K_1$ is $p_0=(1+X_0)/2$, with the action morphism \(\mu: K_1 \otimes Q \to K_1\) given by the identity morphism \(1: \one \otimes \one \to \one \).
Then the boundary interaction \(H^{\mathrm{bdy}} = (p_0^\dagger \otimes  p_1^\dagger)(\mu^\dagger \mu)(p_0 \otimes  p_1)\) could be written as the Pauli operator $(1+X_0+X_1+X_0X_1)/4$,
which shares the same ground state as $(X_0+X_1)/2$. Such boundary term projects the boundary to the trivial \(\Z_2\) charge.

The second one is \(K_2 = e\), the projection from $l$ to $K_2$ is $p_0 = (1-X_0)/2$,  the action morphism \(\mu: K_2 \otimes Q \to K_2\) is given by the identity morphism \(1: e \otimes \one \to e \).
Then the boundary interaction \(H^{\mathrm{bdy}} = (p_0^\dagger \otimes  p_1^\dagger)(\mu^\dagger \mu)(p_0 \otimes  p_1)\) is $(1-X_0+X_1-X_0X_1)/4$, which is equivalent to $(-X_0+X_1)/2$, and it projects the boundary to the nontrivial \(\Z_2\) charge.

\paragraph{Trivial phase with symmetry-breaking boundary condition}
In this case, the Q-system used to construct the bulk Hamiltonian is \(Q = \one\), and the boundary local quantum charge category is \(\mathcal{M} =  \ve\).
We take the boundary degree of freedom  \(l=\mathbb{C}^2\in\ve\).

Although the boundary Hilbert space has the same dimension with the previous case, there is one crucial difference: No symmetry is imposed at the boundary.
That is, here we allow both the Pauli \(X\) and Pauli \(Z\) operator at the boundary, while the Pauli \(Z\) operator is forbidden in the previous case since it's not symmetric.

The right \(Q\)-module category is \(\MQ = \ve\), so the category of boundary conditions is equivalent to \(\MQ^\op\).
There is only one simple boundary condition given by the Frobenius \(Q\)-module \(K = \mathbb{C}\) in \(\M\), with action morphism \(\mu: K \otimes Q \to K\) given by the identity morphism \(1: \one \otimes \mathbb{C} \to \mathbb{C} \).  

There are two projections from $l$ to $K$. For $p_0 = (1+X_0)/2$, the boundary term \(H^{\mathrm{bdy}}_1 = (p_0^\dagger \otimes  p_1^\dagger)(\mu^\dagger \mu)(p_0 \otimes  p_1)\sim(X_0+X_1)/2\). For $p_0 = (1-X_0)/2$,   \(H^{\mathrm{bdy}}_2 = (p_0^\dagger \otimes  p_1^\dagger)(\mu^\dagger \mu)(p_0 \otimes  p_1)\sim(-X_0+X_1)/2\). Although these two boundary terms look different, we can use the local unitary $Z_0$ to transform $H^\bdy_1$ to $H^\bdy_2$ and vice versa, so they are in the same phase.  The boundary local quantum charge category $\cM=\ve$ really implies that there is no nontrivial $\Z_2$ charge on the boundary, as we could apply $Z_0$ to create a " \(\Z_2\) charge" locally.

\paragraph{SSB phase with symmetric boundary condition}
In this case, the Q-system used to construct the bulk Hamiltonian is \(Q = \one \oplus e\), and the boundary local quantum charge category is \(\M = \mathrm{Rep}(\Z_2)\), we take the boundary degree of freedom $l=\one\oplus e$.

The right \(Q\)-module category is \(\MQ = \Rep(\Z_2)_{\one\oplus e}\cong \ve\), so the category of boundary conditions is equivalent to \(\MQ^\op\).
There is only one simple boundary condition given by the Frobenius \(Q\)-module \(K = \one \oplus e\) in \(\M\), and the action morphism \(\mu: K \otimes Q \to K\) coincides with the multiplication morphism of \(Q\), which is given by the following matrix:
\begin{equation}
    \begin{array}{ c|c|c|c|c }
    & \one \otimes \one & \one \otimes e & e \otimes \one & e\otimes e\\
    \hline
    \one & 1/\sqrt{2} & 0 & 0 & 1/\sqrt{2}\\
    \hline
    e & 0 & 1/\sqrt{2} & 1/\sqrt{2} & 0
    \end{array}
\end{equation}
Then the boundary interaction \(H^{\mathrm{bdy}} = (p_0^\dagger \otimes  p_1^\dagger)(\mu^\dagger \mu)(p_0 \otimes  p_1)\) could be written as the Pauli operator \((1+Z_0 Z_1 ) / 2\), which is essentially the Ising coupling $Z_0Z_1$.

\paragraph{SSB phase with symmetry-breaking boundary condition}
\label{paragraph: ssb phase with symmetry-breaking boundary condition}
In this case, the Q-system used to construct the bulk Hamiltonian is \(Q = \one \oplus e\), and the boundary local quantum charge category is \(\M = \ve\), and we take the boundary degree of freedom $l=\mathbb{C}^2$.

The right \(Q\)-module category is \(\MQ = \mathrm{Rep}(\Z_2)\), so the category of boundary conditions is equivalent to \(\MQ^\op\).
There are two simple boundary conditions given by two simple special Frobenius \(Q\)-modules in \(\M\).

The first one is \(K_1 = \mathbb{C}\), with the action morphism \(\mu: K_1 \otimes Q \to K_1\) given by the following matrix:
\begin{equation}
    \begin{array}{ c|c|c|}
    & \mathbb{C} \otimes \one & \mathbb{C} \otimes e\\ 
    \hline
    \mathbb{C} & 1/\sqrt{2} & 1/\sqrt{2}\\
    \hline
    \end{array}
\end{equation}
Consider $K_1 = \mathbb{C}=\mathrm{span} \{ |0\rangle \}$ and take the projection $p_0 = (1+Z_0)/2$,
 the boundary term \(H^{\mathrm{bdy}} = (p_0^\dagger \otimes  p_1^\dagger)(\mu^\dagger \mu)(p_0 \otimes  p_1) = (1+Z_0)(1+Z_1)/4\).  
The Hamiltonian is \(H = -(1+Z_0)(1+Z_1) / 4 - \sum_{i \geq 1} (Z_i Z_{i+1} +1)/2 \), and the ground state is \(\prod_{i \geq 0} | 0 \rangle\).

The underlying vector space of the second one is also \(K_2 = \mathbb{C}\), but the action morphism \(\mu: K_2 \otimes Q \to K_2\) is different from the first one:
\begin{equation}
    \begin{array}{ c|c|c|}
    & \mathbb{C} \otimes \one & \mathbb{C} \otimes e\\
    \hline
    \mathbb{C} & 1/\sqrt{2} & -1/\sqrt{2}\\
    \hline
    \end{array}
\end{equation}
Take $p'_0 = (1-Z_0)/2$, then the boundary interaction \(H^{\mathrm{bdy}} = ((p'_0)^\dagger \otimes  p_1^\dagger)(\mu^\dagger \mu)(p'_0 \otimes  p_1)\) is
\((1-Z_0)(1-Z_1)/4\),
the Hamiltonian is \(H = -(1-Z_0)(1-Z_1) / 4 - \sum_{i \geq 1} (Z_i Z_{i+1} +1)/2 \), and the ground state is \(\prod_{i \geq 0} | 1 \rangle\).

Note that the \(Z_1\) term is not allowed when we consider the symmetric boundary condition, as it doesn't commute with the symmetry action.
However, we allow it when considering the symmetry-breaking boundary condition, as non-symmetric operators are now allowed.
This explains the existence of two distinct boundary conditions.

\begin{remark}
\label{remark: support of non-symmetric operators}
The support of non-symmetric operators must contain the boundary, as operators supported in the bulk are still required to be symmetric.
For example, the support of the operator \(Z_i\) is not site \(i\), but the interval \([0,i]\).
As a corollary, \(\prod_{i \geq 0} | 0 \rangle, \prod_{i \geq 0} | 1 \rangle\) and the GHZ state \(\prod_{i \geq 0} | 0 \rangle + \prod_{i \geq 0} | 1 \rangle\) are all indistinguishable in the bulk, since only measuring symmetric operators is allowed.
\end{remark}

The four cases are summarized in the following table:
\begin{table}[h]
\centering
\begin{tabular}{|c|c|c|c|}
\hline
Phase & Boundary local quantum charge category & Boundary conditions & Hamiltonian term \\
\hline

Trivial & \(\mathrm{Rep}(\Z_2)\) & \(K_1 = \one, K_2 = e\) & \((\pm X_0 + X_1) / 2 \) \\
\hline

Trivial & \(\ve\) & \(K = \mathbb{C}\) & \((X_0 + X_1) / 2 \) \\
\hline

SSB & \(\mathrm{Rep}(\Z_2)\) & \(K = \one\oplus e\) & \(Z_0 Z_1 \) \\
\hline

SSB & \(\ve\) & \(K_1 = \mathbb{C}, K_2 = \mathbb{C}\) & \(\pm (Z_0 + Z_1) / 2 + 1\) \\

\hline

\end{tabular}
\caption{Summary of gapped phases and boundaries on the Ising chain}
\label{table: summary of gapped phases and boundaries}
\end{table}

\subsubsection{ $G$-SPT phases for finite group $G$ and boundary modes}
We consider gapped boundaries of \(1+1\)D bosonic systems with finite symmetry group \(G\). The corresponding charge category is \(\mathcal C=\Rep(G)\). It is known that Morita classes of separable algebras in \(\Rep(G)\) are classified by triples \((G,H,\psi)\), where \(H\) is a subgroup of \(G\) and \(\psi\in \mathrm H^2(H,U(1))\), up to conjugation \cite{Ost03}. This agrees with the physical classification of anomaly-free gapped phases in \(1+1\)D with \(G\)-symmetry \cite{Chen_Gu_Wen_2011_SPTMPS,Cirac_2011_1DPhaseMPS}. Physically, \(H\) is the unbroken symmetry subgroup, and \(\psi\) labels an SPT phase protected by \(H\). In particular, \(G\)-SPT phases correspond to \(H=G\), and are classified by \(\mathrm H^2(G,U(1))\).

Following the \(Q\)-system construction of~\cite{YWL_DefectsIn2+1TO}, we systematically construct the \(Q\)-systems corresponding to \(G\)-SPT phases as follows. Choose a normalized cocycle representative \(\psi\in Z^2(G,U(1))\). Let \(\Rep^\psi(G)\) denote the category of \(\psi\)-projective representations of \(G\). For \((M,\tau)\in \Rep^\psi(G)\), we have \(\tau(g)\tau(h)=\psi(g,h)\tau(gh)\). The dual vector space \(M^*=\Hom(M,\mathbb C)\) naturally carries a \(\psi^{-1}\)-projective representation \(\tau^*\), defined by
\[
(\tau^*(g)f)(m)=\psi(g,g^{-1})^{-1} f\bigl(\tau(g^{-1})m\bigr),
\]
for \(g\in G\), \(f\in M^*\), and \(m\in M\). A direct computation shows that \(\tau^*(g)\tau^*(h)=\psi(g,h)^{-1}\tau^*(gh)\). Therefore, the projective twists cancel on \(M\otimes M^*\), and \(\tau\otimes\tau^*\) defines an ordinary \(G\)-representation. Thus \(M\otimes M^*\) is an object of \(\Rep(G)\). Moreover, \(Q=M\otimes M^*\cong \End(M)\) carries a natural \(Q\)-system structure in \(\Rep(G)\), defined as follows: let \(d=\dim M\), and the multiplication and unit are
\[
\mu((m\otimes f)\otimes(n\otimes \ell))
=
\frac{1}{\sqrt d}f(n)m\otimes \ell,
\qquad
\eta(1)=\sqrt d\sum_i e_i\otimes e^i,
\]
where \(m,n\in M\), \(f,\ell\in M^*\), and \(\{e_i\}\) is an orthonormal basis of \(M\) with dual basis \(\{e^i\}\). With respect to the Hilbert--Schmidt dagger on \(Q\cong \End(M)\), the comultiplication is \(\Delta=\mu^\dagger\), explicitly
\[
\Delta(m\otimes f)
=
\frac{1}{\sqrt d}\sum_i (m\otimes e^i)\otimes(e_i\otimes f).
\]
Then \(\mu\circ\Delta=\id_Q\). Together with associativity, unitality, and the Frobenius property inherited from the matrix algebra structure, this makes \(Q\) a special dagger Frobenius algebra, equivalently a \(Q\)-system, in \(\Rep(G)\). Its Morita class corresponds to the \(G\)-SPT phase labeled by \([\psi]\in \mathrm H^2(G,U(1))\).

Note that for any nonzero \((M,\tau)\in \Rep^\psi(G)\), the algebra \(Q=M\otimes M^*\) represents the same Morita class; in particular, this class depends only on \([\psi]\in \mathrm H^2(G,U(1))\).  To make the $Q$-system  $Q=M\ot M^*$ have smaller dimension, we can take $M$ to be a simple object in $\Rep^\psi(G)$. 

\begin{figure}
    \centering
\input{diagrams/edge-mode-2810.tex}
    \caption{Each oval denotes a lattice site, the red dashed lines denote the interactions. $N\in \Rep^\psi(G)$ on the boundary site $N\ot M^*$ is a dangling edge mode.}
\label{fig:SymmetricEdgeModeForGSPT}
\end{figure}
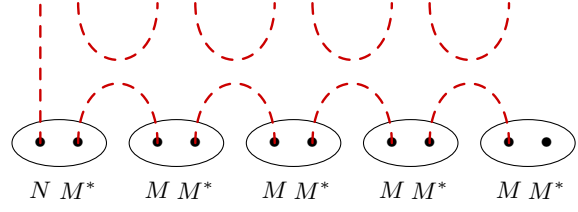

\paragraph{Symmetric boundary condition.}
We now consider the case that the full symmetry is preserved on the boundary, which has the boundary local quantum charge category \(\mathcal M=\Rep(G)\). For the \(Q\)-system \(Q=M\otimes M^*\), the right \(Q\)-module category is \(\mathcal M_Q=\Rep(G)_{M\ot M^*}\), so the category of boundary conditions is equivalent to \(\mathcal M_Q^\op\). We show that symmetric boundary conditions are classified by projective representations of \(G\) with the 2-cocycle $\psi$.

For any object \(N\in \Rep^\psi(G)\), \(N\otimes M^*\) is an object of \(\Rep(G)\), since the projective twists cancel. It has a canonical right \(Q=M\otimes M^*\)-module structure given by evaluation:
\[
    \mu_N\bigl((n\otimes f)\otimes(m\otimes \ell)\bigr)
    =
    \frac{1}{\sqrt{d}} f(m)\, n\otimes \ell ,
\]
where the prefactor comes from the normalization of the \(Q\)-system. This defines a functor
\[
    -\otimes M^*:\Rep^\psi(G)\longrightarrow \Rep(G)_{M\otimes M^*},
\]
which is an equivalence of \(\Rep(G)\)-module categories~\cite{EGNO_2015}. Hence every simple object in \(\Rep(G)_{M\otimes M^*}\) is isomorphic to \(N\otimes M^*\) for some simple object \(N\in\Rep^\psi(G)\).

Physically, if we take \(K=N\otimes M^*\) as a simple boundary condition, the boundary Hamiltonian term associated with the \(Q\)-module action is essentially the evaluation pairing between the right half \(M^*\) of the boundary object \(N\otimes M^*\) and the left half of the bulk object \(M \otimes M^*\). Hence it acts trivially on \(N\), leaving \(N\) as a dangling projective edge mode (see Figure~\ref{fig:SymmetricEdgeModeForGSPT}).  For a nontrivial SPT class \([\psi]\), any simple \(\psi\)-projective representation \(N\) has dimension greater than one, giving rise to a symmetry-protected boundary degeneracy. 

\paragraph{Symmetry breaking boundary condition.}
When the symmetry is completely broken at the boundary, the boundary local quantum charge category is \(\cM=\ve\). Then the right \(Q\)-module category is
   $ \cM_Q=\ve_{M\otimes M^*}.$   
Thus the category of boundary conditions is equivalent to \(\cM_Q^\op\).
Here the \(M\otimes M^*\)-action on objects of \(\MQ\) is obtained as follows: Under the forgetful functor $U:\Rep(G)\to \ve,$
the algebra \(M\otimes M^*\) becomes the matrix algebra $U(M\otimes M^*)\cong \End(M)$; 
objects in \(\MQ\) are ordinary modules of the algebra $U(M\otimes M^*)$.

By the representation theory of matrix algebras,
\[
    \ve_{M\otimes M^*}
    :=
    \ve_{U(M\otimes M^*)}
    \simeq
    \ve_{\End(M)}
    \simeq
    \ve.
\]
Therefore there is a unique simple boundary condition up to isomorphism.

\subsubsection{$\Rep(D_8)$ SPT phases and boundary modes}
In this section we study the boundary modes of the $\Rep(D_8)$ SPT phases, where $D_8$ is the dihedral group of order 8, $D_8:=\langle s,r| s^2=r^4=e, srs=r^{-1}\rangle$.

The bulk charge category associated with the $\Rep(D_8)$ symmetry is given by the canonical Morita dual $\C=\ve_{D_8}$. According to Refs.~\cite{Seifnashri_2024, MYLG25}, there exist three distinct $\Rep(D_8)$ SPT phases, which correspond to three (Morita classes of) separable algebras in $\ve_{D_8}$. These algebras are given by $Q_0 = \mathbb{C}\langle e\rangle$, $Q_1 = \mathbb{C}^\omega\langle r^2, s\rangle$, and $Q_2 = \mathbb{C}^\omega\langle r^2, sr\rangle$. 
Here, $Q_0 = \mathbb{C}\langle e\rangle$ is the trivial algebra. The subsets $\langle r^2, s\rangle$ and $\langle r^2, sr\rangle$ are two distinct $\Z_2 \times \Z_2$ subgroups of $D_8$, and $\omega \in H^2(\Z_2 \times \Z_2, U(1)) \cong \Z_2$ denotes the nontrivial 2-cocycle. The algebras $Q_1$ and $Q_2$ are the twisted group algebras of these $\Z_2 \times \Z_2$ subgroups, where the multiplication is twisted by the 2-cocycle $\omega$. The explicit values of $\omega$ are summarized in Table~\ref{tab.twist}.
\begin{table}[]
    \centering
    \begin{tabular}{|c|c|c|c|c|}
    \hline
        $\omega$ & $(0,0)$ & $(0,1)$&$(1,0)$&$(1,1)$ \\
    \hline
$(0,0)$&$1$&$1$&$1$&$1$\\
         \hline
         $(0,1)$&$1$&$1$&$1$&$1$\\
         \hline
         $(1,0)$& $1$ &$-1$&$1$&$-1$\\
         \hline
         $(1,1)$&$1$ &$-1$&$1$&$-1$\\
         \hline
    \end{tabular}
    \caption{The nontrivial 2-cocycle of $\Z_2\times \Z_2$ with $U(1)$ coefficient. The rows correspond to the first argument $g$ and the columns to the second argument $h$ of the cocycle $\omega(g,h)$.  In the context of two specific subgroups isomorphic to $\mathbb{Z}_2 \times \mathbb{Z}_2$, the generator $r^2$ is identified with $(0,1)$ in both $\langle r^2, s\rangle$ and $\langle r^2, sr\rangle$, while $s$ corresponds to $(1,0)$ in $\langle r^2, s\rangle$ and $sr$ corresponds to $(1,0)$ in $\langle r^2, sr\rangle$.}
    \label{tab.twist}
\end{table}

Note that any indecomposable $\ve_{D_8}$-module $\cM$ yields a valid boundary local quantum charge category. The symmetric boundary condition corresponds to the boundary local quantum charge category $\cM=\ve_{D_8}$, which has already been studied in~\cite{MYLG25}. In the following, we study the fully symmetry-breaking boundary condition (with $\cM=\ve$).

The module action of $\ve_{D_8}$ on $\ve$ is defined by forgetting the $D_8$-grading and then taking the tensor product within $\ve$.   

\paragraph{$Q_0$ with symmetry breaking boundary condition}

Since \(Q_0=\mathbb C\langle e\rangle\) is one dimensional, $\MQ$ has only one simple object. Each bulk term in the Hamiltonian is simply the projection onto \(Q_0\), and the boundary term is the projection onto the one-dimensional space
$K=\mathbb C:=\operatorname{Span}\{v\}$.
Thus the ground state is given by the product state
\[
    \lvert\psi\rangle
    =
    \lvert v \rangle \otimes \lvert e \rangle \otimes \lvert e \rangle \otimes \lvert e \rangle \otimes \cdots .
\]

\paragraph{$Q_1$ and $Q_2$ with symmetry breaking boundary condition}
Upon forgetting the $D_8$-grading structure, $Q_1$ and $Q_2$ become isomorphic to twisted group algebras of $\Z_2\times \Z_2$, $\mathbb{C}^\omega(\Z_2\times \Z_2)$. To compute $\cM_{Q_1}$ and $\cM_{Q_2}$, we  first consider the category of right  $\mathbb{C}^\omega(\Z_2\times \Z_2)$ modules in $\ve$. This category  is equivalent to the category of projective representations of $\Z_2\times \Z_2$, denoted $\Rep^\omega(\Z_2\times \Z_2)$.  We then restore the 
$D_8$-grading structure retrospectively.
Up to isomorphism, $\Rep^\omega(\Z_2\times \Z_2)$ contains exactly one irreducible projective representation $K$ of dimension 2. The twisted right actions  $\mu: K\otimes \mathbb{C}^\omega( \Z_2\times \Z_2)\rightarrow K$, and how elements in $\Z_2\times \Z_2$ are identified with elements in $D_8$ are summarized in Table~\ref{table.Z_2^2action}.
\begin{table}[]
    \centering
    \begin{tabular}{|c|c|c|c|c|}
    \hline
        
        $g\in \Z_2\times \Z_2$ & $(0,0)$ & $(0,1)$ & $(1,0)$ & $(1,1)$  \\
         \hline
   $\mu(g)$   & $\one$ & $X$& $Z$& $XZ$\\
   \hline
    \end{tabular}
    \caption{The twisted right actions on the irreducible  $2d$ $\Z_2\times \Z_2$ representation. $X$, $Z$ are Pauli operators. $r^2$ in both $Q_1$ and $Q_2$ is identified with $(0,1)$; while $s$ in $Q_1$ corresponds to $(1,0)$, and $sr$ in $Q_2$ corresponds to $(1,0)$. }
    \label{table.Z_2^2action}
\end{table}

Omitting the projectors, the effective low energy boundary term is $H^\bdy=  \mu^\dagger \mu$. We can express its components by fixing the incoming and outgoing $\Z_2\times \Z_2$ elements, as shown in Table~\ref{tab.RepD8_SSB_bdy_term}.
\begin{table}[]
    \centering
 \begin{tabular}{|c|c|c|c|c|}
\hline
\diagbox{$g$}{$h$} & $(0,0)$ & $(0,1)$ & $(1,0)$ & $(1,1)$ \\
\hline
$(0,0)$      & $I$  & $X$  & $Z$  & $ZX$  \\
\hline
$(0,1)$      & $X$  & $I$  & $ZX$ & $Z$   \\
\hline
$(1,0)$      & $Z$  & $XZ$& $I$  & $-X$   \\
\hline
$(1,1)$     & $XZ$& $Z$ & $-X$ & $I$   \\
\hline
\end{tabular}
    \caption{$\mu(h)^\dagger \mu(g)$. $r^2$ in both $Q_1$ and $Q_2$ is identified with $(0,1)$; while $s$ in $Q_1$ corresponds to $(1,0)$, and $sr$ in $Q_2$ corresponds to $(1,0)$.}
    \label{tab.RepD8_SSB_bdy_term}
\end{table}

  \begin{remark}
    The reader might have noticed that all SPT phases, protected by either group symmetry or non-invertible symmetry, have only one simple boundary condition when we take the boundary local quantum charge category $\M = \ve$.
    Indeed when $\M=\ve$, boundary conditions and ground state degeneracy are related in a particularly simple way: the number of simple objects in $\ve_Q$ agrees with the ground state degeneracy of the bulk phase labeled by $Q$ after forgetting the symmetry via the same fiber functor that defines the $\C$-module structure on $\ve$.

    To see this, recall that a $\C$-module structure on $\ve$ is equivalent to a fiber functor
    \[
        U:\C\to \ve
    \]
    \cite[Chapter 7]{EGNO_2015}. Therefore, for a Q-system $Q\in \C$, the category $\ve_Q$ can be identified with the category of ordinary right modules over the algebra $U(Q)$ in $\ve$. Since $Q$ is separable, $U(Q)$ is a separable algebra in $\ve$. Hence $U(Q)$ is a finite-dimensional semisimple algebra, and therefore isomorphic to a multi-matrix algebra
    \[
        U(Q)\cong \bigoplus_{i=1}^n \operatorname{Mat}_{d_i}(\mathbb{C}),
    \]
    where $n$ is the number of simple matrix blocks and $d_i$ is the size of the $i$-th block. It follows that $\ve_Q$ has precisely $n$ simple objects.

    On the other hand, if the symmetry is forgotten via the same fiber functor $U:\C\to \ve$, then the corresponding bulk fixed-point algebra is
    \[
        \operatorname{End}_{U(Q)\text{-}U(Q)}(U(Q)),
    \]
    the algebra of $U(Q)$-$U(Q)$ bimodule endomorphisms of $U(Q)$. This algebra is naturally isomorphic to the center $Z(U(Q))$. Since
    \[
        Z(U(Q)) \cong \oplus^n\mathbb{C},
    \]
    its dimension is also $n$. It is proved in~\cite{MYLG25} that $\dim Z(U(Q))$ is precisely the number of extremal points of the simplex of $H$-type ground states, which is the ground state degeneracy. 
    Thus, with respect to the same fiber functor $U$, the number of simple objects in $\ve_Q$ agrees with the ground state degeneracy after forgetting the symmetry.
\end{remark}

%--------------------------------------------------------

\section{DHR theory of fusion spin chains with boundary}\label{sec:DHR theory}

In this section, our goal is to define the boundary SymTFT as the category of boundary DHR bimodules and compute its structure as a fusion category.
Before delving into technicalities of Hilbert bimodules and operator algebras, we would first review the intuitive computations in \cite{Kong_2022_1DEnrichedCat}, where elements in the SymTFT are defined as non-local operators.
Though non-local operators are ill-defined, they do capture the key physical intuition about the SymTFT: Operator charges that could act on topological defects.
We then illustrate how attempts to formulate non-local operators rigorously lead to the concept of Hilbert bimodules.
We recommend the reader to come back here whenever feeling uncertain about physical meanings of DHR bimodules.

\subsection{From non-local operators to Hilbert bimodules}
\label{subsection: non-local ops to hilb bimods}

\paragraph{Bulk SymTFT as non-local operators}
On the Ising chain (infinite in both directions), it is argued that there are four sectors of non-local symmetric operators, which could be written down as follows: the identity operator \(\one\); \(m_i=\otimes_{k \geq i} X_k\); \(Z_i=\otimes_{k \geq i}\left(Z_k Z_{k+1}\right)\); and \(m_i Z_i\).
There is a "tensor product" of non-local operators defined by composition.
Obviously the 'fusion rule' of non-local operators are identical to simple objects in the toric code, namely
\begin{equation}
m_i \otimes m_i = Z_i \otimes Z_i = (m_i Z_i) \otimes (m_i Z_i) = \one, \quad m_i \otimes Z_i = Z_i \otimes m_i = m_i Z_i.
\end{equation}
There is also a "braiding" of non-local operators defined by commutation relations.
For example, the double braiding between \(m_i\) and \(Z_i\) is \(-1\).
In conclusion, these non-local operators should be understood as simple objects in the SymTFT described by \(Z_1(\Rep(\Z_2)) \simeq \mathbf{TC}\).

\paragraph{Topological holography via non-local operators}
What is topological holography?
Indeed there are various aspects of the concept, and an important one is that \emph{\(n\)-dimensional gapped phases are controlled by the SymTFT, which is \((n+1)\)-dimensional}.
We will illustrate that topological holography could be realized as follows: Non-local operators could act on topological charges of any gapped phase, which would give the full classification data of the phase.
Consider the Hamiltonian of the trivial phase
\begin{equation}
H_1 = -\sum_{i \in \mathbb{Z} } X_i,
\end{equation}
and the ground state is
\begin{equation}
|\psi_1\rangle=|\cdots++\cdots\rangle.
\end{equation}
Define a topological charge to be an excitation that cannot be created individually by \emph{symmetric} local unitary operators.
There are two topological charges in the trivial phase: One is the trivial charge, and the other is the \(Z_2\) symmetry charge.
Apparently, applying the non-local operators \(Z_i\) and \(m_i Z_i\) on the ground state yields a \(Z_2\) charge, and applying \(\one\) or \(m_i\) does not.
In other words, the trivial phase behaves like the \(m\)-condensed boundary of the SymTFT.
Similarly, consider the Hamiltonian of the symmetry breaking phase
\begin{equation}
H_2 = -\sum_{i \in \mathbb{Z} } Z_i Z_{i+1},
\end{equation}
and the unique symmetry-invariant ground state is
\begin{equation}
|\psi_2\rangle = \frac{1}{\sqrt 2} \left( |\cdots 000 \cdots\rangle + |\cdots 111 \cdots\rangle \right).
\end{equation}
Note that the relative phase between two summands is unphysical, as it cannot be detected by any symmetric operator.
There are two topological charges in the symmetry-breaking phase: One is the trivial charge, and the other is the domain wall.
Apparently, applying the non-local operators \(m_i\) and \(m_i Z_i\) on the ground state yields a domain wall, and applying \(\one\) or \(Z_i\) does not.
In other words, the symmetry breaking phase behaves like the \(e\)-condensed boundary of the SymTFT.

\paragraph{Boundary SymTFT as non-local operators}
There is a similar story for spin chains with boundary.
However, there is one more choice to make: The boundary symmetry, or the boundary local quantum charge category.
In the case of Ising chain, we could choose from the symmetric boundary, with boundary local quantum charge category \(\Rep(\Z_2)\), or the symmetry-breaking boundary, with boundary local quantum charge category \(\ve\).
We emphasize that the boundary local quantum charge category is the more natural viewpoint, as when the input UFC \(\mathcal{C}\) doesn't admit a fiber functor, there is no notion of symmetry action at all!

We first consider the symmetric boundary.
Consider the right-infinite Ising chain, i.e. there is a spin for each \(n \in \mathbb{Z}_{\geq 0}\).
In this case, the boundary SymTFT behaves precisely as the \(m\)-condensed boundary of the bulk SymTFT, where only two sectors of non-local operators survive: One is the trivial operator \(\one\), and the other is the spin-flip operator \(Z_i=\otimes_{k \geq i}\left(Z_k Z_{k+1}\right)\).
The boundary SymTFT is described by the fusion category \(\Rep(\Z_2)\).

Now consider the symmetry-breaking boundary.
There is a spin for each \(n \in \mathbb{Z}_{\geq 0}\), but there is one more important difference: We allow non-symmetric operators near the boundary.
(This is the meaning of symmetry-breaking.)
In this case, the boundary SymTFT behaves precisely as the \(e\)-condensed boundary of the bulk SymTFT, where only two sectors of non-local operators survive: One is the trivial operator \(\one\), and the other is \(m_i=\otimes_{k \geq i} X_k\).
The boundary SymTFT is described by the fusion category \(\ve_{Z_2}\).

\paragraph{Non-local operators, QCAs, and Hilbert bimodules}
The most obvious problem of this discussion is that non-local operators are ill-defined.
Furthermore, there is no intuition about how to define direct sums of them, but \(Z_1(\Rep(\Z_2))\) has direct sums.
Fortunately, this problem could be fixed by a two-step procedure: We first formulate non-local operators rigorously in terms of \emph{quantum cellular automata} (QCA), then promote QCAs to Hilbert bimodules via the \emph{modulation construction}.

The first observation is, though a non-local operator is ill-defined, conjugation of local operators by non-local operators are well-defined.
Therefore we could switch to the Heisenberg picture and consider how "non-local operators" act on the quasi-local algebra by conjugation, which gives rise to \(*\)-automorphisms of the quasi-local algebra.
For example, conjugation by \(m_i\) has the effect \(m_i (Z_{i-1} Z_i) m_i^{-1} = - Z_{i-1} Z_i\), and conjugation by \(Z_i\) has the effect \(Z_i (X_i) Z_i^{-1} = -X_i\).
Note that there is a set of obvious generators of the symmetric quasi-local algebra: \(\{ X_i, Z_i Z_{i+1} \}_{i \in \mathbb{Z}}\); they generate the whole symmetric quasi-local algebra by multiplication and summation.
Therefore, to describe algebra automorphisms, it suffices to describe how they act on the generators.
The four automorphisms corresponding to the four non-local operators could be written down as follows: the identity automorphism \(\alpha_\one\); \(\alpha_{m_i}\), which sends \(Z_{i-1} Z_i\) to \(-Z_{i-1} Z_i\) and sends all other generators to themselves; \(\alpha_{Z_i}\), which sends \(X_i\) to \(-X_i\) and sends all other generators to themselves; and \(\alpha_{m_i Z_i} = \alpha_{m_i} \alpha_{Z_i}\).
It is apparent that these automorphisms all preserve locality of operators.
In other words, they are all \emph{quantum cellular automata} (QCA).

The next question is how QCAs act on quantum states.
Since QCAs are defined via the quasi-local algebra \(\Loc\), we should also understand quantum states as positive linear functionals on \(\Loc\).
Then a QCA \(\alpha: \Loc \to \Loc\) naturally acts on a state \(\phi: \Loc \to \mathbb{C}\) by composition, i.e. \(\alpha(\phi) := \phi \circ \alpha\).
This action is identical to applying the non-local operator to the quantum state, since it just translated the latter to the Heisenberg picture.

The second observation is, any \(*\)-homomorphism of \(C^*\)-algebras gives rise to a Hilbert bimodule, and Hilbert bimodules have a natural notion of direct sum.

The \emph{modulation construction} converts a \(*\)-homomorphism \(f: A \to B\) to a Hilbert bimodule \(\operatorname{Md} (f)\) as follows.
The underlying vector space is \(B\) itself.
The left action is \(a \cdot x = f(a)x\) for \(a \in A, x \in B\), and the right action is \(x \cdot b = xb\).
Finally, the sesquilinear form \(B \times B \to B\) is defined as \(\left\langle x|y \right\rangle = x^* y\).
Moreover, the modulation construction preserves composition in the following sense:
\begin{proposition}
If \(f: A \to B\) and \(g: B \to C\) are \(*\)-homomorphisms between \(C^*\) algebras, then \(\operatorname{Md} (g \circ f) \simeq \operatorname{Md} (f) \boxtimes_B \operatorname{Md} (g)\).
\end{proposition}
Therefore, composition of non-local operators are faithfully translated into tensor product of Hilbert bimodules (via QCAs as an intermediate step).

How do Hilbert bimodules act on a quantum state?
Note that the GNS construction converts a quantum state to a left module, and Hilbert bimodules naturally act on left modules by relative tensor product.
Furthermore, the GNS construction also preserves composition:
\begin{proposition}
If \(f: A \to B\) is a \(*\)-homomorphism between \(C^*\) algebras, and \(g: B \to \mathbb{C}\) is a state, then \(\operatorname{GNS} (g \circ f) \simeq \operatorname{Md} (f) \boxtimes_B \operatorname{GNS} (g)\).
\end{proposition}
Therefore, action of non-local operators on quantum states are faithfully translated into tensor product of Hilbert bimodules with left modules.

Note that boundary non-local operators could also be translated into Hilbert bimodules in the same way.

Unfortunately, QCAs cannot produce the whole SymTFT in general: When we consider non-abelian group symmetries, there are non-invertible objects in \(Z_1(\mathrm{Rep}(G))\), but modulating QCAs must produce invertible Hilbert bimodules.
To obtain the correct description of the SymTFT, we must consider more general Hilbert bimodules, which turn out to be DHR bimodules.

\subsection{The DHR category of fusion spin chains without boundary}
\label{subsec:jones_review}

The heuristic correspondence between non-local operators and Hilbert bimodules is placed on rigorous footing by the discrete DHR theory developed by Jones \cite{Jones_2024_DHR}. 
In this framework, the non-local operators are formulated as DHR bimodules.

\begin{definition}\label{def:bulk-dhr-localized-basis}
Consider an abstract spin chain without boundary $\Loc^{\bulk}_\bullet$. Given some $I \in \mathcal{I}$ and $M \in \mathrm{Bim}(\Loc^{\bulk})$, we say that a vector $x \in M$ is \emph{localized in} $I$ if
\[
    b x = x b \quad \text{for all } b \in \Loc^{\bulk}_{I^c}.
\]
A projective basis $\{ s_i \}$ is \emph{localized in} $I$ if each $s_i$ is localized in $I$.
\end{definition}

\begin{definition}[\cite{Jones_2024_DHR}, Def. 3.2]
$M \in \mathrm{Bim}(\Loc^{\bulk})$ is called a \emph{DHR-bimodule}, if there exists some $r \in \mathbb{Z}_+$ such that for any interval $I$ with length greater than $r$, there is a projective basis localized in $I$.
The category $\DHR(\Loc^{\bulk}_\bullet)$ is the full subcategory of $\mathrm{Bim}(\Loc^{\bulk})$ consisting of DHR bimodules.
\end{definition}

Note that DHR bimodules are defined for any abstract spin chain.

Physically, a bimodule $M$ localized in $I$ represents a non-local operator that is indistinguishable from the identity sector outside the region $I$, and the requirement that the bimodule could be localized everywhere reflects the fact that non-local operators are mobile.
The collection of all such bimodules that are localizable in some finite region forms the \emph{DHR category} of the net, denoted $\DHR(\Loc^{\bulk}_\bullet)$. This category is naturally a \(C^*\)-tensor category where the monoidal product is the relative tensor product $\otimes_{\Loc^{\bulk}}$ of bimodules. 

Under the assumption of (weak) algebraic Haag duality---which ensures that operators commuting with the complement of a region are contained within that region---the category $\DHR(\Loc^{\bulk}_\bullet)$ becomes a braided monoidal category \cite[Theorem B]{Jones_2024_DHR}. This braiding captures the commutation relations of non-local string operators when their endpoints are "slid" past one another on the lattice.

The central result for fusion spin chains built from the unitary fusion category $\C$ is the identification of the DHR category with the Drinfeld center of the category of local quantum charges, $Z_1(\C)$:

\begin{theorem}[\cite{Jones_2024_DHR}, Theorem C]
Let $\Loc^{\bulk}_\bullet$ be the net of symmetric operators of a 1D fusion spin chain associated with a unitary fusion category $\C$ and a strong tensor generating object $x \in \C$. Then there is an equivalence of braided \(C^*\)-tensor categories:
\begin{equation}
    \DHR(\Loc^{\bulk}_\bullet) \cong Z_1(\C).
\end{equation}
\end{theorem}

Furthermore, it is shown that symmetric quantum cellular automata (QCA) acts naturally on the category of DHR bimodules as follows: Since any bounded-spread automorphism $\alpha \in \mathrm{QCA}(\Loc^{\bulk}_\bullet)$ maps localizable bimodules to localizable bimodules, it induces a braided autoequivalence $\DHR(\alpha)$ of $Z_1(\C)$.
This map descends to a group homomorphism
\begin{equation}
    \DHR: \mathrm{QCA}(\Loc^{\bulk}_\bullet) / \mathrm{FDQC}(\Loc^{\bulk}_\bullet) \longrightarrow \mathrm{Aut}_{\DHR(\Loc^{\bulk}_\bullet)} \simeq \mathrm{Aut}_{\mathrm{br}}(Z_1(\C)),
\end{equation}
where $\mathrm{FDQC}(\Loc^{\bulk}_\bullet)$ denotes finite-depth quantum circuits, which are shown to act trivially on the DHR category \cite[Theorem A]{Jones_2024_DHR}.

The physical intuition is, this homomorphism tracks how dualities in the SymTFT permute non-local operators.
As non-local operators control gapped phases (see the discussion of topological holography in Section~\ref{subsection: non-local ops to hilb bimods}), this implies the group homomorphisms see how dualities (QCAs) permute different phases.
Besides, DHR bimodules serve as a powerful invariant for the classification of symmetric quantum cellular automata (QCA).  

In the convention used in this paper, this identification is applied to the oppositely oriented chain. Equivalently, the bulk DHR category is
\begin{equation}
\DHR(\Loc^{\bulk}_\bullet) \cong Z_1(\C^\rev)\cong \overline{Z_1(\C)}.
\end{equation}
The difference from the convention of \cite{Jones_2024_DHR} is only an orientation reversal; see Appendix~\ref{app:orientation:dhr}.

\subsection{The DHR category of spin chains with boundary}
In this section, we extend the framework of DHR bimodules to spin chains with boundary. 
Fix some abstract spin chain with boundary $\Loc^{\bdy}_\bullet$.
Now we define the DHR category of $\Loc^{\bdy}_\bullet$ and show it is a monoidal category.

\begin{definition}\label{def:boundary-dhr-localized-basis}
Given some $I \in \mathcal{I}^+$ and $M \in \mathrm{Bim}(\Loc^{\bdy})$, we say that a vector $x \in M$ is \emph{localized in} $I$ if
\[
    b x = x b \quad \text{for all } b \in \Loc^{\bdy}_{I^c}.
\]
A projective basis $\{ s_i \}$ is said to be \emph{localized in} $I$ if each $s_i$ is localized in $I$.
\end{definition}

\begin{definition}\label{def:boundary-dhr-bimodule}
$M \in \mathrm{Bim}(\Loc^{\bdy})$ is called a \emph{DHR-bimodule} if it is semisimple, and there exists some $r \in \mathbb{Z}_+$ such that there is a projective basis localized in region $[0,r]$.
The category $\DHR(\Loc^{\bdy}_\bullet)$ is the full subcategory of $\mathrm{Bim}(\Loc^{\bdy})$ consisting of DHR bimodules.
\end{definition}

\begin{proposition}
$\DHR(\Loc^{\bdy}_\bullet)$ is a monoidal category.
\end{proposition}
\begin{proof}
First, we show $\DHR(\Loc^{\bdy}_\bullet)$ is closed under relative tensor product.
Suppose $M,N \in \DHR(\Loc^{\bdy}_\bullet)$, $\{ s_i \}$ is a projective basis of $M$ localized in $[0, r_1]$, $\{ t_j \}$ is a projective basis of $N$ localized in $[0, r_2]$, then $M \boxtimes_{\Loc^{\bdy}} N$ is still semisimple, and $\{ s_i \otimes t_j \}$ is a projective basis of $M \boxtimes_{\Loc^{\bdy}} N$ localized in $[0, \max (r_1,r_2)]$, so $M \boxtimes_{\Loc^{\bdy}} N$ is still a DHR bimodule.

To find the tensor unit, note that $\Loc^{\bdy}$ is a semisimple bimodule since $\Loc^{\bdy}$ is a simple $\text{C}^*$-algebra, and there is a natural isomorphism $\Loc^{\bdy} \otimes_{\Loc^{\bdy}} M \simeq M$ for $M \in \mathrm{Bim}(\Loc^{\bdy})$.

Finally, the obvious map $(L \boxtimes_{\Loc^{\bdy}} M) \boxtimes_{\Loc^{\bdy}} N \simeq L \boxtimes_{\Loc^{\bdy}} (M \boxtimes_{\Loc^{\bdy}} N)$ is a natural bimodule intertwiner, and it is straightforward to verify it satisfies the pentagon and triangle consistency equations.
\end{proof}

\subsection{Computation of $\DHR(\Loc^{\bdy}_\bullet)$}\label{ComputationDHRB}

In this subsection, we would prove that $(\CMdual)^\rev \simeq \DHR(\Loc^{\bdy}_\bullet) $ by showing that the standard action constructed in Section~\ref{StandardAction}, $\Std \colon (\CMdual)^\rev \to \mathrm{Bim}(\Loc^{\bdy})$ restricts to a monoidal equivalence $\ReaDHR \colon (\CMdual)^\rev \to \DHR(\Loc^{\bdy}_\bullet)$. 
The reader could refer to figure~\ref{fig:Std and ReaDHR} to review its definition.
On the other hand, as there is a canonical isomorphism $\WCW \simeq \CMdual$ after translating the left-module convention of \cite[Remark 7.12.5]{EGNO_2015} to our right-module convention, $\ReaDHR$ coincides with the functor $G \colon \WCW \to \mathrm{Bim}(|Q|)$ constructed in Theorem~\ref{thm:functor-G-fully-faithful}. This comparison with $G$ is made after that convention translation; the monoidal source of $\ReaDHR$ remains $(\CMdual)^\rev$. See Appendix~\ref{app:orientation:right-modules}.

\begin{theorem}
The image of $\Std$ lies in $\DHR(\Loc^{\bdy}_\bullet)$.
\end{theorem}
\begin{proof}
Since $\Std$ is fully faithful, $\Std(L) \in \DHR(\Loc^{\bdy}_\bullet)$ must be semisimple for any $L \in \CMdual$.
It remains to show $\Std(L)$ has a projective basis localized in some interval $[0,r]$.

We first show that $\Std_1(L)$ has a projective basis $S = \{ s_i \}$.
Since $\M$ is semisimple, we have direct sum decompositions $l \otimes x \simeq \oplus_{m \in \operatorname{Irr}(\M)} p_m m$ and $L \otimes l \otimes x \simeq \oplus_{m \in \operatorname{Irr}(\M)} q_m m$.
Now $\M(l \otimes x, L \otimes l \otimes x)$ is identified with the multi-matrix $\oplus_{m \in \operatorname{Irr}(\M)} \mathbb{M}(p_m, q_m)$, and $\M(L \otimes l \otimes x, L \otimes l \otimes x)$ is identified with the multi-matrix algebra $\oplus_{m \in \operatorname{Irr}(\M)} \mathbb{M}(q_m, q_m)$; from elementary linear algebra, we could find $\{ s_i \in \M(l \otimes x, L \otimes l \otimes x) \}$ such that
\[
    \sum_i\left|s_i\right\rangle_{\Loc^{\bdy}_{[0,1]}}\left\langle s_i\right| = \sum_i s_i \circ s_i^*=1_{L \otimes l \otimes x}.
\]
Next we show the image of $S$ is a projective basis of $\Std(L)$.
If $S$ is a projective basis of $\ReaDHR_n(L)$, then $k^L_n(S)$ is a projective basis of $\ReaDHR_{n+1}(L)$, since
\[
    \sum_i (s_i \otimes 1_x) \circ (s_i \otimes 1_x)^* = 1_{L \otimes l \otimes x^{n+1}}.
\]
By induction, the image of $S$ is a projective basis of all $\ReaDHR_{n}(L)$.
It follows the image of $S$ is a projective basis of the inductive limit $\Std(L)$.
It is obvious that each $s_i$ is localized in $[0,1]$.
\end{proof}

In light of this theorem, from now on we change our notation from $\Std$ to $\ReaDHR$ whenever we are dealing with DHR categories as the target.
Since we have shown that $\ReaDHR \colon (\CMdual)^{\rev} \to \DHR(\Loc^{\bdy}_\bullet)$ is fully faithful by Theorem~\ref{FullyFaithful}, it remains to prove it is essentially surjective.

\begin{theorem}
$\ReaDHR$ is essentially surjective.
\end{theorem}
\begin{proof}
Take some $M \in \DHR(\Loc^{\bdy}_\bullet)$; we need to show $M$ is inside the replete image of $\ReaDHR$.
Since $M$ is a DHR bimodule, we have some $r \in \mathbb{Z}_+$ such that $M$ has a projective basis $\{ s_i \}$ localized in $[0,r]$.

Consider the action $G \colon \C \to \mathrm{Bim}(\Loc^{\bdy}_{(r, \infty)})$ defined by $G_n(y) = \C(x^n, y \otimes x^n)$; it is precisely the standard action for $\M = \C$.
Define $Q = [l \otimes x^r, l \otimes x^r]$, which is an algebra in $\C$.

First we show $\Loc^{\bdy} \simeq |Q|$.
By \cite[Theorem 4.6]{Chen_2024_InductiveLimitAF}, $|Q| \simeq \operatorname{colim} |Q|_n$, where $|Q|_n$ denotes the realization with respect to the functor $G_n$.
Therefore, it suffices to prove isomorphisms of the finite pieces $\Loc^{\bdy}_{[0,n+r]} \simeq |Q|_n$.

Note that we have
\[
    Q \simeq \bigoplus_{y \in \operatorname{Irr} (\C )} \M (l \otimes x^r \otimes y, l \otimes x^r) \otimes y
\]
where $\M (l \otimes x^r \otimes y, l \otimes x^r) \otimes y := y^{\oplus \operatorname{dim} \M (l \otimes x^r \otimes y, l \otimes x^r)}$.

Therefore
\begin{align*}
    |Q|_n &\simeq \bigoplus_{y \in \operatorname{Irr} (\C )} \M (l \otimes x^r \otimes y, l \otimes x^r) \otimes \ReaDHR(y) \\
    &= \bigoplus_{y \in \operatorname{Irr} (\C )} \M (l \otimes x^r \otimes y, l \otimes x^r) \otimes \C (x^n, y \otimes x^n)
\end{align*}
which is manifestly isomorphic to $\Loc^{\bdy}_{[0,n+r]}$.

By Theorem~\ref{ImageOfG}, it suffices to show that $M$ is within $\mathrm{Bim}(|Q|, \C)$, i.e., the restriction of $M$ to $\Loc^{\bdy}_{(r, \infty)}$ is within the image of standard action $\ReaDHR$.

Note that $\Loc^{\bdy} = |Q|$, and $|Q| = \Std(Q)$ as a $\Loc^{\bdy}_{(r, \infty)}$-bimodule, so $\Loc^{\bdy}$ decomposes as a $\Loc^{\bdy}_{(r, \infty)}$-bimodule in the following form:
\[
    \Loc^{\bdy} = \bigoplus_{y \in \operatorname{Irr} (\C )} \ReaDHR(y)^{\oplus N_y}, \quad N_y = \operatorname{dim} \M (l \otimes x^r \otimes y, l \otimes x^r).
\]
Next we shall decompose $M$ into simple $\Loc^{\bdy}_{(r, \infty)}$-bimodules.
For each $y$ and $1 \leq i \leq N_y$, let $\ReaDHR(y)_i$ denote the $i$-th copy of $\ReaDHR(y)$ in this decomposition.
\begin{lemma}
$s_k \vartriangleleft \ReaDHR(y)_i$ is a $\Loc^{\bdy}_{(r, \infty)}$-submodule of $M$.
\end{lemma}
\begin{proof}
First, it is closed under right action by $\Loc^{\bdy}_{(r, \infty)}$, since $\ReaDHR(y)_i$ is closed under right action by $\Loc^{\bdy}_{(r, \infty)}$.
It is also closed under left action by $\Loc^{\bdy}_{(r, \infty)}$, since $s_k$ is localized in $[0,r]$, so for any $b \in \Loc^{\bdy}_{(r, \infty)}$, $b \vartriangleright (s_k \vartriangleleft \ReaDHR(y)_i) = (s_k \vartriangleleft b) \vartriangleleft \ReaDHR(y)_i = s_k \vartriangleleft (b \vartriangleright \ReaDHR(y)_i)$, and $\ReaDHR(y)_i$ is closed under left action.
\end{proof}

\begin{corollary}
There is an algebraic module intertwiner $\alpha_{k,y,i} \colon \ReaDHR(y)_i \to s_k \vartriangleleft \ReaDHR(y)_i$ via $\xi \mapsto s_k \vartriangleleft \xi$.
Furthermore, $\ReaDHR(y)_i$ is simple since $\ReaDHR$ is fully faithful, so $\alpha_{k,y,i}$ must be either zero or an isomorphism.
\end{corollary}

Finally, since $\{ s_k \vartriangleleft \ReaDHR(y)_i \}_{k,y,i}$ generates $M$, it follows that every simple summand of $M$ must be some $\ReaDHR(y)$.
\end{proof}

Putting these theorems together, we arrive at the following theorem:
\LinkedRestatableTarget{target:boundarysymtft}\boundarysymtft

As an immediate corollary, we have the \emph{charge-operator duality}, namely we could not only construct a fusion spin chain with boundary $\Loc^{\bdy}_{\bullet}$ from a fusion category $\C$ encoding local quantum charges, but also recover $\C^{\rev}$ from $\Loc^{\bdy}_{\bullet}$:
\begin{corollary}
Given a fusion spin chain with boundary $\Loc^{\bdy}_{\bullet} \colon \mathcal{I} \to \text{C}^*\text{-alg}$ constructed from a fusion category $\C$ with boundary local quantum charge category $\M=\C$, denote $(\Loc^{\bdy})^-_\bullet \colon \mathcal{I}^+ \to \text{C}^*\text{-alg}$ the fusion spin chain with boundary obtained by restricting $\Loc^{\bdy}_\bullet$ to intervals in $\mathbb{Z}_{\geq 0}$.
Then $\DHR(\Loc^{\bdy}_\bullet) \simeq \C^{\rev}$.
\end{corollary}

\begin{remark}
If given the net of local algebras $\Loc_\bullet$, we could recover $\mathrm{Rep}(G)$ but not $G$, as different finite groups could have the same representation category \cite{Etingof_2000_IsocategoricalGroups}. On the other hand, if our universe indeed has some global or gauge symmetry, the symmetry action could never be observed, as everything we could see must be symmetric; but the local quantum charges could be observed (e.g. electric charges), therefore more physical.
\end{remark}

In this section, we discuss the relation between $\DHR(\Loc^{\bulk}_\bullet)$ and $\DHR(\Loc^{\bdy}_\bullet)$, where $\Loc^{\bulk}_\bullet$ is the fusion spin chain without boundary built from $\C$, and $\Loc^{\bdy}_\bullet$ is the fusion spin chain with boundary constructed in Section~\ref{section: fusion spin chains with bdy}.
We shall prove that there is a natural action of $\DHR(\Loc^{\bulk}_\bullet)$ on $\DHR(\Loc^{\bdy}_\bullet)$ equipped with a central structure, and the action coincides with the action of $Z_1(\C^\rev)$ on $(\C_\M^\vee)^{\rev}$.

We build $\Loc^{\bulk}_\bullet$ following \cite{Jones_2024_DHR}: Let $x = \oplus_{a \in \mathrm{Irr}(\C)} a$.
We put one copy of $x$ on each site $n \in \mathbb{Z}$.
Local operators supported in some $I \in \mathcal{I}$ are just endomorphisms of the tensor product of all objects in $I$; that is, we set $\Loc^{\bulk}_{[a,b]} := \C(x^{b-a+1},x^{b-a+1})$.
Note that $\Loc^{\bulk}_{[a,b]}$ is a multi-matrix ${}^*$-algebra, so it is a \(C^*\)-algebra with the \(C^*\)-norm given by the operator norm.

\subsection{Definition of the bulk-boundary action}

Now we define the action of $\DHR(\Loc^{\bulk}_{\bullet})$ on $\DHR(\Loc^{\bdy}_\bullet)$, that is, a monoidal functor $U: \DHR(\Loc^{\bulk}_{\bullet}) \to \DHR(\Loc^{\bdy}_\bullet)$.
We take the viewpoint of $\M = \WC$, thus we have a canonical inclusion $\Loc^{\bdy} \hookrightarrow \Loc^{\bulk}$.
Concretely, the half-infinite spin chain with boundary is identified with the right half of the infinite spin chain: the boundary site is placed at $0$, and the bulk sites are $1,2,\dots$.
Thus an operator supported on $[0,n]$ in the boundary chain is viewed as an operator supported on the same interval inside the infinite chain.

We first define the object $U(X)$.
For $X \in \DHR(\Loc^{\bulk}_{\bullet})$, we take some projective basis $S = \{ s_n \}$ of $X$ localized in some $[p,q]$ such that $0 < p < q$.
Such a basis exists because a bulk DHR bimodule is mobile: by definition it admits localized projective bases in sufficiently long intervals, and we may choose such an interval entirely inside the right half-chain.
Then we define $U(X) = \Loc^{\bdy} S \Loc^{\bdy}$; this definition makes sense since $\Loc^{\bdy}$ is viewed as a subalgebra of $\Loc^{\bulk}$.
Intuitively, $U(X)$ is the same bulk non-local operator sector, but now only boundary observables are allowed to act on its left and right.
We shall now show $\Loc^{\bdy} S \Loc^{\bdy}$ is indeed a DHR bimodule.

\begin{proposition}
$S$ is a projective basis of $\Loc^{\bdy} S \Loc^{\bdy}$.
\end{proposition}
\begin{proof}
Note that in $X$, any element $a s_i b$ could be written as a linear combination of the projective basis elements by $a s_i b = \sum_j s_j \left\langle s_j|a s_i b \right\rangle = \sum_j s_j \left\langle s_j|a s_i\right\rangle b$.
Therefore, it suffices to show that for any $a \in \Loc^{\bdy}$, $\left\langle s_j|a s_i \right\rangle \in \Loc^{\bdy}$.

Since $\Loc^{\bdy}$ is boundedly generated, we could write $a$ as a linear combination of products $a = \sum_j a_{1,j} \dots a_{n,j}$, such that each $a_{i,j}$ is supported in an interval with length smaller than some constant $R$.
Then $a s_i b = \sum_j a_{1,j} \dots a_{n,j} s_i b$.
Therefore, it suffices to prove that for $a_{n,j}$ supported in an interval with length smaller than $R$, we have $\left\langle s_h|a_{n,j} s_i \right\rangle \in \Loc^{\bdy}$ for every basis element $s_h$; since $s_i = \sum_h s_h \left\langle s_h|a_{n,j} s_i \right\rangle$, by induction we could prove $a_{1,j} \dots a_{n,j} s_i b \in \Loc^{\bdy}$, thus $a s_i b \in \Loc^{\bdy}$.

If the support of $a$ has overlap with $[p,q]$, then the support of $a$ must be a subset of $[p-R,q+R]$.
Take any $c \in \Loc^{\bulk}_{(-\infty,p-R)} \cup \Loc^{\bulk}_{(q+R, \infty)}$.
Since $s_i, s_j, a$ are all localized in $[p-R, q+R]$, $c$ could commute through all of them.
Then we have
\begin{equation}
\begin{aligned}
\left\langle s_j|a s_i \right\rangle c &= \left\langle s_j|a s_i c \right\rangle \\
&= \left\langle s_j|a c s_i  \right\rangle = \left\langle s_j| c a s_i  \right\rangle \\
&= \left\langle c^* s_j| a s_i  \right\rangle = \left\langle s_j c^*| a s_i  \right\rangle \\
&= c \left\langle s_j | a s_i  \right\rangle.
\end{aligned}
\end{equation}
Since $\Loc^{\bulk}_\bullet$ satisfies algebraic Haag duality, this implies $\left\langle s_j|a s_i \right\rangle \in \Loc^{\bulk}_{[p-R,q+R]} = \Loc^{\bdy}_{[p-R,q+R]}$.

If the support of $a$ does not overlap with $[p,q]$,
\begin{equation}
\left\langle s_j|a s_i \right\rangle = \left\langle s_j |s_i a \right\rangle = \left\langle s_j|s_i \right\rangle a.
\end{equation}
We already know that $\left\langle s_i|s_j \right\rangle = \left\langle s_i|1_{\Loc^{\bulk}} \cdot s_j \right\rangle \in \Loc^{\bdy}_{[p,q]}$ by the previous case, therefore $\left\langle s_j|a s_i \right\rangle \in \Loc^{\bdy}_{[p-R,q+R]}$.
\end{proof}

\begin{proposition}
$\Loc^{\bdy} S \Loc^{\bdy}$ has a natural $\Loc^{\bdy}$ valued inner product, that is, for any $x,y \in \Loc^{\bdy} S \Loc^{\bdy}, \left\langle x|y \right\rangle_{\Loc^{\bulk}} \in \Loc^{\bdy}$.
Therefore, $\Loc^{\bdy} S \Loc^{\bdy}$ is indeed a Hilbert $\Loc^{\bdy}$-bimodule.
\end{proposition}

\begin{proof}
First note that $\left\langle s_m|s_n \right\rangle \in \Loc^{\bdy}$ by the previous proposition.
Then since $\{ s_n \}$ is a projective basis, any $x \in \Loc^{\bdy} S \Loc^{\bdy}$ could be written as $\sum_i s_i b_i$, where $b_i \in \Loc^{\bdy}$.
Therefore, it suffices to verify $\langle s_m b_i |  s_n b_j \rangle \in \Loc^{\bdy}$; evidently $\langle s_m b_i |  s_n b_j \rangle = b_i^* \left\langle s_m|s_n \right\rangle b_j \in \Loc^{\bdy}$.
\end{proof}

\begin{proposition}
The definition $U(X)$ is independent of the choice of $S$.
\end{proposition}
\begin{proof}
Consider two projective bases of $X$, $S_1 = \{ s_m \}, S_2 = \{ t_n \}$, localized in $I_1, I_2$ respectively, where $I_1, I_2 \in \mathcal{I}^+$.
Take the smallest interval $I \in \mathcal{I}^+$ containing $I_1 \cup I_2$.
Note that $s_m = \sum_n t_n \left\langle s_m|t_n \right\rangle$.
We shall show $a \left\langle s_m|t_n \right\rangle = \left\langle s_m|t_n \right\rangle a$ for any $a \in \Loc^{\bulk}_{I^c}$:
\begin{equation}
\begin{aligned}
\left\langle s_i|t_j \right\rangle a &= \left\langle s_i|t_j a \right\rangle = \left\langle s_i|a t_j \right\rangle \\
&= \left\langle a^* s_i|t_j \right\rangle = \left\langle s_i a^*|t_j \right\rangle \\
&= a\left\langle s_i|t_j \right\rangle.
\end{aligned}
\end{equation}
Then by algebraic Haag duality, $\left\langle s_i|t_j \right\rangle \in \Loc^{\bulk}_{I} \subset \Loc^{\bdy}$, so indeed we could use operators in $\Loc^{\bdy}$ to generate $s_m$ from $\{ t_n \}$.
\end{proof}

Next we should define how the functor acts on morphisms.
For some bimodule intertwiner $f: X \to Y$, we take a projective basis of $X$ localized in some $[p,q]$ and denote it as $\{ s_i \}$; define $U(f): U(X) \to U(Y)$ to send $s_i$ to $f(s_i)$.
Obviously it extends to a bimodule intertwiner by bilinearity, and the definition is functorial.

The independence of the chosen localized basis is important here: replacing $S$ by another localized basis only changes the generators by coefficients in $\Loc^{\bdy}$, hence gives the same sub-bimodule and the same morphism after extension by bilinearity.

Finally, we should prove the associativity of the action, that is, there are natural isomorphisms $U(X \otimes Y) \simeq U(X) \otimes U(Y)$.
Pick a projective basis $\{ s_i \}$ of $X$ and a projective basis $\{ t_j \}$ of $Y$, both localized in some $[p,q]$.
Note that $\{ s_i \}$ is a projective basis of $U(X)$, and $\{ t_j \}$ is a projective basis of $U(Y)$, therefore $\{ s_i \boxtimes t_j\}$ is a projective basis of $U(X) \otimes U(Y)$.
On the other hand, $\{ s_i \boxtimes t_j\}$ is a projective basis of $X \otimes Y$, therefore a projective basis of $U(X \otimes Y)$.
Then we have a bimodule isomorphism $U(X \otimes Y) \to U(X) \otimes U(Y)$ defined by $s_i \boxtimes t_j \mapsto s_i \boxtimes t_j$; it is easy to verify the definition is natural and satisfies the pentagon equation.

\subsection{Computation of the module structure}\label{ComputationModStructure}

Knowing that $\DHR(\Loc^{\bulk}_{\bullet}) \simeq Z_1(\C^\rev), \DHR(\Loc^{\bdy}_\bullet) \simeq (\CMdual)^{\rev}$, we compare the action of $\DHR(\Loc^{\bulk}_{\bullet})$ on $\DHR(\Loc^{\bdy}_\bullet)$ to the action of $Z_1(\C^\rev)$ on $(\CMdual)^{\rev}$. In this section we would show that these two central actions are equivalent.
More precisely, we consider the following diagram, which is commutative (up to a monoidal natural isomorphism $\alpha$):

\[\input{diagrams/bulk-boundary-action-3473.tex}\]

The goal of this section is to construct the natural isomorphism $\alpha: U \circ \ReaDHR^{\bulk}\Rightarrow \ReaDHR^{\bdy} \circ K$.

First we recall some relevant definitions and facts.
There is a canonical action $K: Z_1(\C^\rev) \to (\CMdual)^{\rev}$, defined as follows:
Given $(z, \sigma) \in Z_1(\C^\rev)$, we define the functor $K_{z, \sigma}: \M \to \M$ to be $m \mapsto m \otimes z$ on objects, and $f \mapsto f \otimes \mathrm{id}_z$ on morphisms.
$K_{z, \sigma}$ is a module functor since for $y \in \C$, we have the isomorphism $\beta_{m,y}: K_{z, \sigma}(m \otimes y) \to K_{z, \sigma}(m) \otimes y$ given by $\mathrm{id}_m \otimes \sigma_{z,y}$, where $\sigma_{z,y}:y\otimes z\to z\otimes y$ is the half-braiding in $\C^\rev$.
The orientation is summarized by the following diagram:
\[
\input{diagrams/canonical-action-3492.tex}
\]
The action $\ReaDHR^{\bulk}: Z_1(\C^\rev) \to \DHR(\Loc^{\bulk}_{\bullet})$ is again an AF-action, so we review the definitions of finite pieces $\ReaDHR_n^\bulk$ \cite[Page 35]{Jones_2024_DHR}: We define $\ReaDHR_n^{\bulk}(z, \sigma) = \C(x^{n}, x^{n} \otimes z \otimes x^n)$, with the inner product $\left\langle \xi|\eta \right\rangle = \xi^\dagger \circ \eta$.
The left and right action by $\Loc^{\bulk}_{[1,2n+1]}$ is defined as
\begin{equation}
a\triangleright \xi \triangleleft b:= ((\sigma_{z, x^{n}})^{-1} \otimes \mathrm{id}_{x^n}) \circ (\mathrm{id}_z \otimes a) \circ (\sigma_{z, x^{n}} \otimes \mathrm{id}_{x^n}) \circ \xi \circ b.
\end{equation}

Next we construct $\alpha$, which boils down to defining $\Loc^{\bdy}$-bimodule isomorphisms $\alpha_{z, \sigma}: U \circ \ReaDHR^{{\bulk}}(z, \sigma) \to \ReaDHR^{{\bdy}} \circ K_{z,\sigma}$.
Since both are AF actions, it suffices to prove that their finite pieces are isomorphic, i.e. construct $\alpha_{(z, \sigma),n}: U \circ \ReaDHR^{{\bulk},n}(z, \sigma) \to \ReaDHR^{{\bdy},2n} \circ K_{z,\sigma}$.
We may take $U \circ \ReaDHR^{{\bulk},n}(z, \sigma) = \M(m \otimes x^{n}, m \otimes x^{n} \otimes z \otimes x^n)$.
Recall that $\ReaDHR^{{\bdy},2n} \circ K_{z,\sigma} = \M(m \otimes x^{n}, m \otimes z \otimes x^{2n})$, and the action is defined by
\begin{equation}
a\triangleright \xi \triangleleft b:= (\mathrm{id}_m \otimes \sigma_{z, x^{2n}}) \circ (a \otimes \mathrm{id}_z) \circ (\mathrm{id}_m \otimes (\sigma_{z, x^{2n}})^{-1}) \circ \xi \circ b.
\end{equation}

Then apparently, $\alpha_{z, \sigma}: U \circ \ReaDHR^{{\bulk}}(z, \sigma) \to \ReaDHR^{{\bdy}} \circ K_{z,\sigma}$ may be constructed as $\xi \mapsto (\mathrm{id}_{m} \otimes \sigma_{z, x^n} \otimes \mathrm{id}_{x^n})\circ \xi$.
$\alpha_{z, \sigma}$ is invertible since $\sigma$ is; it is straightforward to verify it is a bimodule intertwiner.

To summarize, we proved the following theorem in this section:

\LinkedRestatableTarget{target:symtftbulkbdy}\symtftbulkbdy

%----

\section{Bulk-boundary correspondence of 1D phases}
\label{section: bulk-boundary correspondence}

In this section, we discuss the bulk-boundary correspondence of (1+1)D gapped quantum phases with symmetry.

We are interested in gapped phases with symmetry. In this situation, topological defects do not capture the complete data of the phase; indeed, the category of boundary conditions of any (1+1)D gapped phase is merely a linear 1-category, whose center only reflects the number of distinct simple boundary conditions.
What's missing? We have already seen in Section~\ref{subsection: non-local ops to hilb bimods} that the bulk SymTFT plays a central role in characterizing the bulk phase.
The physical intuition is, the operator algebra has nontrivial structure due to the presence of symmetry, which must be incorporated in the categorical description of the phase. It turns out that the SymTFT is the correct categorical encapsulation of symmetry.

In (1+1)D, there are several works extracting the enriched fusion category data from lattice models of gapped phase with abelian group symmetry and studying the bulk-boundary relations \cite{Kong_2022_QL1,Xu_2024_1DPhaseAbelianSym}.
The enriched background with input symmetry $G$ is $Z_1(\Rep(G))$ interpreted as ``sectors of symmetric non-local operators", which is unfortunately ill-defined.
In \cite{LanZhou_2024_QuantumCurrent}, the authors studied $Q$-systems models of (1+1)D gapped phases with symmetry. They proposed a complete categorical description of these phases by the enriched fusion category ${}^{Z_1(\C)} ({}_Q \C_Q)^\rev$, where $\C$ is the category of local quantum charges, $({}_Q \C_Q)^\rev$ is the category of topological excitations of the phase, and $Z_1(\C)$ is interpreted as the category of so-called quantum currents. However, the bulk-boundary relation is not examined closely in their work.

In this section, we first review the discussion of bulk-boundary relations in the Ising chain in \cite{Kong_2022_1DEnrichedCat}, which provides essential physical intuition.
We then rigorously establish the correspondence for general gapped phases using the language of operator algebras and Hilbert bimodules. We would define the action of the boundary DHR category $\DHR(\Loc^{\bdy}_\bullet)$ on the category of boundary conditions $\BCond$ as tensor product of Hilbert bimodules.
Then we show that this action is equivalent to the canonical action of $(\CMdual)^{\rev}$ on $\MQ^{\op}$, with the bending isomorphism identifying it with the action of $(\CMdual)^{\op}$, and we prove that the bulk phase is the center of the enriched category ${}^{(\CMdual)^{\rev}} \MQ^{\op}$.

\subsection{First example: The Ising chain}
\label{subsection: first example the ising chain}
We only review the discussion of symmetric boundary for brevity; the story for the symmetry breaking boundary is quite similar. 
Note that the discussion uses the language of non-local operators, which is unrigorous but intuitive. 
As in Section~\ref{subsection: non-local ops to hilb bimods}, this discussion could be rigorously reformulated as QCAs and Hilbert bimodules.

\paragraph{The trivial phase}
Obviously there are two simple boundary conditions of the phase: One has trivial \(Z_2\) charge, and the other has a nontrivial \(Z_2\) charge.
Namely, the first boundary is given by the Hamiltonian
\begin{equation}
H_+ = -\sum_{i \geq 0} X_i,
\end{equation}
and the ground state is
\begin{equation}
|\psi_+ \rangle=|+++ \dots \rangle.
\end{equation}
The second boundary is given by the Hamiltonian
\begin{equation}
H_- = X_0 -\sum_{i >0} X_i,
\end{equation}
and the ground state has a \(Z_2\) charge trapped on the boundary:
\begin{equation}
|\psi_- \rangle=|-++ \dots \rangle.
\end{equation}
Therefore, the category of boundary conditions \(\BCond\) has two simple objects: \(b_+\) corresponds to \(| \psi_+ \rangle\), and \(b_-\) corresponds to \(| \psi_- \rangle\).
Now we consider how the boundary SymTFT acts on the category of boundary conditions.
The action of \(Z_i\) induces a spin flip, which exchanges the two boundary conditions.
Therefore, \(\BCond\) is a module category over \(\Rep(\Z_2)\), with rules of the action given by
\begin{equation}
\begin{aligned}
& b_+  \otimes\one = b_+, \quad  b_-\otimes \one = b_- \\
& b_+ \otimes Z_i = b_-, \quad  b_-\otimes Z_i = b_+ .
\end{aligned}
\end{equation}
This is equivalent to the indecomposable module category \(\Rep(\Z_2)\) (as a linear category) over \(\Rep(\Z_2)\).

\paragraph{The symmetry-breaking phase}
There is only one simple boundary condition of the phase, which is given by the Hamiltonian
\begin{equation}
H_+ = -\sum_{i \geq 0} Z_i Z_{i+1},
\end{equation}
and the unique symmetry-invariant ground state is
\begin{equation}
|\psi\rangle = \frac{1}{\sqrt 2} \left( |000 \cdots\rangle + |111 \cdots\rangle \right).
\end{equation}
Therefore, the category of boundary conditions \(\BCond\) has only one simple object \(b\).
Now we consider how the boundary SymTFT acts on the category of boundary conditions.
The action of \(Z_i\) induces a relative phase of \(\pi\) between \(|000 \cdots\rangle\) and \(|111 \cdots\rangle\), but this relative phase is again unphysical, as it cannot be detected by any symmetric operator.
Therefore, \(\BCond\) is a module category over \(\Rep(\Z_2)\), with rules of the action given by
\begin{equation}
b \otimes \one = b, \quad b \otimes Z_i = b .
\end{equation}
This is isomorphic to the indecomposable module category \(\ve\) (as a linear category) over \(\Rep(\Z_2)\).

\subsection{Action of the boundary DHR category on boundary conditions}
\label{subsec:action of bdyDHR on BCond}

In general, objects of the SymTFT should be formulated by DHR bimodules instead of non-local operators.
As in Section~\ref{subsection: non-local ops to hilb bimods}, non-local operators in the previous subsection could be translated into DHR bimodules, and applying non-local operators to quantum states would be faithfully translated into tensor product of Hilbert bimodules with left modules.
This indicates we should consider the action of the boundary DHR category, \(\DHR(\Loc^{\bdy}_\bullet)\), on the category of boundary conditions, \(\BCond\), by tensor product of Hilbert bimodules:
\begin{equation}
\otimes: \DHR(\Loc^{\bdy}_\bullet) \times \BCond \to \BCond, \quad (M,N) \mapsto M \boxtimes_{\Loc^{\bdy}} N.
\end{equation}

We would like to compare this action to categorical data.
Recall that \(\DHR(\Loc^{\bdy}_\bullet) \simeq (\C_\M^\vee)^{\rev}, \BCond \simeq \MQ^{\op}\), and there is a canonical action of \(\C_\M^\vee\) on \(\MQ\) by applying module functors to objects in \(\M\):
\begin{equation}
\otimes: \C_\M^\vee \times \MQ \to \MQ, \quad (L,K) \mapsto L(K)
\end{equation}

However, we need to somehow adjust \((\C_\M^\vee)^{\rev}\) to \((\C_\M^\vee)^{\op}\), since only the latter could act on \(\MQ^{\op}\) by the canonical action.
This problem could be visualized as follows: If we naively consider \(\ReaDHR(L) \boxtimes \ReaBCond(K)\), the result has an unwanted \(L\) leg on the top, which cannot be compared to \(\ReaBCond (L(K))\), which has the leg \(L\) on the bottom instead of on the top.

Obviously we would like to bend the leg \(L\) to the downside.
Fortunately there is a canonical way to do this: It is well-known that when \(\M\) is an indecomposable semisimple module category, \(\C_\M^\vee\) is a fusion category \cite[Theorem 9.3.2]{EGNO_2015}, which means any \(L \in \C_\M^\vee\) has a left dual \(L^\vee\).
This supplies us with a canonical isomorphism, called the \emph{bending isomorphism}, as follows:
\begin{equation}\label{eq: bending isomorphism}
\M(K_1, L(K_2)) \simeq \M(L^\vee(K_1), K_2).
\end{equation}

This isomorphism is natural in both \(K_1\) and \(K_2\). 
Then we can establish the isomorphism of \(\ReaDHR(L) \boxtimes \ReaBCond(K)\) and \(\ReaBCond(L^\vee(K))\), as shown in \cref{fig:bending isomorphism}.

\begin{figure}
    \centering

\input{diagrams/bending-isomorphism-3635.tex}

    \caption{The bending isomorphism allows comparing \(\ReaDHR(L) \boxtimes \ReaBCond(K)\) to \(\ReaBCond(L^\vee(K))\).}
    \label{fig:bending isomorphism}

\end{figure}

This also fixes the problems on the categorical side, as taking the left dual is a functor \((-)^\vee: (\CMdual)^{\rev} \to {(\CMdual)}^{\op}\).
Thus the induced action of \(L\in(\CMdual)^\rev\) on \(K\in\MQ^\op\) is represented objectwise by \(L\cdot K:=L^\vee(K)\).

To summarize, we obtain an action of \((\C_\M^\vee)^{\rev}\) on \(\BCond\) defined as follows:
\[
\begin{aligned}
    (\C_\M^\vee)^{\rev} \times \BCond & \xrightarrow{\ReaDHR \times \mathrm{id}} \DHR(\Loc^{\bdy}_\bullet) \times \BCond \\
    & \xrightarrow{\boxtimes_{\Loc^{\bdy}}} \BCond.
\end{aligned}
\]

Now both \(\BCond\) and \(\MQ^{\op}\) are left module categories over \((\C_\M^\vee)^{\rev}\), so they can be compared.
In this section, we would show that these two modules are equivalent; that is, \(\ReaBCond: \MQ^{\op} \to \BCond\) is not only an equivalence of categories, but an \emph{equivalence of module categories} \cite[Section 7.2]{EGNO_2015}.

We should first promote \(\ReaBCond\) to a module functor. 
In formulas involving objects, we freely identify the objects of \(\MQ\) and \(\MQ^\op\); morphisms are reversed only when the category structure is used.
Recall that a (left) module functor between two module categories $\mathcal{M}$ and $\mathcal{N}$ over the same tensor category $\mathcal{D}$ is a pair $(F, \beta)$, where $F: \mathcal{M} \rightarrow \mathcal{N}$ is a functor and $\beta_{X, M}: F(X \otimes M) \rightarrow X \otimes F(M)$ is a family of natural isomorphism satisfying the pentagon and triangle equations. Therefore we need to find
\begin{equation}
\label{eq.bdy_DHR_action_on_BC}
\beta_{L,K}: \ReaBCond (L^\vee(K)) \simeq \ReaDHR(L) \boxtimes_{\Loc^{\bdy}} \ReaBCond(K), \quad L \in \C_\M^\vee, K \in \MQ
\end{equation}
satisfying consistency conditions.

Note that both \(\ReaBCond (L^\vee(K))\) and \(\ReaDHR(L) \boxtimes_{\Loc^{\bdy}} \ReaBCond(K)\) arise from AF actions.
Therefore, it suffices to find natural isomorphisms for the finite-dimensional building block:
\begin{equation}
\beta_{L,K,n}: \ReaBCond_n (L^\vee(K)) \to \ReaDHR_n(L) \boxtimes_{\Loc_n} \ReaBCond_n(K), \quad L \in \CMdual, K \in \MQ;
\end{equation}
Obviously the bending isomorphism defined in \eqref{eq: bending isomorphism} does the job.
Taking the inductive limit gives \(\beta_{L,K}\), which finishes the promotion of \(\ReaBCond\) to a module functor.

Note that we may treat both \((\CMdual)^{\rev}\) and \(\DHR(\Loc^{\bdy}_\bullet)\) as strict monoidal categories, then all consistency isomorphism of module actions become trivial. Thus \(\beta_{L,K}\) satisfy the consistency equations required.

We have proved in Section~\ref{subsec: ReaBCond is an equivalence of categories} that \(\ReaBCond\) is an equivalence, which finishes the proof that \(\ReaBCond\) is an equivalence of left module categories over \((\CMdual)^{\rev}\).
To summarize, we obtained the following theorem:

\LinkedRestatableTarget{target:bulkbdyaction}\bulkbdyaction
\begin{proof}
The bending isomorphism constructed above promotes \(\ReaBCond\) to a module functor for the \((\C_\M^\vee)^{\rev}\)-actions. Since \(\ReaBCond\) is an equivalence by Theorem~\ref{thm:intro-bcond-classification} and \(\ReaDHR\) is a monoidal equivalence by Theorem~\ref{thm:intro-boundary-symtft}, the diagram above identifies the operator-algebraic action with the canonical categorical action.
\end{proof}

The theorem could be expressed by the following commutative diagram of categories:
\[\input{diagrams/module-equivalence-3819.tex}\]

\subsection{Enriched center and bulk-boundary correspondence}

With all necessary preparations done, we could give the full description of a 1D gapped symmetric phase with gapped boundary, as shown in Figure~\ref{fig:1D phase complete description}. The full description consists of four elements: the category of bulk topological defects, \(\mathrm{TopDef}\); the category of boundary conditions, \(\BCond\); the boundary SymTFT realized as the boundary DHR category; the bulk SymTFT realized as the (bulk) DHR category \cite{Jones_2024_DHR}.

\begin{figure}
    \centering  

\input{diagrams/phase-description-3843.tex}
    
    \caption{Complete description of 1D symmetric gapped phases with gapped boundary.}
    \label{fig:1D phase complete description}
  
\end{figure}
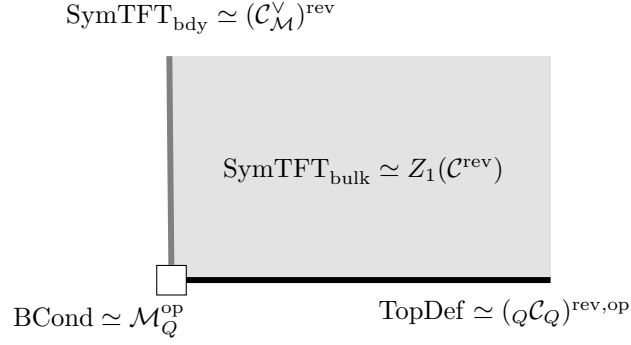

The complete categorical description of the boundary is the collection of macroscopic observables on the boundary. In our framework of operator algebras, for bulk charge category $\cC$, fixing a boundary local quantum charge category $\cM$, all possible boundary conditions form the category $\BC$. This is not everything that happens on the boundary. Although local symmetric operators cannot change the type of boundary condition, boundary $\DHR$ bimodules can serve as the sectors of boundary changing operators as illustrated in Figure~\ref{fig:bending isomorphism}.  $\BC$ as the $\DHR(\Loc_\bullet^\bdy)$ module, with action~\eqref{eq.bdy_DHR_action_on_BC}, is then the complete macroscopic observables on the boundary. Such data is equivalent to the enriched category
\begin{equation}
{}^{\DHR(\Loc_\bullet^\bdy)}\BC,
\end{equation}
given by the so-called \emph{canonical construction}.
The hom-object is defined as the internal hom of the corresponding module action. For a review on relevant concepts in enriched category theory, see Appendix~\ref{app.enriched_cat}.

Due to the equivalences $\DHR(\Loc_\bullet^\bdy)\simeq(\cC_\M^\vee)^\rev$~(Theorem~\ref{thm:intro-boundary-symtft}), $\BC\simeq \MQ^\op$~(Theorem~\ref{thm:intro-bcond-classification}), and the fact that the equivalences are compatible with module structures (Theorem~\ref{thm:intro-bulk-boundary-action}), Theorem~3.45 in~\cite{Kong_Yuan_2024_EnrichedMonoidalCat} (reviewed in Theorem~\ref{thm.cc_ECat}) yields that  
\begin{equation}\label{eq:bdy enriched cat vs category theory}
{}^{\DHR(\Loc_\bullet^\bdy)}\BC \simeq {}^{(\cC_\M^\vee)^\rev}{\MQ^\op}
\end{equation}
as enriched categories.

On the other hand, the enriched monoidal category describing the bulk should be 
\[
    {}^{\DHR(\Loc_\bullet^\bulk)}\mathrm{TopDef}.
\]

It is rigorously proved that $\DHR(\Loc_\bullet^\bulk) \simeq Z_1(\C^\rev)$ \cite{Jones_2024_DHR} and $\mathrm{TopDef} \simeq (\QCQ)^{\rev,\op}$ (Remark~\ref{rmk:TopDef}), and we can apply techniques in Section~\ref{subsec:action of bdyDHR on BCond} to show that the action of $\DHR(\Loc_\bullet^\bulk) \simeq Z_1(\C^\rev)$ on $\mathrm{TopDef} \simeq (\QCQ)^{\rev,\op}$ is equivalent to the categorical action of $Z_1(\C^\rev)$ on $(\QCQ)^{\rev,\op}$.
% \footnote{There is a canonical categorical action of $Z_1(\C)$ on $\QCQ$. Our action involves a bending isomorphism $Z_1(\C) \xrightarrow{(-)^\vee} Z_1(\C)^{\rev,op}$}.
Therefore we have the equivalence of enriched monoidal categories
\begin{equation}\label{eq:bulk enriched cat vs category theory}
    {}^{\DHR(\Loc_\bullet^\bulk)}\mathrm{TopDef} \simeq {}^{Z_1(\C^\rev)}(\QCQ)^{\rev,\op}
\end{equation}
This agrees with \cite{Kong_2022_1DEnrichedCat, Xu_2024_1DPhaseAbelianSym,LanZhou_2024_QuantumCurrent} after translating orientation conventions; see Appendix~\ref{app:orientation:enriched}. In particular, our convention replaces $\C$ by $\C^\rev$, and $(\QCQ)^{\rev,\op}$ is canonically equivalent to $\QCQ$ by rigidity.

With all necessary preparations ready, we could finally state and prove the main theorem of this paper:

\LinkedRestatableTarget{target:bulkbdycorrespondence}\bulkbdycorrespondence
\begin{proof}
According to Corollary 4.41 in~\cite{Kong_Yuan_2024_EnrichedMonoidalCat}, given a monoidal category $\cD$ and a left $\cD$ module $\cN$, the $E_0$ center of ${}^\cD\cN$ is given by 
\begin{equation}
    Z_0({}^\cD\cN)= {}^{Z_1(\cD)}\Fun_\cD(\cN,\cN).
\end{equation}
 Here,  $\Fun_\cD(\cN,\cN)$ is the category of left module functors from $\cN$ to $\cN$. The enriched monoidal category structure comes from the structure of monoidal left $Z_1(\cD^\rev)\cong \overline{Z_1(\cD)}$ module structure of $\Fun_\cD(\cN,\cN)$, given by 
 the canonical central functor 
 $Z_1(\cD^\rev)\simeq Z_1(\Fun_\cD(\cN,\cN)   )\xrightarrow{\mathrm{forget}} \Fun_\cD(\cN,\cN)$ (see~\ref{prp.MonoidalCC}, or Proposition 4.17 in~\cite{Kong_Yuan_2024_EnrichedMonoidalCat}).  
 
In our case,  we have 
\begin{equation}
    Z_0({}^{\DHR(\Loc_\bullet^\bdy)}\BC) \simeq Z_0( {}^{(\cC_\M^\vee)^\rev}{\MQ^\op})= {}^{Z_1((\cC_\M^\vee)^\rev)} \Fun_{(\cC_{\cM}^\vee)^\rev}(\MQ^\op, \MQ^\op)
\end{equation}
For the underlying part, passing to opposite module categories gives
\begin{equation}
\Fun_{(\cC_{\cM}^\vee)^\rev}(\MQ^\op, \MQ^\op)
\simeq \Fun_{(\cC_{\cM}^\vee)^{\rev,\op}}(\MQ, \MQ)^\op.
\end{equation}
Then the duality functor identifies the acting category with \(\cC_\cM^\vee\), so
\begin{equation}
\Fun_{(\cC_{\cM}^\vee)^{\rev,\op}}(\MQ, \MQ)^\op
\simeq \Fun_{\cC_\cM^\vee}(\MQ,\MQ)^\op.
\end{equation}
Finally, the standard Morita computation gives
\begin{equation}
\Fun_{\cC_\cM^\vee}(\MQ,\MQ)^\op\simeq ({}_Q\cC_Q)^{\rev,\op}.
\end{equation}
For the background part, let $\cM\simeq {}_W\cC$ as right $\cC$ modules for certain separable algebra $W\in \cC$, then
\begin{equation}
Z_1((\cC_\cM^\vee)^\rev)\simeq Z_1(({}_W\cC_W)^\rev)\simeq Z_1(\cC^\rev).
\end{equation}
Combining with Equation~\ref{eq:bdy enriched cat vs category theory} and Equation~\ref{eq:bulk enriched cat vs category theory}, we obtain
\begin{equation}
    Z_0({}^{\DHR(\Loc_\bullet^\bdy)}\BC) \simeq Z_0({}^{(\cC_\M^\vee)^\rev}{\MQ^\op}) \simeq {}^{Z_1(\C^\rev)}(\QCQ)^{\rev,\op} \simeq {}^{\DHR(\Loc_\bullet^\bulk)}\mathrm{TopDef}
\end{equation}
which finishes the proof.
\end{proof}

%----------
\paragraph{Acknowledgments}
TL is supported by start-up funding from The Chinese University of
Hong Kong, by funding from Research Grants Council, University Grants Committee
of Hong Kong (ECS No.~24304722, CRF No.~C7015-24GF), and also by Guangdong Provincial Quantum Science Strategic Initiative (No.~GDZX2501013).

\appendix

\section{Enriched category}
\label{app.enriched_cat}
In this appendix, we review the basic theory of enriched category and enriched monoidal category based on \cite{Kong_Yuan_2024_EnrichedMonoidalCat}.
\begin{definition}

Let $(\mathcal{A},\otimes,\mathbbm{1}_{\mathcal{A}})$ be a monoidal category. An $\mathcal{A}$-enriched category $^{\mathcal{A}}\mathcal{L}$ consists of the following data
\begin{itemize}
    \item A collection of objects $\mathrm{Ob}(^{\mathcal{A}}\mathcal{L})$, for simplicity, we just use $^{\mathcal{A}}\mathcal{L}$ to denote $\mathrm{Ob}(^{\mathcal{A}}\mathcal{L})$.
    \item For any $x,y\in {^{\mathcal{A}}}\mathcal{L}$, there is a so-called hom-object $^{\mathcal{A}}\mathcal{L}(x,y)\in \mathcal{A}$.
    \item For any $x\in {^{\mathcal{A}}}\mathcal{L}$, there is a distinguished identity morphism:
\begin{equation*}
    1_x: \mathbbm{1}_\mathcal{A}\rightarrow {^{\mathcal{A}}\mathcal{L}}(x,x) 
\end{equation*}
\item  For any triple $x,y,z\in{^{\mathcal{A}}}\mathcal{L}$, there is a composition morphism for the hom-objects: 
\begin{equation*}
    \circ:\  {^{\mathcal{A}}\mathcal{L}}(y,z)\otimes {^{\mathcal{A}}\mathcal{L}}(x,y) \rightarrow {^{\mathcal{A}}\mathcal{L}}(x,z)
\end{equation*}
\end{itemize}
The data satisfy the following commuting diagram, for any $x,y,z,w\in {^{\mathcal{A}}}\mathcal{L}$:
\begin{equation}
\input{diagrams/enriched-associativity-3974.tex}
\end{equation}
here the associator in $\mathcal{A}$ is implied.
\begin{equation}
\label{identity}
\input{diagrams/enriched-identity-3986.tex}
\end{equation}
\end{definition}
We will call the category $\mathcal{A}$ the \emph{background category} of $^{\cA}\cL$.
\begin{definition}
    For an enriched category $^{\cA}\cL$, we define an ordinary category $\cL$, 
    \begin{itemize}
        \item $\mathrm{Ob}(\mathcal{L}):= \mathrm{Ob}({^{\mathcal{A}}\mathcal{L}})$;
        \item $\mathcal{L}(x,y):= \mathcal{A}(\mathbbm{1}_{\mathcal{A}}, {^{\mathcal{A}}\mathcal{L}}(x,y))$, for any $x,y\in \mathcal{L}$
        \item Suppose $g\in \mathcal{L}(y,z)$ and $f\in \mathcal{L}(x,y)$, their composition is defined as:
\begin{equation}
\label{composition}
    g\circ f: \mathbbm{1}_{\mathcal{A}} \rightarrow \mathbbm{1}_{\mathcal{A}}\otimes \mathbbm{1}_{\mathcal{A}}\stackrel{g\otimes f}{\longrightarrow}  {^{\mathcal{A}}\mathcal{L}}
(y,z) \otimes  {^{\mathcal{A}}\mathcal{L}}(x,y) \stackrel{\circ}{\longrightarrow}  {^{\mathcal{A}}\mathcal{L}}(x,z)
\end{equation}
\item $1_x: \mathbbm{1}_{\mathcal{A}}\rightarrow  {^{\mathcal{A}}\mathcal{L}}(x,x)
$ is the identity in $\mathcal{L}(x,x)$.
    \end{itemize}
\end{definition}

It's straightforward to verify $\cL$ is a well-defined category, and we will call $\cL$ the \emph{underlying category} of $^{\cA}\cL$.

\begin{definition}
    Let $^\cA\cL$ and $^\cB\cM$ be two enriched categories, their Cartesian product  $^\cA\cL\times\ ^\cB\cM$ is a $\cA\times \cB$ enriched category defined as follows
    \begin{itemize}
        \item $\ob(^\cA\cL\times\ ^\cB\cM)=\ob(^\cA\cL)\times \ob(^\cB\cM)$
        \item for $(x_1,y_1),(x_2,y_2)\in \ob(^\cA\cL\times\ ^\cB\cM)$, the hom-object is defined as 
        \begin{equation}
            (^\cA\cL\times\ ^\cB\cM)((x_1,y_1),(x_2,y_2)):= (^\cA\cL(x_1,x_2), \ ^\cB\cM(y_1,y_2))\in \cA\times \cB
        \end{equation}
        \item the composition and identity morphism are both defined component-wise.
    \end{itemize}
\end{definition}

\begin{definition}
Let $ {}^{\mathcal{A}}\mathcal{L}$ and $ {}^{\mathcal{B}}\mathcal{M}$ be two enriched categories. An enriched functor $F:\ ^{\cA}\cL\rightarrow \ ^{\cB}\cM$  consists of the following data:
\begin{itemize}
    \item A monoidal functor $\hat{F}: \mathcal{A}\rightarrow \mathcal{B}$, which is so-called a background changing functor;
    \item A map $F: \ob(\mathcal{L})\rightarrow \ob(\mathcal{M})$
    \item For any $x,y\in \mathcal{L}$, there is a morphism 
\begin{equation*}
    F_{x,y}:\hat{F}({}^{\mathcal{A}}\mathcal{L}(x,y))\rightarrow {}^{\mathcal{B}}\mathcal{M}(F(x),F(y))
\end{equation*}
\end{itemize}
 satisfying the following commuting diagram (we omit the  unitor and  monoidal structure of $\hat{F}$):
\begin{equation}
\input{diagrams/enriched-functor-4044.tex}
\end{equation}

\begin{equation}
\label{functor}
\input{diagrams/functor-identity-4056.tex}
\end{equation}
\end{definition}
An enriched functor $F:\ ^{\cA}\cL\rightarrow \ ^{\cB}\cM$ can induce an ordinary functor between the underlying categories $\underline{F}:\cL\rightarrow \cM$. On object level $\underline{F} = F$. For every morphism $f\in \mathcal{L}(x,y): \mathbbm{1}_{\mathcal{A}}\rightarrow  {}^{\mathcal{A}}\mathcal{L}(x,y)$, $\underline{F}(f)$ is defined by the following composition

\begin{equation}
\input{diagrams/underlying-functor-4070.tex}
\end{equation}

It preserves the composition and identity morphism.

\begin{definition}

Let $F,G: {}^{\mathcal{A}}\mathcal{L} \rightarrow {}^{\mathcal{B}}\mathcal{M}$ be two enriched functors,  an enriched natural transformation $\xi: F\Rightarrow G$ consists of the following data:
\begin{itemize}
    \item a monoidal natural transformation $\hat{\xi}:\hat{F}\rightarrow \hat{G}$, which is called \emph{the background changing natural transformation}
    \item a family of morphisms:
\begin{equation*}
    \xi_{x}: \mathbbm{1}_{\mathcal{B}}\rightarrow {}^{\mathcal{B}}\mathcal{M}(F(x),G(x)),\ \ \forall x \in {}^{\mathcal{A}}\mathcal{L}
\end{equation*}
\end{itemize}
satisfying the following commuting diagram
\begin{equation}
\label{NT}
\input{diagrams/enriched-transformation-4096.tex}
\end{equation}
\end{definition}
An enriched natural transformation $\xi:F\Rightarrow G$ automatically gives an ordinary natural transformation between induced ordinary functors $\underline{\xi}:\underline{F}\Rightarrow \underline{G}$, where ${\underline{\xi}}_{x}\in \cM(F(x),G(x)):=\cB(\mathbbm{1}_\cB, \ ^{\cB}\cM(F(x),G(x)))$ is just $\xi_{x}$.

\begin{definition}
    Let $\cA$ be a monoidal category and $\cL$ be a left $\cA$ module, with module action $\odot: \cA\times \cL\rightarrow \cL$. For $x,y\in \cL$, if the functor $\cL(-\odot x,y):\cA^\op\rightarrow \catset$ is representable, we call the representing object as internal hom of $x$ and $y$, denoted by $[x,y]$, i.e.
there is a natural isomorphism 
\begin{equation}
    \cL(-\odot x,y)\cong \cA(-,[x,y]).
\end{equation}
\end{definition}
When internal hom exists, we define unit $u_x:\one_\cA\rightarrow [x,x]$ as the image of $1_x$ under the isomorphism, and evaluation $\ev_{x,y}:[x,y]\odot x\rightarrow y$ as the preimage of $1_{[x,y]}$ under the isomorphism. For $x,y,z\in \cL$, there is a composition of internal homs
\begin{equation}
   \circ: [y,z]\ot [x,y]\rightarrow [x,z]
\end{equation}
defined as the image of $\ev_{y,z}$ under the map 
\begin{equation}
    \cL([y,z]\odot y, z)\xrightarrow{\cL([y,z]\odot \ev_{x,y},z)} \cL([y,z]\odot ([x,y]\odot x),z)\cong \cL([y,z]\ot [x,y]\odot x,z)\cong \cA([y,z]\ot[x,y],[x,z]).
\end{equation}
The composition of internal homs is associative and compatible with units.

\begin{lemma}[Canonical Construction]
\label{lem.cc}
      Let $\cA$ be a monoidal category and $\cL$ be a left $\cA$ module, if $[x,y]$ exists for any $x,y\in \cL$, we can define an enriched category $^{\cA}\cL$, where
    \begin{itemize}
         \item $\mathrm{Ob}(^{\cA}\cL) = \mathrm{Ob}(\cL)$
    \item hom-object $^{\cA}\cL(x,y):=[x,y]$
    \item identity morphism is $u_x:\one_\cA\rightarrow [x,x]$
    \item composition is $\circ: [y,z]\ot [x,y]\rightarrow [x,z]$
    \end{itemize}
\end{lemma}
  We will call the construction in Lemma~\ref{lem.cc} the canonical construction of enriched categories. The underlying category of the canonical construction $^{\cA}\cL$ is equivalent to $\cL$.
  \begin{theorem}[Theorem 3.45 in \cite{Kong_Yuan_2024_EnrichedMonoidalCat}]
  \label{thm.cc_ECat}
      The canonical construction defines a 2-equivalence from $\mathbf{fsLMod}$ to $\mathbf{fsECat}$, where $\mathbf{fsLMod}$ is the 2-category of  finite semisimple left module categories over finite semisimple monoidal categories, and $\mathbf{fsECat}$ is the 2-category of finite semisimple enriched categories, linear enriched functors and enriched natural transformations.
  \end{theorem}

  Let $(\cA,\ot_\cA,c)$ be a braided monoidal category, where $c$ is the braiding, then the tensor product functor $\ot_\cA: \cA\times \cA \rightarrow \cA$ can be equipped with a structure of monoidal functor, we will use the convention of the following monoidal structure 
  \begin{equation}
      \ot_\cA(a_1,a_2) \ot_\cA \ot_\cA(a_1',a_2')=(a_1\ot_\cA a_2) \ot_\cA (a_1'\ot_\cA a_2') \xrightarrow{ c_{a_2,a_1'}} (a_1\ot_\cA a_1')\ot_\cA (a_2\ot_\cA a_2') = \ot_\cA (a_1\ot_\cA a_1', a_2\ot_\cA a_2').
  \end{equation}
\begin{definition}
    An $\cA$-enriched monoidal category consists of the following data
    \begin{itemize}
    \item an $\cA$-enriched category $^\cA\cL$;
    \item a tensor product enriched functor 
    $\ot:\ ^\cA\cL\times\ ^\cA\cL \rightarrow\ ^\cA\cL $ such that $\hat{\ot}=\ot_\cA$;
     \item an enriched functor $I:\ast\rightarrow\ ^\cA\cL$, where $\ast$ is the enriched category with one object, and $I(\ast)=\one$ is called tensor unit of $^\cA\cL$;
     \item an enriched natural isomorphism $\alpha: \ot\circ (\ot \times \id)\Rightarrow \ot \circ( \id \times \ot)$ called associator such that $\hat{\alpha}$ is the associator in $\cA$;
     \item two enriched natural isomorphism $\lambda: \ot \circ (I\times \id)\Rightarrow \id$, $\rho:  \ot\circ (\id \times I)\Rightarrow \id$ called the left and right unitors, such that $\hat{\lambda}$ and $\hat{\rho}$ are the left and right unitors in $\cA$,
     \end{itemize}
     such that $(\cL,\underline{\ot}, \one, \underline{\alpha},\underline{\lambda},\underline{\rho})$ is a monoidal category, which is called the underlying monoidal category of $(^{\cA}\cL, \ot,I,\alpha,\lambda,\rho)$.
\end{definition}
\begin{definition}
    Let $^{\cA}\cL$, $^{\cB}\cM$ be two enriched monoidal categories, an enriched monoidal functor $F:\ ^{\cA}\cL\rightarrow\ ^{\cB}\cM$ consists of the following data
    \begin{itemize}
        \item an enriched functor $F:\ ^{\cA}\cL\rightarrow\ ^{\cB}\cM$;
        \item an enriched isomorphism $F^2: \ot\circ (F\times F)\Rightarrow F\circ \otimes $, such that $\hat{F^2}$ is the monoidal structure of $\hat{F}$;
        \item an enriched isomorphism $F^0: I_\cM\Rightarrow F\circ I_\cL$ such that $\hat{F^0}$ is the monoidal structure of $\hat{F}$;
    \end{itemize}
    such that 
    \begin{itemize}
        \item $\hat{F}:\cA\rightarrow \cB$ is a braided monoidal functor 
        \item $(\underline{F},\underline{F^2},\underline{F^0}):\cL\rightarrow \cM$ is a monoidal functor between the underlying monoidal categories.
    \end{itemize}
\end{definition}
\begin{definition}
    Let $F,G:\ ^\cA\cL\rightarrow\ ^\cB\cM$ be two enriched monoidal functors between two enriched monoidal categories, an enriched monoidal natural transformation $\xi:F\Rightarrow G$ is an enriched natural transformation such that the underlying natural transformation $\underline{\xi}:\underline{F}\Rightarrow \underline{G}$ is monoidal.
\end{definition}

\begin{definition}
    Let $\cA$ be a braided monoidal category, $\cL$ be a monoidal category, and $H:\cA\rightarrow \cL$ is a monoidal functor. A central functor structure of $H$ is a braided monoidal functor $\varphi:\cA\rightarrow Z_1(\cL)$ such that
    \begin{equation*}
        (\cA\xrightarrow{\varphi}Z_1(\cL)\xrightarrow{\mathrm{forget}}\cL) \cong H
    \end{equation*}
    as monoidal functors.
\end{definition}
A central functor $H:\cA\rightarrow \cL$ equips $\cL$ a monoidal left $\cA$ module structure, given by 
\begin{equation}
\begin{split}
    \odot: \cA\times \cL&\rightarrow \cL\\
    (a,x)&\mapsto a\odot x:= \varphi(a)\otimes x.
\end{split}
\end{equation}
Conversely, if $\cL$ is a monoidal left $\cA$ module with action $\odot$, then $\odot \one:\cA\rightarrow \cL$ has a canonical central structure,
\begin{equation}
    x\otimes (a\odot \one)\cong (\one_\cA\odot x) \otimes (a\odot \one)\cong (\one_\cA\otimes_\cA a)\odot (x\otimes \one)\cong (a\ot_\cA \one_\cA)\odot(\one\otimes x)\cong (a\odot \one)\otimes x,
\end{equation}
and these two constructions are inverse of each other.

\begin{proposition}[Proposition 4.17 in \cite{Kong_Yuan_2024_EnrichedMonoidalCat}]
\label{prp.MonoidalCC}
Given a braided monoidal category $\cA$ and a left  $\cA$-module $\cL$. Let ${}^\cA\cL$ be the enriched category given by the canonical construction. There is a one-to-one correspondence between the monoidal structures on ${}^\cA\cL$ and the monoidal structures on the left $\overline{\cA}$-oplax module  $\cL$. 
\end{proposition}

\section{Orientation conventions}\label{app:orientation}
In this appendix, we collect orientation conventions in this paper and compare them to those in the literature.

There are three orientation choices that affect the formulas in the main text. 
The first is whether module categories are left or right module categories. 
The second is the orientation of the one-dimensional chain used in the DHR construction. 
The third is the direction in which topological defects are fused. 
None of these choices changes the physical content, but they change whether one writes a category, its reverse, or its opposite. 
The convention translations used below are: EGNO's left-module statements are converted by a left/right module conversion; Jones' bulk DHR identification is converted by a chain orientation reversal; the source $(\CMdual)^{\rev}$ of $\ReaDHR$ records the monoidal reverse coming from placing module functors on the left in the standard action; and the source $\MQ^\op$ of $\ReaBCond$ records the opposite category forced by contravariance.

\subsection{Right module categories}\label{app:orientation:right-modules}

Throughout this paper, a module category means a right module category unless otherwise stated. Thus the boundary local quantum charge category is a right $\C$-module category $\M$, and the fusion spin chain with boundary has local algebras
\[
    \Loc_{[0,j]}^{\bdy}=\M(l\otimes x^j,l\otimes x^j).
\]
This is the convention used in Section~\ref{BFusionSpinChain}. Many references, including standard categorical references such as \cite{EGNO_2015}, more often state analogous results for left module categories. Passing between the two conventions is a left/right module conversion: it reverses the side on which the tensor product acts and may insert $\rev$ in formulas involving dual categories of module endofunctors.

For example, if $\M\simeq {}_W\C$ is written as a category of right $W$-modules internal to $\C$, then the dual category used in this paper satisfies
\[
    \CMdual\simeq \WCW,
\]
whereas the corresponding statement in a left-module convention is usually written with an additional reverse. This is the convention translation used when comparing $\ReaDHR$ with the functor $G$ constructed from $\WCW$; the additional $\rev$ in the source of the realization functor
\[
    \ReaDHR:(\CMdual)^{\rev}\longrightarrow \DHR(\Loc_\bullet^\bdy)
\]
in Theorem~\ref{thm:intro-boundary-symtft} comes instead from the monoidal reverse in the standard action.

\subsection{Boundary conditions and bulk topological defects}\label{app:orientation:defects}

The realization functor for boundary conditions has source $\MQ^\op$:
\[
    \ReaBCond: \MQ^\op\longrightarrow \BCond.
\]
This is not only a matter of notation. As explained after Definition~\ref{Def. The realization functor for boundary conditions}, the opposite category is needed because the construction sends a right $Q$-module $K$ to the left Hilbert module obtained from morphisms $\M(K,l\otimes x^n)$. It is also the convention compatible with the action of the boundary SymTFT on boundary conditions in Theorem~\ref{thm:intro-bulk-boundary-action}.

For bulk topological defects, the underlying linear category obtained from the same Q-system model is $(\QCQ)^\op$; see Remark~\ref{rmk:TopDef}. To make the fusion direction compatible with the enriched-center convention of \cite{Kong_Yuan_2024_EnrichedMonoidalCat}, we use the monoidal category
\[
    \mathrm{TopDef}\simeq (\QCQ)^{\rev,\op}.
\]
Equivalently, defects are fused from right to left in our convention. If one chooses the opposite spatial orientation, one may instead write $(\QCQ)^\op$ for the same underlying defect category with the other fusion direction.

This differs from the convention in \cite{Kong_2022_1DEnrichedCat,Xu_2024_1DPhaseAbelianSym,LanZhou_2024_QuantumCurrent}, where the bulk defect category is written as $(\QCQ)^\rev$. The discrepancy comes from the choice of time orientation. After reversing time and using rigidity, $(\QCQ)^{\rev,\op}$ and $\QCQ$ give the same physical defect data.

\subsection{The DHR category}\label{app:orientation:dhr}

Jones' DHR construction for fusion spin chains identifies the DHR category with the Drinfeld center of the input category \cite{Jones_2024_DHR}:
\[
    \DHR(\Loc_\bullet)\simeq Z_1(\C).
\]
Our standard action in Section~\ref{StandardAction} uses the opposite orientation of the one-dimensional chain relative to that convention. Equivalently, one may regard our input category for the oriented chain as $\C^\rev$. Hence the bulk SymTFT appearing in this paper is
\[
    \DHR(\Loc_\bullet^\bulk)\simeq Z_1(\C^\rev)\simeq \overline{Z_1(\C)}.
\]
Here the bar denotes the same braided category with the braiding reversed. This is the reason why Theorem~\ref{thm:SymTFT-bulk-bdy} identifies the bulk action on the boundary DHR category with the canonical action of $Z_1(\C^\rev)$ on $(\CMdual)^{\rev}$.

The same convention also explains the formula for the standard boundary action. We define finite pieces by
\[
    \Std_n(L)=\M(l\otimes x^n, L\otimes l\otimes x^n),
\]
and the monoidal structure maps an elementary tensor as
\[
    \Std_n(L_1) \boxtimes_{\Loc^{\bdy}_n} \Std_n(L_2) \ni \xi\otimes \eta \longmapsto (1_{L_2} \otimes \xi)\circ \eta \in \Std_n(L_2\otimes L_1).
\]
Thus the product in the source is reversed, and the standard action is a functor out of $(\CMdual)^{\rev}$. This differs from \cite{Jones_2024_DHR}, where the corresponding module functor is placed on the other side.

\subsection{Enriched categories}\label{app:orientation:enriched}

The enriched-category part of the paper follows the conventions of \cite{Kong_Yuan_2024_EnrichedMonoidalCat}, reviewed in Appendix~\ref{app.enriched_cat}. In particular, the canonical construction uses left module categories over the background category. Since the microscopic boundary SymTFT in this paper is $(\CMdual)^{\rev}$ acting on $\MQ^\op$, the boundary is written as the enriched category
\[
    {}^{(\CMdual)^\rev}\MQ^\op.
\]
The enriched center theorem then gives
\[
    Z_0\bigl({}^{(\CMdual)^\rev}\MQ^\op\bigr)
    \simeq {}^{Z_1(\C^\rev)}(\QCQ)^{\rev,\op},
\]
which is the categorical form of the bulk side used in Equation~\eqref{eq:bulk enriched cat vs category theory}.

Therefore the final bulk-boundary statement in Theorem~\ref{thm:bulk-boundary-correspondence} is the same enriched-center statement as in \cite{Kong_2022_1DEnrichedCat,Xu_2024_1DPhaseAbelianSym,LanZhou_2024_QuantumCurrent}, after translating their convention to ours by reversing the orientation of the chain and the corresponding time direction.

\bibliographystyle{hyperalpha}
\bibliography{bib}

\end{document}

%% file: diagrams/gns-isometry-0479.tex
\begin{tikzcd}
	{A_i \times \mathcal{H}_{\rho_i}} &&&& {A_j \times \mathcal{H}_{\rho_j}} \\
	\\
	\\
	{\mathcal{H}_{\rho_i}} &&&& {\mathcal{H}_{\rho_j}}
	\arrow["{\phi_{ij} \times \tilde \phi_{ij}}", from=1-1, to=1-5]
	\arrow["{\text{module action}}"', from=1-1, to=4-1]
	\arrow["{\text{module action}}"', from=1-5, to=4-5]
	\arrow["{\tilde\phi_{ij}}", from=4-1, to=4-5]
\end{tikzcd}

%% file: diagrams/standard-action-0788.tex
\tikzset{every picture/.style={line width=0.75pt}} %set default line width to 0.75pt

\begin{tikzpicture}[x=0.75pt,y=0.75pt,yscale=-1,xscale=1]
%uncomment if require: \path (0,221); %set diagram left start at 0, and has height of 221

%Shape: Rectangle [id:dp06652131263546035]
\draw   (249.87,116.2) -- (337.27,116.2) -- (337.27,132.6) -- (249.87,132.6) -- cycle ;
%Straight Lines [id:da1798551463574921]
\draw    (294.75,116.67) -- (294.1,63.25) ;
%Straight Lines [id:da6109264495920503]
\draw    (328.5,115.95) -- (327.85,62.53) ;
%Straight Lines [id:da4637150683721014]
\draw    (276.2,116.1) -- (275.55,62.68) ;
%Straight Lines [id:da9484826695563863]
\draw    (295.25,172.67) -- (294.6,133.01) ;
%Straight Lines [id:da16714836432732505]
\draw    (329,172.14) -- (328.35,132.48) ;
%Straight Lines [id:da631110613972884]
\draw    (276.7,172.24) -- (276.05,132.58) ;
%Straight Lines [id:da9725051477470585]
\draw    (256.5,115.95) -- (255.85,62.53) ;
%Shape: Rectangle [id:dp2714526057613075]
\draw   (438.37,116.7) -- (525.77,116.7) -- (525.77,133.1) -- (438.37,133.1) -- cycle ;
%Straight Lines [id:da679099932309548]
\draw    (486.25,117.17) -- (485.6,60.99) ;
%Straight Lines [id:da158786505363671]
\draw    (520,116.42) -- (519.35,60.24) ;
%Straight Lines [id:da8026534704509869]
\draw    (467.7,116.57) -- (467.05,60.39) ;
%Straight Lines [id:da6695590132941431]
\draw    (486.75,173.17) -- (486.1,133.51) ;
%Straight Lines [id:da7086642195693698]
\draw    (520.5,172.64) -- (519.85,132.98) ;
%Straight Lines [id:da6966728822187357]
\draw    (468.2,172.74) -- (467.55,133.08) ;
%Shape: Rectangle [id:dp314110992908067]
\draw   (437.32,77.84) -- (456.17,77.84) -- (456.17,95.84) -- (437.32,95.84) -- cycle ;
%Straight Lines [id:da290588365181799]
\draw    (446.83,116.85) -- (446.75,96.34) ;
%Straight Lines [id:da27743133028626343]
\draw    (447.17,77.84) -- (447.25,60.84) ;

% Text Node
\draw (37.33,104.73) node [anchor=north west][inner sep=0.75pt]  [font=\large]  {$\Std_n \left( L_{1}\xrightarrow{f} L_{2}\right) :$};
% Text Node
\draw (290.27,119.8) node [anchor=north west][inner sep=0.75pt]    {$a$};
% Text Node
\draw (272,42.8) node [anchor=north west][inner sep=0.75pt]    {$l$};
% Text Node
\draw (289.63,44.94) node [anchor=north west][inner sep=0.75pt]  [font=\normalsize]  {$x$};
% Text Node
\draw (322.7,44.71) node [anchor=north west][inner sep=0.75pt]  [font=\normalsize]  {$x$};
% Text Node
\draw (304.23,76.38) node [anchor=north west][inner sep=0.75pt]    {$...$};
% Text Node
\draw (361.5,116.23) node [anchor=north west][inner sep=0.75pt]  [font=\Large]  {$\longmapsto $};
% Text Node
\draw (273.5,174.87) node [anchor=north west][inner sep=0.75pt]    {$l$};
% Text Node
\draw (291.13,176.49) node [anchor=north west][inner sep=0.75pt]  [font=\normalsize]  {$x$};
% Text Node
\draw (324.2,176.32) node [anchor=north west][inner sep=0.75pt]  [font=\normalsize]  {$x$};
% Text Node
\draw (304.73,140.8) node [anchor=north west][inner sep=0.75pt]    {$...$};
% Text Node
\draw (250.7,41.06) node [anchor=north west][inner sep=0.75pt]  [font=\normalsize]  {$L_{1}$};
% Text Node
\draw (478.77,120.3) node [anchor=north west][inner sep=0.75pt]    {$a$};
% Text Node
\draw (463.5,39.88) node [anchor=north west][inner sep=0.75pt]    {$l$};
% Text Node
\draw (481.13,41.92) node [anchor=north west][inner sep=0.75pt]  [font=\normalsize]  {$x$};
% Text Node
\draw (514.2,41.68) node [anchor=north west][inner sep=0.75pt]  [font=\normalsize]  {$x$};
% Text Node
\draw (495.73,75.02) node [anchor=north west][inner sep=0.75pt]    {$...$};
% Text Node
\draw (465,175.37) node [anchor=north west][inner sep=0.75pt]    {$l$};
% Text Node
\draw (482.63,176.99) node [anchor=north west][inner sep=0.75pt]  [font=\normalsize]  {$x$};
% Text Node
\draw (515.7,176.82) node [anchor=north west][inner sep=0.75pt]  [font=\normalsize]  {$x$};
% Text Node
\draw (496.23,141.3) node [anchor=north west][inner sep=0.75pt]    {$...$};
% Text Node
\draw (432.03,101.06) node [anchor=north west][inner sep=0.75pt]  [font=\footnotesize]  {$L_{1}$};
% Text Node
\draw (441.5,81.88) node [anchor=north west][inner sep=0.75pt]  [font=\footnotesize]  {$f$};
% Text Node
\draw (440.2,39.56) node [anchor=north west][inner sep=0.75pt]  [font=\normalsize]  {$L_{2}$};

\end{tikzpicture}

%% file: diagrams/bimodule-action-0917.tex
\begin{tikzpicture}[x=0.75pt,y=0.75pt,yscale=-1,xscale=1]
%uncomment if require: \path (0,201); %set diagram left start at 0, and has height of 201

%Shape: Rectangle [id:dp9074665999381528]
\draw   (349.34,91.34) -- (492,91.34) -- (492,110.33) -- (349.34,110.33) -- cycle ;
%Straight Lines [id:da7789749645119602]
\draw    (418.09,62.16) -- (418.17,91.11) ;
%Straight Lines [id:da3864559295570723]
\draw    (461.73,61.71) -- (461.21,91.5) ;
%Straight Lines [id:da3241043912138082]
\draw    (371.54,19.3) -- (370.64,91.31) ;
%Straight Lines [id:da37157664080206365]
\draw    (396.05,110.45) -- (396.05,135.27) ;
%Straight Lines [id:da0447781519636099]
\draw    (438.2,110.27) -- (438.2,135.09) ;
%Shape: Rectangle [id:dp42309482469493087]
\draw   (407.33,45.49) -- (470.1,45.49) -- (470.1,61.71) -- (407.33,61.71) -- cycle ;
%Shape: Rectangle [id:dp8538870962933545]
\draw   (385.66,135.27) -- (448.43,135.27) -- (448.43,151.79) -- (385.66,151.79) -- cycle ;
%Straight Lines [id:da9017218114882871]
\draw    (417.72,18.62) -- (417.72,45.1) ;
%Straight Lines [id:da016223614998364888]
\draw    (460.76,19.01) -- (460.76,45.49) ;
%Straight Lines [id:da1007488857421921]
\draw    (395.15,152.29) -- (395.15,174.64) ;
%Straight Lines [id:da3892078382319647]
\draw    (438.65,152.05) -- (438.65,174.4) ;

% Text Node
\draw (406.71,70.01) node [anchor=north west][inner sep=0.75pt]    {$l$};
% Text Node
\draw (442.22,68.67) node [anchor=north west][inner sep=0.75pt]    {$x^{n}$};
% Text Node
\draw (414.18,94.05) node [anchor=north west][inner sep=0.75pt]    {$\xi $};
% Text Node
\draw (434.06,49.16) node [anchor=north west][inner sep=0.75pt]    {$a$};
% Text Node
\draw (411.49,136.72) node [anchor=north west][inner sep=0.75pt]    {$b$};
% Text Node
\draw (406.71,17.61) node [anchor=north west][inner sep=0.75pt]    {$l$};
% Text Node
\draw (442.22,16.28) node [anchor=north west][inner sep=0.75pt]    {$x^{n}$};
% Text Node
\draw (355.38,16.27) node [anchor=north west][inner sep=0.75pt]    {$L$};
% Text Node
\draw (197.08,90.38) node [anchor=north west][inner sep=0.75pt]  [font=\large]  {$a\triangleright \xi \triangleleft b\quad :=$};
% Text Node
\draw (384.09,163.69) node [anchor=north west][inner sep=0.75pt]    {$l$};
% Text Node
\draw (419.6,162.41) node [anchor=north west][inner sep=0.75pt]    {$x^{n}$};
% Text Node
\draw (384.09,119.18) node [anchor=north west][inner sep=0.75pt]    {$l$};
% Text Node
\draw (419.6,117.89) node [anchor=north west][inner sep=0.75pt]    {$x^{n}$};

\end{tikzpicture}

%% file: diagrams/inner-product-0978.tex
\begin{tikzpicture}[x=0.75pt,y=0.75pt,yscale=-1,xscale=1]
%uncomment if require: \path (0,176); %set diagram left start at 0, and has height of 176

%Shape: Rectangle [id:dp14335840236409747]
\draw   (368.67,119.8) -- (511.33,119.8) -- (511.33,138.79) -- (368.67,138.79) -- cycle ;
%Straight Lines [id:da7332498577275243]
\draw    (437.42,90.62) -- (437.5,119.57) ;
%Straight Lines [id:da0539773148893149]
\draw    (481.06,90.17) -- (480.55,119.96) ;
%Straight Lines [id:da2802250982727017]
\draw    (415.38,138.91) -- (415.38,163.73) ;
%Straight Lines [id:da34212171044870177]
\draw    (457.53,138.73) -- (457.53,163.55) ;
%Shape: Rectangle [id:dp22605286961958537]
\draw   (367.34,70.8) -- (510,70.8) -- (510,89.79) -- (367.34,89.79) -- cycle ;
%Straight Lines [id:da8640138267356017]
\draw    (415.39,43.75) -- (415.39,70.23) ;
%Straight Lines [id:da04345023131590997]
\draw    (458.43,44.14) -- (458.43,70.62) ;
%Straight Lines [id:da48576577599848647]
\draw    (388.76,90.95) -- (388.84,119.91) ;

% Text Node
\draw (250.33,93.57) node [anchor=north west][inner sep=0.75pt]  [font=\large]  {$\langle \xi |\eta \rangle \quad :=$};
% Text Node
\draw (426.04,98.47) node [anchor=north west][inner sep=0.75pt]    {$l$};
% Text Node
\draw (461.55,97.13) node [anchor=north west][inner sep=0.75pt]    {$x^{n}$};
% Text Node
\draw (433.51,122.51) node [anchor=north west][inner sep=0.75pt]    {$\eta $};
% Text Node
\draw (403.43,147.63) node [anchor=north west][inner sep=0.75pt]    {$l$};
% Text Node
\draw (438.94,146.35) node [anchor=north west][inner sep=0.75pt]    {$x^{n}$};
% Text Node
\draw (432.18,73.51) node [anchor=north west][inner sep=0.75pt]    {$\xi ^{*}$};
% Text Node
\draw (404.37,42.74) node [anchor=north west][inner sep=0.75pt]    {$l$};
% Text Node
\draw (439.88,41.4) node [anchor=north west][inner sep=0.75pt]    {$x^{n}$};
% Text Node
\draw (376.37,98.8) node [anchor=north west][inner sep=0.75pt]    {$L$};

\end{tikzpicture}

%% file: diagrams/screening-map-1293.tex
\begin{tikzpicture}[x=0.75pt,y=0.75pt,yscale=-1,xscale=1]

%Straight Lines [id:da2328714574174171]
\draw    (391.8,146.2) -- (390.48,223.11) -- (389.8,261.4) ;
%Straight Lines [id:da43926872313551724]
\draw    (432.2,146.6) -- (390.96,185.61) ;
%Straight Lines [id:da6231031801341773]
\draw    (461.4,146.6) -- (391.25,207.4) ;
%Straight Lines [id:da08478019777347334]
\draw    (477.45,156.1) -- (390.48,230.11) ;
%Shape: Triangle [id:dp9861883559964102]
\draw  [fill={rgb, 255:red, 255; green, 255; blue, 255 }  ,fill opacity=1 ] (425.71,177.35) -- (433.36,176.08) -- (428.68,170.54) -- cycle ;
%Shape: Triangle [id:dp538706821816028]
\draw  [fill={rgb, 255:red, 255; green, 255; blue, 255 }  ,fill opacity=1 ] (410.05,167.87) -- (417.71,166.6) -- (413.02,161.06) -- cycle ;
%Straight Lines [id:da6908254855787712]
\draw    (476.05,125) -- (390.6,39.8) ;
%Straight Lines [id:da07295020809953157]
\draw    (458,129.4) -- (391,64.5) ;
%Straight Lines [id:da03734009375966241]
\draw    (426,130) -- (391.5,90.5) ;
%Shape: Triangle [id:dp24013956887821852]
\draw  [fill={rgb, 255:red, 255; green, 255; blue, 255 }  ,fill opacity=1 ] (419.12,91.6) -- (421.49,98.99) -- (426.3,93.55) -- cycle ;
%Shape: Triangle [id:dp8292661749861381]
\draw  [fill={rgb, 255:red, 255; green, 255; blue, 255 }  ,fill opacity=1 ] (407.06,108.08) -- (408.83,115.64) -- (414.05,110.6) -- cycle ;
%Shape: Rectangle [id:dp1109698265327479]
\draw   (308.33,129.8) -- (468.6,129.8) -- (468.6,146.2) -- (308.33,146.2) -- cycle ;
%Straight Lines [id:da998212525592724]
\draw    (390.6,15) -- (391.05,107.8) -- (391.2,129.6) ;
%Curve Lines [id:da7681744254688603]
\draw    (476.05,125) .. controls (475.67,124.53) and (481.4,129.8) .. (485,135.8) .. controls (488.6,141.8) and (484.2,150.2) .. (477.45,156.1) ;
%Shape: Triangle [id:dp06955065406182648]
\draw  [fill={rgb, 255:red, 255; green, 255; blue, 255 }  ,fill opacity=1 ] (394.95,114.03) -- (391.05,107.8) -- (387.09,113.89) -- cycle ;
%Shape: Triangle [id:dp015154498967470342]
\draw  [fill={rgb, 255:red, 255; green, 255; blue, 255 }  ,fill opacity=1 ] (395.59,157.92) -- (387.83,158.06) -- (391.66,164.22) -- cycle ;
%Straight Lines [id:da8361609384448809]
\draw    (343.53,147) -- (391,193) ;
%Straight Lines [id:da5532105579454207]
\draw    (315.29,147.33) -- (390.5,217) ;
%Straight Lines [id:da32723112231359863]
\draw    (294.21,154.8) -- (390,240.5) ;
%Shape: Triangle [id:dp6913554943865754]
\draw  [fill={rgb, 255:red, 255; green, 255; blue, 255 }  ,fill opacity=1 ] (348.03,177.75) -- (340.37,176.48) -- (345.06,170.94) -- cycle ;
%Shape: Triangle [id:dp08091769111075886]
\draw  [fill={rgb, 255:red, 255; green, 255; blue, 255 }  ,fill opacity=1 ] (364.68,167.27) -- (357.03,166) -- (361.71,160.46) -- cycle ;
%Straight Lines [id:da884211944233997]
\draw    (294.68,125.4) -- (391,28.5) ;
%Straight Lines [id:da9080969240966476]
\draw    (311.73,129.8) -- (390.5,51) ;
%Straight Lines [id:da510247761611698]
\draw    (338.93,129.6) -- (390.5,79.5) ;
%Shape: Triangle [id:dp8241929751392575]
\draw  [fill={rgb, 255:red, 255; green, 255; blue, 255 }  ,fill opacity=1 ] (350.61,91) -- (348.24,98.39) -- (343.44,92.95) -- cycle ;
%Shape: Triangle [id:dp17376527744225168]
\draw  [fill={rgb, 255:red, 255; green, 255; blue, 255 }  ,fill opacity=1 ] (363.68,105.48) -- (361.9,113.04) -- (356.68,108) -- cycle ;
%Shape: Arc [id:dp553626198299961]
\draw  [draw opacity=0] (294.21,154.8) .. controls (290.22,151.24) and (287.71,146.07) .. (287.71,140.3) .. controls (287.71,134.32) and (290.42,128.96) .. (294.68,125.4) -- (307.13,140.3) -- cycle ; \draw   (294.21,154.8) .. controls (290.22,151.24) and (287.71,146.07) .. (287.71,140.3) .. controls (287.71,134.32) and (290.42,128.96) .. (294.68,125.4) ;

% Text Node
\draw (376.2,251) node [anchor=north west][inner sep=0.75pt]  [font=\footnotesize]  {$Q$};
% Text Node
\draw (379.7,166.2) node [anchor=north west][inner sep=0.75pt]  [font=\scriptsize]  {$m ^{\dagger }$};
% Text Node
\draw (378.9,214.9) node [anchor=north west][inner sep=0.75pt]  [font=\scriptsize]  {$m ^{\dagger }$};
% Text Node
\draw (378.48,190.51) node [anchor=north west][inner sep=0.75pt]  [font=\scriptsize]  {$m ^{\dagger }$};
% Text Node
\draw (391.6,24.5) node [anchor=north west][inner sep=0.75pt]  [font=\scriptsize]  {$m $};
% Text Node
\draw (392.3,51.7) node [anchor=north west][inner sep=0.75pt]  [font=\scriptsize]  {$m $};
% Text Node
\draw (392.2,76.3) node [anchor=north west][inner sep=0.75pt]  [font=\scriptsize]  {$m $};
% Text Node
\draw (401.2,177.4) node [anchor=north west][inner sep=0.75pt]    {$...$};
% Text Node
\draw (406.4,98.4) node [anchor=north west][inner sep=0.75pt]    {$...$};
% Text Node
\draw (488.5,128.3) node [anchor=north west][inner sep=0.75pt]  [font=\scriptsize]  {$Q$};
% Text Node
\draw (275.3,127.5) node [anchor=north west][inner sep=0.75pt]  [font=\scriptsize]  {$Q$};
% Text Node
\draw (384.93,135.4) node [anchor=north west][inner sep=0.75pt]    {$a$};
% Text Node
\draw (356.53,176.8) node [anchor=north east][inner sep=0.75pt]  [xscale=-1]  {$...$};
% Text Node
\draw (346.33,99.8) node [anchor=north east][inner sep=0.75pt]  [xscale=-1]  {$...$};
% Text Node
\draw (377.1,5.9) node [anchor=north west][inner sep=0.75pt]  [font=\footnotesize]  {$Q$};
% Text Node
\draw (180.5,128.9) node [anchor=north west][inner sep=0.75pt]  [font=\large]  {$\mathrm{scr}( a) \quad :=$};
% Text Node
\draw (364.3,74) node [anchor=north west][inner sep=0.75pt]  [font=\scriptsize]  {$Q$};
% Text Node
\draw (407.3,72) node [anchor=north west][inner sep=0.75pt]  [font=\scriptsize]  {$Q$};
% Text Node
\draw (324.8,115.5) node [anchor=north west][inner sep=0.75pt]  [font=\footnotesize]  {$x$};
% Text Node
\draw (436.8,117.5) node [anchor=north west][inner sep=0.75pt]  [font=\footnotesize]  {$x$};
% Text Node
\draw (351.3,117) node [anchor=north west][inner sep=0.75pt]  [font=\footnotesize]  {$x$};
% Text Node
\draw (407.3,118.5) node [anchor=north west][inner sep=0.75pt]  [font=\footnotesize]  {$x$};
% Text Node
\draw (394.3,88.5) node [anchor=north west][inner sep=0.75pt]  [font=\scriptsize]  {$Q$};
% Text Node
\draw (370.8,95.5) node [anchor=north west][inner sep=0.75pt]  [font=\scriptsize]  {$Q$};
% Text Node
\draw (323.8,149.5) node [anchor=north west][inner sep=0.75pt]  [font=\footnotesize]  {$x$};
% Text Node
\draw (442.8,149.5) node [anchor=north west][inner sep=0.75pt]  [font=\footnotesize]  {$x$};
% Text Node
\draw (352.3,149) node [anchor=north west][inner sep=0.75pt]  [font=\footnotesize]  {$x$};
% Text Node
\draw (414.3,149.5) node [anchor=north west][inner sep=0.75pt]  [font=\footnotesize]  {$x$};
% Text Node
\draw (396.3,163.5) node [anchor=north west][inner sep=0.75pt]  [font=\scriptsize]  {$Q$};
% Text Node
\draw (404.8,192) node [anchor=north west][inner sep=0.75pt]  [font=\scriptsize]  {$Q$};
% Text Node
\draw (367.3,162) node [anchor=north west][inner sep=0.75pt]  [font=\scriptsize]  {$Q$};
% Text Node
\draw (356.3,193.5) node [anchor=north west][inner sep=0.75pt]  [font=\scriptsize]  {$Q$};

\end{tikzpicture}

%% file: diagrams/frobenius-module-1436.tex
\begin{tikzpicture}[x=0.75pt,y=0.75pt,yscale=-1,xscale=1]
%uncomment if require: \path (0,313); %set diagram left start at 0, and has height of 313

%Straight Lines [id:da7438025785744227]
\draw    (81.8,165.1) -- (80.48,242.01) -- (79.8,280.3) ;
%Straight Lines [id:da8669102427179995]
\draw    (122.2,165.5) -- (80.96,204.51) ;
%Straight Lines [id:da49192713075396577]
\draw    (151.4,165.5) -- (81.25,226.3) ;
%Straight Lines [id:da8288446422997828]
\draw    (167.45,175) -- (80.48,249.01) ;
%Shape: Triangle [id:dp574843806834555]
\draw  [fill={rgb, 255:red, 255; green, 255; blue, 255 }  ,fill opacity=1 ] (115.71,196.25) -- (123.36,194.98) -- (118.68,189.44) -- cycle ;
%Shape: Triangle [id:dp3820973315020433]
\draw  [fill={rgb, 255:red, 255; green, 255; blue, 255 }  ,fill opacity=1 ] (100.05,186.77) -- (107.71,185.5) -- (103.02,179.96) -- cycle ;
%Straight Lines [id:da10519033123223387]
\draw    (166.05,143.9) -- (80.6,58.7) ;
%Straight Lines [id:da9164973276295902]
\draw    (148,148.3) -- (81,83.4) ;
%Straight Lines [id:da9348150340999962]
\draw    (116,148.9) -- (81.5,109.4) ;
%Shape: Triangle [id:dp15800712757608593]
\draw  [fill={rgb, 255:red, 255; green, 255; blue, 255 }  ,fill opacity=1 ] (109.12,110.5) -- (111.49,117.89) -- (116.3,112.45) -- cycle ;
%Shape: Triangle [id:dp5061166768196215]
\draw  [fill={rgb, 255:red, 255; green, 255; blue, 255 }  ,fill opacity=1 ] (97.06,126.98) -- (98.83,134.54) -- (104.05,129.5) -- cycle ;
%Shape: Rectangle [id:dp19123877918262588]
\draw   (-1.67,148.7) -- (158.6,148.7) -- (158.6,165.1) -- (-1.67,165.1) -- cycle ;
%Straight Lines [id:da9328884177633282]
\draw    (80.6,33.9) -- (81.05,126.7) -- (81.2,148.5) ;
%Curve Lines [id:da7097294366465914]
\draw    (166.05,143.9) .. controls (165.67,143.43) and (171.4,148.7) .. (175,154.7) .. controls (178.6,160.7) and (174.2,169.1) .. (167.45,175) ;
%Shape: Triangle [id:dp11026381308157374]
\draw  [fill={rgb, 255:red, 255; green, 255; blue, 255 }  ,fill opacity=1 ] (84.95,132.93) -- (81.05,126.7) -- (77.09,132.79) -- cycle ;
%Shape: Triangle [id:dp4944301939721133]
\draw  [fill={rgb, 255:red, 255; green, 255; blue, 255 }  ,fill opacity=1 ] (85.59,176.82) -- (77.83,176.96) -- (81.66,183.12) -- cycle ;
%Straight Lines [id:da0821254640288992]
\draw    (33.53,165.9) -- (81,211.9) ;
%Straight Lines [id:da5648330439980862]
\draw    (5.29,166.23) -- (80.5,235.9) ;
%Straight Lines [id:da24340227069735787]
\draw    (-15.79,173.7) -- (80,259.4) ;
%Shape: Triangle [id:dp1227586352794594]
\draw  [fill={rgb, 255:red, 255; green, 255; blue, 255 }  ,fill opacity=1 ] (38.03,196.65) -- (30.37,195.38) -- (35.06,189.84) -- cycle ;
%Shape: Triangle [id:dp5596698606875515]
\draw  [fill={rgb, 255:red, 255; green, 255; blue, 255 }  ,fill opacity=1 ] (54.68,186.17) -- (47.03,184.9) -- (51.71,179.36) -- cycle ;
%Straight Lines [id:da21759608019557453]
\draw    (-15.32,144.3) -- (81,47.4) ;
%Straight Lines [id:da9243839063735605]
\draw    (1.73,148.7) -- (80.5,69.9) ;
%Straight Lines [id:da8279438759272657]
\draw    (28.93,148.5) -- (80.5,98.4) ;
%Shape: Triangle [id:dp21480081241052795]
\draw  [fill={rgb, 255:red, 255; green, 255; blue, 255 }  ,fill opacity=1 ] (40.61,109.9) -- (38.24,117.29) -- (33.44,111.85) -- cycle ;
%Shape: Triangle [id:dp9360116149487699]
\draw  [fill={rgb, 255:red, 255; green, 255; blue, 255 }  ,fill opacity=1 ] (53.68,124.38) -- (51.9,131.94) -- (46.68,126.9) -- cycle ;
%Shape: Arc [id:dp218294799674829]
\draw  [draw opacity=0] (-15.79,173.7) .. controls (-19.78,170.14) and (-22.29,164.97) .. (-22.29,159.2) .. controls (-22.29,153.22) and (-19.58,147.86) .. (-15.32,144.3) -- (-2.87,159.2) -- cycle ; \draw   (-15.79,173.7) .. controls (-19.78,170.14) and (-22.29,164.97) .. (-22.29,159.2) .. controls (-22.29,153.22) and (-19.58,147.86) .. (-15.32,144.3) ;
%Straight Lines [id:da5382153714416286]
\draw    (80.5,270) -- (110,280) ;
%Straight Lines [id:da34349920898617714]
\draw    (360.5,170.7) -- (359.18,247.61) -- (358.5,285.9) ;
%Straight Lines [id:da7605181656941171]
\draw    (400.9,171.1) -- (359.66,210.11) ;
%Straight Lines [id:da07883075925239769]
\draw    (430.1,171.1) -- (359.95,231.9) ;
%Straight Lines [id:da5516525263656364]
\draw    (446.15,180.6) -- (359.18,254.61) ;
%Shape: Triangle [id:dp6001345975534604]
\draw  [fill={rgb, 255:red, 255; green, 255; blue, 255 }  ,fill opacity=1 ] (394.41,201.85) -- (402.06,200.58) -- (397.38,195.04) -- cycle ;
%Shape: Triangle [id:dp08727978631654354]
\draw  [fill={rgb, 255:red, 255; green, 255; blue, 255 }  ,fill opacity=1 ] (378.75,192.37) -- (386.41,191.1) -- (381.72,185.56) -- cycle ;
%Straight Lines [id:da060221009529701]
\draw    (444.75,149.5) -- (359.3,64.3) ;
%Straight Lines [id:da10392247211669692]
\draw    (426.7,153.9) -- (359.7,89) ;
%Straight Lines [id:da9085567162198843]
\draw    (394.7,154.5) -- (360.2,115) ;
%Shape: Triangle [id:dp04098949046020661]
\draw  [fill={rgb, 255:red, 255; green, 255; blue, 255 }  ,fill opacity=1 ] (387.82,116.1) -- (390.19,123.49) -- (395,118.05) -- cycle ;
%Shape: Triangle [id:dp8262251259792709]
\draw  [fill={rgb, 255:red, 255; green, 255; blue, 255 }  ,fill opacity=1 ] (375.76,132.58) -- (377.53,140.14) -- (382.75,135.1) -- cycle ;
%Shape: Rectangle [id:dp976393081741011]
\draw   (277.03,154.3) -- (437.3,154.3) -- (437.3,170.7) -- (277.03,170.7) -- cycle ;
%Straight Lines [id:da8715615530467257]
\draw    (359.3,39.5) -- (359.75,132.3) -- (359.9,154.1) ;
%Curve Lines [id:da2671149216033609]
\draw    (444.75,149.5) .. controls (444.37,149.03) and (450.1,154.3) .. (453.7,160.3) .. controls (457.3,166.3) and (452.9,174.7) .. (446.15,180.6) ;
%Shape: Triangle [id:dp49565039813688716]
\draw  [fill={rgb, 255:red, 255; green, 255; blue, 255 }  ,fill opacity=1 ] (363.65,138.53) -- (359.75,132.3) -- (355.79,138.39) -- cycle ;
%Shape: Triangle [id:dp8553079811607096]
\draw  [fill={rgb, 255:red, 255; green, 255; blue, 255 }  ,fill opacity=1 ] (364.29,182.42) -- (356.53,182.56) -- (360.36,188.72) -- cycle ;
%Straight Lines [id:da23378394559278493]
\draw    (312.23,171.5) -- (359.7,217.5) ;
%Straight Lines [id:da1357942512032928]
\draw    (283.99,171.83) -- (359.2,241.5) ;
%Straight Lines [id:da9376723993939036]
\draw    (262.91,179.3) -- (358.7,265) ;
%Shape: Triangle [id:dp5158679825983292]
\draw  [fill={rgb, 255:red, 255; green, 255; blue, 255 }  ,fill opacity=1 ] (316.73,202.25) -- (309.07,200.98) -- (313.76,195.44) -- cycle ;
%Shape: Triangle [id:dp23619558263205098]
\draw  [fill={rgb, 255:red, 255; green, 255; blue, 255 }  ,fill opacity=1 ] (333.38,191.77) -- (325.73,190.5) -- (330.41,184.96) -- cycle ;
%Straight Lines [id:da3909749739412438]
\draw    (263.38,149.9) -- (359.7,53) ;
%Straight Lines [id:da30888637914335604]
\draw    (280.43,154.3) -- (359.2,75.5) ;
%Straight Lines [id:da861444935343907]
\draw    (307.63,154.1) -- (359.2,104) ;
%Shape: Triangle [id:dp39148818013894093]
\draw  [fill={rgb, 255:red, 255; green, 255; blue, 255 }  ,fill opacity=1 ] (319.31,115.5) -- (316.94,122.89) -- (312.14,117.45) -- cycle ;
%Shape: Triangle [id:dp2759335307289804]
\draw  [fill={rgb, 255:red, 255; green, 255; blue, 255 }  ,fill opacity=1 ] (332.38,129.98) -- (330.6,137.54) -- (325.38,132.5) -- cycle ;
%Shape: Arc [id:dp9017769114876162]
\draw  [draw opacity=0] (262.91,179.3) .. controls (258.92,175.74) and (256.41,170.57) .. (256.41,164.8) .. controls (256.41,158.82) and (259.12,153.46) .. (263.38,149.9) -- (275.83,164.8) -- cycle ; \draw   (262.91,179.3) .. controls (258.92,175.74) and (256.41,170.57) .. (256.41,164.8) .. controls (256.41,158.82) and (259.12,153.46) .. (263.38,149.9) ;
%Straight Lines [id:da8786509683406637]
\draw    (643.5,173.2) -- (642.18,250.11) -- (641.5,288.4) ;
%Straight Lines [id:da5100946896828316]
\draw    (683.9,173.6) -- (642.66,212.61) ;
%Straight Lines [id:da4826650777421515]
\draw    (713.1,173.6) -- (642.95,234.4) ;
%Straight Lines [id:da7274108239239004]
\draw    (729.15,183.1) -- (642.18,257.11) ;
%Shape: Triangle [id:dp6229775700396352]
\draw  [fill={rgb, 255:red, 255; green, 255; blue, 255 }  ,fill opacity=1 ] (677.41,204.35) -- (685.06,203.08) -- (680.38,197.54) -- cycle ;
%Shape: Triangle [id:dp814624914769608]
\draw  [fill={rgb, 255:red, 255; green, 255; blue, 255 }  ,fill opacity=1 ] (661.75,194.87) -- (669.41,193.6) -- (664.72,188.06) -- cycle ;
%Straight Lines [id:da0992068016894927]
\draw    (727.75,152) -- (642.3,66.8) ;
%Straight Lines [id:da6602379529202727]
\draw    (709.7,156.4) -- (642.7,91.5) ;
%Straight Lines [id:da5699610716129291]
\draw    (677.7,157) -- (643.2,117.5) ;
%Shape: Triangle [id:dp8940782565676124]
\draw  [fill={rgb, 255:red, 255; green, 255; blue, 255 }  ,fill opacity=1 ] (670.82,118.6) -- (673.19,125.99) -- (678,120.55) -- cycle ;
%Shape: Triangle [id:dp10539850230625059]
\draw  [fill={rgb, 255:red, 255; green, 255; blue, 255 }  ,fill opacity=1 ] (658.76,135.08) -- (660.53,142.64) -- (665.75,137.6) -- cycle ;
%Shape: Rectangle [id:dp3469027445265428]
\draw   (560.03,156.8) -- (720.3,156.8) -- (720.3,173.2) -- (560.03,173.2) -- cycle ;
%Straight Lines [id:da4979139857977024]
\draw    (641.5,25.5) -- (642.75,134.8) -- (642.9,156.6) ;
%Curve Lines [id:da5783421487605731]
\draw    (727.75,152) .. controls (727.37,151.53) and (733.1,156.8) .. (736.7,162.8) .. controls (740.3,168.8) and (735.9,177.2) .. (729.15,183.1) ;
%Shape: Triangle [id:dp6702037308778169]
\draw  [fill={rgb, 255:red, 255; green, 255; blue, 255 }  ,fill opacity=1 ] (646.65,141.03) -- (642.75,134.8) -- (638.79,140.89) -- cycle ;
%Shape: Triangle [id:dp8501196178217094]
\draw  [fill={rgb, 255:red, 255; green, 255; blue, 255 }  ,fill opacity=1 ] (647.29,184.92) -- (639.53,185.06) -- (643.36,191.22) -- cycle ;
%Straight Lines [id:da8964340144563354]
\draw    (595.23,174) -- (642.7,220) ;
%Straight Lines [id:da40538482697650025]
\draw    (566.99,174.33) -- (642.2,244) ;
%Straight Lines [id:da6174595559268967]
\draw    (545.91,181.8) -- (641.7,267.5) ;
%Shape: Triangle [id:dp4718293852644968]
\draw  [fill={rgb, 255:red, 255; green, 255; blue, 255 }  ,fill opacity=1 ] (599.73,204.75) -- (592.07,203.48) -- (596.76,197.94) -- cycle ;
%Shape: Triangle [id:dp5105336153524782]
\draw  [fill={rgb, 255:red, 255; green, 255; blue, 255 }  ,fill opacity=1 ] (616.38,194.27) -- (608.73,193) -- (613.41,187.46) -- cycle ;
%Straight Lines [id:da11328529364205786]
\draw    (546.38,152.4) -- (642,56) ;
%Straight Lines [id:da21916524795865133]
\draw    (563.43,156.8) -- (642.2,78) ;
%Straight Lines [id:da2775169796376221]
\draw    (590.63,156.6) -- (642.2,106.5) ;
%Shape: Triangle [id:dp2819585030705265]
\draw  [fill={rgb, 255:red, 255; green, 255; blue, 255 }  ,fill opacity=1 ] (602.31,118) -- (599.94,125.39) -- (595.14,119.95) -- cycle ;
%Shape: Triangle [id:dp5285455629704218]
\draw  [fill={rgb, 255:red, 255; green, 255; blue, 255 }  ,fill opacity=1 ] (615.38,132.48) -- (613.6,140.04) -- (608.38,135) -- cycle ;
%Shape: Arc [id:dp6162416498856178]
\draw  [draw opacity=0] (545.91,181.8) .. controls (541.92,178.24) and (539.41,173.07) .. (539.41,167.3) .. controls (539.41,161.32) and (542.12,155.96) .. (546.38,152.4) -- (558.83,167.3) -- cycle ; \draw   (545.91,181.8) .. controls (541.92,178.24) and (539.41,173.07) .. (539.41,167.3) .. controls (539.41,161.32) and (542.12,155.96) .. (546.38,152.4) ;
%Curve Lines [id:da24977817976627759]
\draw    (642.5,49) .. controls (696,41) and (820,133) .. (721,276) ;
%Curve Lines [id:da9380704847375879]
\draw    (455,163) .. controls (454,205) and (452,208.5) .. (444,230.5) .. controls (436,252.5) and (432.5,263) .. (423,280.5) ;

% Text Node
\draw (66.2,269.9) node [anchor=north west][inner sep=0.75pt]  [font=\footnotesize]  {$Q$};
% Text Node
\draw (69.7,185.1) node [anchor=north west][inner sep=0.75pt]  [font=\scriptsize]  {$m ^{\dagger }$};
% Text Node
\draw (68.9,233.8) node [anchor=north west][inner sep=0.75pt]  [font=\scriptsize]  {$m ^{\dagger }$};
% Text Node
\draw (68.48,209.41) node [anchor=north west][inner sep=0.75pt]  [font=\scriptsize]  {$m ^{\dagger }$};
% Text Node
\draw (81.6,43.4) node [anchor=north west][inner sep=0.75pt]  [font=\scriptsize]  {$m $};
% Text Node
\draw (82.3,70.6) node [anchor=north west][inner sep=0.75pt]  [font=\scriptsize]  {$m $};
% Text Node
\draw (82.2,95.2) node [anchor=north west][inner sep=0.75pt]  [font=\scriptsize]  {$m $};
% Text Node
\draw (91.2,196.3) node [anchor=north west][inner sep=0.75pt]    {$...$};
% Text Node
\draw (96.4,117.3) node [anchor=north west][inner sep=0.75pt]    {$...$};
% Text Node
\draw (178.5,147.2) node [anchor=north west][inner sep=0.75pt]  [font=\scriptsize]  {$Q$};
% Text Node
\draw (-34.7,146.4) node [anchor=north west][inner sep=0.75pt]  [font=\scriptsize]  {$Q$};
% Text Node
\draw (74.93,154.3) node [anchor=north west][inner sep=0.75pt]    {$a$};
% Text Node
\draw (46.53,195.7) node [anchor=north east][inner sep=0.75pt]  [xscale=-1]  {$...$};
% Text Node
\draw (36.33,118.7) node [anchor=north east][inner sep=0.75pt]  [xscale=-1]  {$...$};
% Text Node
\draw (67.1,24.8) node [anchor=north west][inner sep=0.75pt]  [font=\footnotesize]  {$Q$};
% Text Node
\draw (54.3,92.9) node [anchor=north west][inner sep=0.75pt]  [font=\scriptsize]  {$Q$};
% Text Node
\draw (97.3,90.9) node [anchor=north west][inner sep=0.75pt]  [font=\scriptsize]  {$Q$};
% Text Node
\draw (14.8,134.4) node [anchor=north west][inner sep=0.75pt]  [font=\footnotesize]  {$x$};
% Text Node
\draw (126.8,136.4) node [anchor=north west][inner sep=0.75pt]  [font=\footnotesize]  {$x$};
% Text Node
\draw (41.3,135.9) node [anchor=north west][inner sep=0.75pt]  [font=\footnotesize]  {$x$};
% Text Node
\draw (97.3,137.4) node [anchor=north west][inner sep=0.75pt]  [font=\footnotesize]  {$x$};
% Text Node
\draw (84.3,107.4) node [anchor=north west][inner sep=0.75pt]  [font=\scriptsize]  {$Q$};
% Text Node
\draw (60.8,114.4) node [anchor=north west][inner sep=0.75pt]  [font=\scriptsize]  {$Q$};
% Text Node
\draw (13.8,168.4) node [anchor=north west][inner sep=0.75pt]  [font=\footnotesize]  {$x$};
% Text Node
\draw (132.8,168.4) node [anchor=north west][inner sep=0.75pt]  [font=\footnotesize]  {$x$};
% Text Node
\draw (42.3,167.9) node [anchor=north west][inner sep=0.75pt]  [font=\footnotesize]  {$x$};
% Text Node
\draw (104.3,168.4) node [anchor=north west][inner sep=0.75pt]  [font=\footnotesize]  {$x$};
% Text Node
\draw (86.3,182.4) node [anchor=north west][inner sep=0.75pt]  [font=\scriptsize]  {$Q$};
% Text Node
\draw (94.8,210.9) node [anchor=north west][inner sep=0.75pt]  [font=\scriptsize]  {$Q$};
% Text Node
\draw (57.3,180.9) node [anchor=north west][inner sep=0.75pt]  [font=\scriptsize]  {$Q$};
% Text Node
\draw (46.3,212.4) node [anchor=north west][inner sep=0.75pt]  [font=\scriptsize]  {$Q$};
% Text Node
\draw (113.2,271.4) node [anchor=north west][inner sep=0.75pt]  [font=\footnotesize]  {$Q$};
% Text Node
\draw (205.5,149.4) node [anchor=north west][inner sep=0.75pt]  [font=\large]  {$=$};
% Text Node
\draw (486.5,153.4) node [anchor=north west][inner sep=0.75pt]  [font=\large]  {$=$};
% Text Node
\draw (344.9,275.5) node [anchor=north west][inner sep=0.75pt]  [font=\footnotesize]  {$Q$};
% Text Node
\draw (348.4,190.7) node [anchor=north west][inner sep=0.75pt]  [font=\scriptsize]  {$m ^{\dagger }$};
% Text Node
\draw (347.6,239.4) node [anchor=north west][inner sep=0.75pt]  [font=\scriptsize]  {$m ^{\dagger }$};
% Text Node
\draw (347.18,215.01) node [anchor=north west][inner sep=0.75pt]  [font=\scriptsize]  {$m ^{\dagger }$};
% Text Node
\draw (360.3,49) node [anchor=north west][inner sep=0.75pt]  [font=\scriptsize]  {$m $};
% Text Node
\draw (361,76.2) node [anchor=north west][inner sep=0.75pt]  [font=\scriptsize]  {$m $};
% Text Node
\draw (360.9,100.8) node [anchor=north west][inner sep=0.75pt]  [font=\scriptsize]  {$m $};
% Text Node
\draw (369.9,201.9) node [anchor=north west][inner sep=0.75pt]    {$...$};
% Text Node
\draw (375.1,122.9) node [anchor=north west][inner sep=0.75pt]    {$...$};
% Text Node
\draw (457.2,152.8) node [anchor=north west][inner sep=0.75pt]  [font=\scriptsize]  {$Q$};
% Text Node
\draw (244,152) node [anchor=north west][inner sep=0.75pt]  [font=\scriptsize]  {$Q$};
% Text Node
\draw (353.63,159.9) node [anchor=north west][inner sep=0.75pt]    {$a$};
% Text Node
\draw (325.23,201.3) node [anchor=north east][inner sep=0.75pt]  [xscale=-1]  {$...$};
% Text Node
\draw (315.03,124.3) node [anchor=north east][inner sep=0.75pt]  [xscale=-1]  {$...$};
% Text Node
\draw (345.8,30.4) node [anchor=north west][inner sep=0.75pt]  [font=\footnotesize]  {$Q$};
% Text Node
\draw (333,98.5) node [anchor=north west][inner sep=0.75pt]  [font=\scriptsize]  {$Q$};
% Text Node
\draw (376,96.5) node [anchor=north west][inner sep=0.75pt]  [font=\scriptsize]  {$Q$};
% Text Node
\draw (293.5,140) node [anchor=north west][inner sep=0.75pt]  [font=\footnotesize]  {$x$};
% Text Node
\draw (405.5,142) node [anchor=north west][inner sep=0.75pt]  [font=\footnotesize]  {$x$};
% Text Node
\draw (320,141.5) node [anchor=north west][inner sep=0.75pt]  [font=\footnotesize]  {$x$};
% Text Node
\draw (376,143) node [anchor=north west][inner sep=0.75pt]  [font=\footnotesize]  {$x$};
% Text Node
\draw (363,113) node [anchor=north west][inner sep=0.75pt]  [font=\scriptsize]  {$Q$};
% Text Node
\draw (339.5,120) node [anchor=north west][inner sep=0.75pt]  [font=\scriptsize]  {$Q$};
% Text Node
\draw (292.5,174) node [anchor=north west][inner sep=0.75pt]  [font=\footnotesize]  {$x$};
% Text Node
\draw (411.5,174) node [anchor=north west][inner sep=0.75pt]  [font=\footnotesize]  {$x$};
% Text Node
\draw (321,173.5) node [anchor=north west][inner sep=0.75pt]  [font=\footnotesize]  {$x$};
% Text Node
\draw (383,174) node [anchor=north west][inner sep=0.75pt]  [font=\footnotesize]  {$x$};
% Text Node
\draw (365,188) node [anchor=north west][inner sep=0.75pt]  [font=\scriptsize]  {$Q$};
% Text Node
\draw (373.5,216.5) node [anchor=north west][inner sep=0.75pt]  [font=\scriptsize]  {$Q$};
% Text Node
\draw (336,186.5) node [anchor=north west][inner sep=0.75pt]  [font=\scriptsize]  {$Q$};
% Text Node
\draw (325,218) node [anchor=north west][inner sep=0.75pt]  [font=\scriptsize]  {$Q$};
% Text Node
\draw (405.4,274.5) node [anchor=north west][inner sep=0.75pt]  [font=\footnotesize]  {$Q$};
% Text Node
\draw (627.9,278) node [anchor=north west][inner sep=0.75pt]  [font=\footnotesize]  {$Q$};
% Text Node
\draw (631.4,193.2) node [anchor=north west][inner sep=0.75pt]  [font=\scriptsize]  {$m ^{\dagger }$};
% Text Node
\draw (630.6,241.9) node [anchor=north west][inner sep=0.75pt]  [font=\scriptsize]  {$m ^{\dagger }$};
% Text Node
\draw (630.18,217.51) node [anchor=north west][inner sep=0.75pt]  [font=\scriptsize]  {$m ^{\dagger }$};
% Text Node
\draw (643.3,51.5) node [anchor=north west][inner sep=0.75pt]  [font=\scriptsize]  {$m $};
% Text Node
\draw (644,78.7) node [anchor=north west][inner sep=0.75pt]  [font=\scriptsize]  {$m $};
% Text Node
\draw (643.9,103.3) node [anchor=north west][inner sep=0.75pt]  [font=\scriptsize]  {$m $};
% Text Node
\draw (652.9,204.4) node [anchor=north west][inner sep=0.75pt]    {$...$};
% Text Node
\draw (658.1,125.4) node [anchor=north west][inner sep=0.75pt]    {$...$};
% Text Node
\draw (740.2,155.3) node [anchor=north west][inner sep=0.75pt]  [font=\scriptsize]  {$Q$};
% Text Node
\draw (527,154.5) node [anchor=north west][inner sep=0.75pt]  [font=\scriptsize]  {$Q$};
% Text Node
\draw (636.63,162.4) node [anchor=north west][inner sep=0.75pt]    {$a$};
% Text Node
\draw (608.23,203.8) node [anchor=north east][inner sep=0.75pt]  [xscale=-1]  {$...$};
% Text Node
\draw (598.03,126.8) node [anchor=north east][inner sep=0.75pt]  [xscale=-1]  {$...$};
% Text Node
\draw (626.3,27.4) node [anchor=north west][inner sep=0.75pt]  [font=\footnotesize]  {$Q$};
% Text Node
\draw (616,101) node [anchor=north west][inner sep=0.75pt]  [font=\scriptsize]  {$Q$};
% Text Node
\draw (659,99) node [anchor=north west][inner sep=0.75pt]  [font=\scriptsize]  {$Q$};
% Text Node
\draw (576.5,142.5) node [anchor=north west][inner sep=0.75pt]  [font=\footnotesize]  {$x$};
% Text Node
\draw (688.5,144.5) node [anchor=north west][inner sep=0.75pt]  [font=\footnotesize]  {$x$};
% Text Node
\draw (603,144) node [anchor=north west][inner sep=0.75pt]  [font=\footnotesize]  {$x$};
% Text Node
\draw (659,145.5) node [anchor=north west][inner sep=0.75pt]  [font=\footnotesize]  {$x$};
% Text Node
\draw (646,115.5) node [anchor=north west][inner sep=0.75pt]  [font=\scriptsize]  {$Q$};
% Text Node
\draw (622.5,122.5) node [anchor=north west][inner sep=0.75pt]  [font=\scriptsize]  {$Q$};
% Text Node
\draw (575.5,176.5) node [anchor=north west][inner sep=0.75pt]  [font=\footnotesize]  {$x$};
% Text Node
\draw (694.5,176.5) node [anchor=north west][inner sep=0.75pt]  [font=\footnotesize]  {$x$};
% Text Node
\draw (604,176) node [anchor=north west][inner sep=0.75pt]  [font=\footnotesize]  {$x$};
% Text Node
\draw (666,176.5) node [anchor=north west][inner sep=0.75pt]  [font=\footnotesize]  {$x$};
% Text Node
\draw (648,190.5) node [anchor=north west][inner sep=0.75pt]  [font=\scriptsize]  {$Q$};
% Text Node
\draw (656.5,219) node [anchor=north west][inner sep=0.75pt]  [font=\scriptsize]  {$Q$};
% Text Node
\draw (619,189) node [anchor=north west][inner sep=0.75pt]  [font=\scriptsize]  {$Q$};
% Text Node
\draw (608,220.5) node [anchor=north west][inner sep=0.75pt]  [font=\scriptsize]  {$Q$};
% Text Node
\draw (700.23,276.95) node [anchor=north west][inner sep=0.75pt]  [font=\footnotesize]  {$Q$};

\end{tikzpicture}

%% file: diagrams/boundary-screening-1906.tex
\begin{tikzpicture}[x=0.75pt,y=0.75pt,yscale=-1,xscale=1,baseline = (current bounding box.center)]
%Straight Lines [id:da9319577591460301]
\draw    (391.8,146.2) -- (390.48,223.11) -- (389.8,261.4) ;
%Straight Lines [id:da10578108900526628]
\draw    (432.2,146.6) -- (390.96,185.61) ;
%Straight Lines [id:da9996663344556211]
\draw    (461.4,146.6) -- (391.25,207.4) ;
%Straight Lines [id:da8129698097338967]
\draw    (477.45,156.1) -- (390.48,230.11) ;
%Shape: Triangle [id:dp8096233458101415]
\draw  [fill={rgb, 255:red, 255; green, 255; blue, 255 }  ,fill opacity=1 ] (422.71,179.35) -- (430.36,178.08) -- (425.68,172.54) -- cycle ;
%Shape: Triangle [id:dp9568537543027305]
\draw  [fill={rgb, 255:red, 255; green, 255; blue, 255 }  ,fill opacity=1 ] (408.05,169.87) -- (415.71,168.6) -- (411.02,163.06) -- cycle ;
%Shape: Triangle [id:dp265802028762005]
\draw  [fill={rgb, 255:red, 255; green, 255; blue, 255 }  ,fill opacity=1 ] (439.07,188.64) -- (446.69,187.12) -- (441.82,181.73) -- cycle ;
%Straight Lines [id:da3702088750347281]
\draw    (476.05,125) -- (390.6,39.8) ;
%Straight Lines [id:da07656625481273005]
\draw    (459,129.4) -- (390.85,61.8) ;
%Straight Lines [id:da5758233210582552]
\draw    (431.8,130.2) -- (390.45,88.6) ;
%Shape: Triangle [id:dp3514734228093983]
\draw  [fill={rgb, 255:red, 255; green, 255; blue, 255 }  ,fill opacity=1 ] (420.12,91.6) -- (422.49,98.99) -- (427.3,93.55) -- cycle ;
%Shape: Triangle [id:dp6581571945513973]
\draw  [fill={rgb, 255:red, 255; green, 255; blue, 255 }  ,fill opacity=1 ] (407.06,105.08) -- (408.83,112.64) -- (414.05,107.6) -- cycle ;
%Shape: Triangle [id:dp09412773996299817]
\draw  [fill={rgb, 255:red, 255; green, 255; blue, 255 }  ,fill opacity=1 ] (431.72,79.8) -- (433.35,87.39) -- (438.66,82.44) -- cycle ;
%Shape: Rectangle [id:dp16996778191236817]
\draw   (381.2,129.8) -- (468.6,129.8) -- (468.6,146.2) -- (381.2,146.2) -- cycle ;
%Straight Lines [id:da6699941119725045]
\draw    (390.6,15) -- (391.05,107.8) -- (391.2,129.6) ;
%Curve Lines [id:da4861166342443547]
\draw    (476.05,125) .. controls (475.67,124.53) and (481.4,129.8) .. (485,135.8) .. controls (488.6,141.8) and (484.2,150.2) .. (477.45,156.1) ;
%Shape: Triangle [id:dp5066005345853508]
\draw  [fill={rgb, 255:red, 255; green, 255; blue, 255 }  ,fill opacity=1 ] (394.95,114.03) -- (391.05,107.8) -- (387.09,113.89) -- cycle ;
%Shape: Triangle [id:dp3642657227691165]
\draw  [fill={rgb, 255:red, 255; green, 255; blue, 255 }  ,fill opacity=1 ] (395.59,157.92) -- (387.83,158.06) -- (391.66,164.22) -- cycle ;
% Text Node
\draw (476.4,111) node [anchor=north west][inner sep=0.75pt]  [font=\footnotesize]  {$x$};
% Text Node
\draw (374.2,244) node [anchor=north west][inner sep=0.75pt]  [font=\footnotesize]  {$K$};
% Text Node
\draw (375.2,175.2) node [anchor=north west][inner sep=0.75pt]  [font=\scriptsize]  {$\mu ^{\dagger }$};
% Text Node
\draw (376.4,202.4) node [anchor=north west][inner sep=0.75pt]  [font=\scriptsize]  {$\mu ^{\dagger }$};
% Text Node
\draw (376,224) node [anchor=north west][inner sep=0.75pt]  [font=\scriptsize]  {$\mu ^{\dagger }$};
% Text Node
\draw (375.6,38) node [anchor=north west][inner sep=0.75pt]  [font=\scriptsize]  {$\mu $};
% Text Node
\draw (376.8,59.2) node [anchor=north west][inner sep=0.75pt]  [font=\scriptsize]  {$\mu $};
% Text Node
\draw (375.2,84.8) node [anchor=north west][inner sep=0.75pt]  [font=\scriptsize]  {$\mu $};
% Text Node
\draw (401.2,177.4) node [anchor=north west][inner sep=0.75pt]    {$...$};
% Text Node
\draw (399.4,61.4) node [anchor=north west][inner sep=0.75pt]    {$...$};
% Text Node
\draw (416,51.8) node [anchor=north west][inner sep=0.75pt]  [font=\scriptsize]  {$Q$};
% Text Node
\draw (414.8,209) node [anchor=north west][inner sep=0.75pt]  [font=\scriptsize]  {$Q$};
% Text Node
\draw (373.8,15.2) node [anchor=north west][inner sep=0.75pt]  [font=\footnotesize]  {$K$};
% Text Node
\draw (421.6,134.4) node [anchor=north west][inner sep=0.75pt]    {$a$};
\end{tikzpicture}

%% file: diagrams/gns-kernel-2056.tex
\begin{tikzpicture}[x=0.75pt,y=0.75pt,yscale=-1,xscale=1]
%uncomment if require: \path (0,291); %set diagram left start at 0, and has height of 291

%Straight Lines [id:da9817425026044349]
\draw    (413.47,124.27) -- (412.15,201.18) -- (411.47,239.47) ;
%Straight Lines [id:da45473579280331167]
\draw    (453.87,124.67) -- (412.63,163.67) ;
%Straight Lines [id:da11698480364579966]
\draw    (483.07,124.67) -- (412.92,185.47) ;
%Straight Lines [id:da9974172924840617]
\draw    (495.82,135.9) -- (412.15,208.18) ;
%Shape: Triangle [id:dp18236216241499426]
\draw  [fill={rgb, 255:red, 255; green, 255; blue, 255 }  ,fill opacity=1 ] (444.37,157.42) -- (452.03,156.15) -- (447.34,150.6) -- cycle ;
%Shape: Triangle [id:dp10295438160616444]
\draw  [fill={rgb, 255:red, 255; green, 255; blue, 255 }  ,fill opacity=1 ] (429.72,147.94) -- (437.37,146.67) -- (432.69,141.13) -- cycle ;
%Shape: Triangle [id:dp6509166634377024]
\draw  [fill={rgb, 255:red, 255; green, 255; blue, 255 }  ,fill opacity=1 ] (460.74,166.7) -- (468.35,165.18) -- (463.49,159.8) -- cycle ;
%Shape: Rectangle [id:dp6059827937972057]
\draw   (402.87,107.87) -- (490.27,107.87) -- (490.27,124.27) -- (402.87,124.27) -- cycle ;
%Shape: Triangle [id:dp949059298887954]
\draw  [fill={rgb, 255:red, 255; green, 255; blue, 255 }  ,fill opacity=1 ] (417.25,135.98) -- (409.49,136.13) -- (413.32,142.29) -- cycle ;
%Shape: Rectangle [id:dp5652161581973146]
\draw   (175.87,121.73) -- (263.27,121.73) -- (263.27,138.13) -- (175.87,138.13) -- cycle ;
%Straight Lines [id:da304701764466891]
\draw    (212.92,122.34) -- (212.27,74.5) ;
%Straight Lines [id:da5054818770230927]
\draw    (253.67,121.7) -- (253.02,73.86) ;
%Straight Lines [id:da18023756328744478]
\draw    (185.87,121.82) -- (185.22,73.99) ;
%Straight Lines [id:da8600612937357188]
\draw    (214.58,185.74) -- (213.93,139.06) ;
%Straight Lines [id:da8913996977558774]
\draw    (254.33,184.74) -- (253.68,138.06) ;
%Straight Lines [id:da8776962575158472]
\draw    (186.53,184.86) -- (185.88,138.18) ;
%Straight Lines [id:da43106585207052284]
\draw    (440.25,108.34) -- (439.6,60.5) ;
%Straight Lines [id:da3483200837023911]
\draw    (481,107.7) -- (480.35,59.86) ;
%Straight Lines [id:da8724323257346835]
\draw    (413.2,107.82) -- (412.55,59.99) ;
%Straight Lines [id:da002476885307968968]
\draw    (503.23,112.51) -- (503.15,59.9) ;
%Shape: Arc [id:dp13386445932995616]
\draw  [draw opacity=0] (503.23,112.51) .. controls (503.36,121.24) and (500.74,129.52) .. (495.82,135.9) -- (467.81,109.41) -- cycle ; \draw   (503.23,112.51) .. controls (503.36,121.24) and (500.74,129.52) .. (495.82,135.9) ;

% Text Node
\draw (396.87,228.07) node [anchor=north west][inner sep=0.75pt]  [font=\footnotesize]  {$K$};
% Text Node
\draw (396.87,153.27) node [anchor=north west][inner sep=0.75pt]  [font=\scriptsize]  {$\mu ^{\dagger }$};
% Text Node
\draw (398.07,180.47) node [anchor=north west][inner sep=0.75pt]  [font=\scriptsize]  {$\mu ^{\dagger }$};
% Text Node
\draw (397.67,202.07) node [anchor=north west][inner sep=0.75pt]  [font=\scriptsize]  {$\mu ^{\dagger }$};
% Text Node
\draw (422.87,154.47) node [anchor=north west][inner sep=0.75pt]    {$...$};
% Text Node
\draw (436.47,187.07) node [anchor=north west][inner sep=0.75pt]  [font=\scriptsize]  {$Q$};
% Text Node
\draw (442.27,109.47) node [anchor=north west][inner sep=0.75pt]    {$a$};
% Text Node
\draw (215.27,123.33) node [anchor=north west][inner sep=0.75pt]    {$a$};
% Text Node
\draw (108,117.4) node [anchor=north west][inner sep=0.75pt]    {$\mu^* :$};
% Text Node
\draw (309.33,118.07) node [anchor=north west][inner sep=0.75pt]  [font=\Large]  {$\longmapsto $};
% Text Node
\draw (182.33,189.07) node [anchor=north west][inner sep=0.75pt]    {$l$};
% Text Node
\draw (209.47,187.6) node [anchor=north west][inner sep=0.75pt]  [font=\normalsize]  {$x$};
% Text Node
\draw (248.53,187.4) node [anchor=north west][inner sep=0.75pt]  [font=\normalsize]  {$x$};
% Text Node
\draw (181.67,55.4) node [anchor=north west][inner sep=0.75pt]    {$l$};
% Text Node
\draw (207.8,54.93) node [anchor=north west][inner sep=0.75pt]  [font=\normalsize]  {$x$};
% Text Node
\draw (247.87,55.4) node [anchor=north west][inner sep=0.75pt]  [font=\normalsize]  {$x$};
% Text Node
\draw (224.4,85.47) node [anchor=north west][inner sep=0.75pt]    {$...$};
% Text Node
\draw (225.07,154.13) node [anchor=north west][inner sep=0.75pt]    {$...$};
% Text Node
\draw (409,41.4) node [anchor=north west][inner sep=0.75pt]    {$l$};
% Text Node
\draw (435.13,40.93) node [anchor=north west][inner sep=0.75pt]  [font=\normalsize]  {$x$};
% Text Node
\draw (475.2,40.73) node [anchor=north west][inner sep=0.75pt]  [font=\normalsize]  {$x$};
% Text Node
\draw (451.73,71.47) node [anchor=north west][inner sep=0.75pt]    {$...$};
% Text Node
\draw (498.2,40.73) node [anchor=north west][inner sep=0.75pt]  [font=\normalsize]  {$x$};

\end{tikzpicture}

%% file: diagrams/reabcond-functor-1362.tex
\begin{tikzpicture}[
    x=1cm, % was 2cm; this was making the picture too wide
    y=1cm,
    strand/.style={line width=0.75pt},
    box/.style={draw, inner sep=0pt},
    morphism/.style={draw, minimum size=0.48cm, inner sep=0pt},
    label/.style={font=\footnotesize, inner sep=1pt},
    title/.style={font=\normalsize, inner sep=1pt}
]

% height of the rectangular box
\pgfmathsetmacro{\boxTop}{1+9/20}

\node[title, anchor=east] at (-1.35,1.15)
    {$\mathrm{ReaBCond}_{n}\left(K_{1}\xrightarrow{f}K_{2}\right):$};

\begin{scope}[shift={(0,0)}]
    \draw[box] (-1.1,1.0) rectangle (1.1,\boxTop);

    \draw[strand] (-0.8,\boxTop) -- (-0.8,2.55);
    \draw[strand] (0,\boxTop) -- (0,2.55);
    \draw[strand] (0.8,\boxTop) -- (0.8,2.55);

    \draw[strand] (0,1.0) -- (0,0.1);

    \node[label] at (0,1.23) {$a$};
    \node[label] at (-0.8,2.72) {$l$};
    \node[label] at (0,2.72) {$x$};
    \node[label] at (0.8,2.72) {$x$};
    \node[label] at (0.4,2.0) {$\cdots$};
    \node[label] at (0,-0.1) {$K_{2}$};
\end{scope}

\node[font=\Large] at (2.2,1.25) {$\longmapsto$};

\begin{scope}[shift={(4.4,0)}]
    \draw[box] (-1.1,1.0) rectangle (1.1,\boxTop);

    \draw[strand] (-0.8,\boxTop) -- (-0.8,2.55);
    \draw[strand] (0,\boxTop) -- (0,2.55);
    \draw[strand] (0.8,\boxTop) -- (0.8,2.55);

    \draw[strand] (0,1.0) -- (0,0.35);
    \node[morphism] at (0,0.1) {$f$};
    \draw[strand] (0,-0.14) -- (0,-0.75);

    \node[label] at (0,1.23) {$a$};
    \node[label] at (-0.8,2.72) {$l$};
    \node[label] at (0,2.72) {$x$};
    \node[label] at (0.8,2.72) {$x$};
    \node[label] at (0.4,2.0) {$\cdots$};
    \node[label] at (0,-0.95) {$K_{1}$};
\end{scope}

\end{tikzpicture}

%% file: diagrams/partial-screening-2281.tex
\begin{tikzpicture}[x=0.75pt,y=0.75pt,yscale=-1,xscale=1]
%uncomment if require: \path (0,324); %set diagram left start at 0, and has height of 324

%Shape: Rectangle [id:dp038696415000486484]
\draw   (136.87,158.23) -- (277.75,158.23) -- (277.75,174.38) -- (136.87,174.38) -- cycle ;
%Straight Lines [id:da547051538467689]
\draw    (165.92,158.84) -- (165.27,111) ;
%Straight Lines [id:da7834762081217445]
\draw    (205.67,158.2) -- (205.02,110.36) ;
%Straight Lines [id:da7193073763654981]
\draw    (146.87,158.32) -- (146.22,110.49) ;
%Straight Lines [id:da5492564873013105]
\draw    (167.58,222.24) -- (166.93,175.56) ;
%Straight Lines [id:da7684976944941739]
\draw    (206.33,221.24) -- (205.68,174.56) ;
%Straight Lines [id:da14988886483700414]
\draw    (147.53,221.36) -- (146.88,174.68) ;
%Straight Lines [id:da7033276838765374]
\draw    (224.92,158.84) -- (224.27,111) ;
%Straight Lines [id:da3358996823838052]
\draw    (264.67,158.2) -- (264.02,110.36) ;
%Straight Lines [id:da14446824833738237]
\draw    (226.58,222.24) -- (225.93,175.56) ;
%Straight Lines [id:da8781694380282821]
\draw    (265.33,221.24) -- (264.68,174.56) ;
%Straight Lines [id:da12343909770158645]
\draw    (500.8,176.77) -- (499.48,253.69) -- (498.8,291.97) ;
%Straight Lines [id:da08609577456842288]
\draw    (541.2,177.17) -- (499.96,216.18) ;
%Straight Lines [id:da4580194378212571]
\draw    (570.4,177.17) -- (500.25,237.97) ;
%Straight Lines [id:da39177775797899805]
\draw    (586.45,186.67) -- (499.48,260.69) ;
%Shape: Triangle [id:dp5580769332832513]
\draw  [fill={rgb, 255:red, 255; green, 255; blue, 255 }  ,fill opacity=1 ] (531.71,209.93) -- (539.36,208.65) -- (534.68,203.11) -- cycle ;
%Shape: Triangle [id:dp29022656368536337]
\draw  [fill={rgb, 255:red, 255; green, 255; blue, 255 }  ,fill opacity=1 ] (517.05,200.45) -- (524.71,199.18) -- (520.02,193.64) -- cycle ;
%Shape: Triangle [id:dp26784593385727296]
\draw  [fill={rgb, 255:red, 255; green, 255; blue, 255 }  ,fill opacity=1 ] (548.07,219.21) -- (555.69,217.69) -- (550.82,212.31) -- cycle ;
%Straight Lines [id:da13967799094847377]
\draw    (585.05,155.57) -- (499.6,70.38) ;
%Straight Lines [id:da7641279602215851]
\draw    (568,159.98) -- (499.85,92.38) ;
%Straight Lines [id:da8824085322891522]
\draw    (540.8,160.77) -- (499.45,119.17) ;
%Shape: Triangle [id:dp6848866413038792]
\draw  [fill={rgb, 255:red, 255; green, 255; blue, 255 }  ,fill opacity=1 ] (529.12,122.18) -- (531.49,129.57) -- (536.3,124.12) -- cycle ;
%Shape: Triangle [id:dp3846958458101022]
\draw  [fill={rgb, 255:red, 255; green, 255; blue, 255 }  ,fill opacity=1 ] (516.06,135.66) -- (517.83,143.21) -- (523.05,138.17) -- cycle ;
%Shape: Triangle [id:dp93786761029164]
\draw  [fill={rgb, 255:red, 255; green, 255; blue, 255 }  ,fill opacity=1 ] (540.72,110.37) -- (542.35,117.96) -- (547.66,113.02) -- cycle ;
%Shape: Rectangle [id:dp2590823679231101]
\draw   (391.25,160.38) -- (577.6,160.38) -- (577.6,176.77) -- (391.25,176.77) -- cycle ;
%Straight Lines [id:da8429872942196789]
\draw    (499.6,45.57) -- (500.05,138.38) -- (500.2,160.18) ;
%Curve Lines [id:da8547955502777654]
\draw    (585.05,155.57) .. controls (584.67,155.1) and (590.4,160.38) .. (594,166.38) .. controls (597.6,172.38) and (593.2,180.77) .. (586.45,186.67) ;
%Shape: Triangle [id:dp6969792921310897]
\draw  [fill={rgb, 255:red, 255; green, 255; blue, 255 }  ,fill opacity=1 ] (503.95,144.61) -- (500.05,138.38) -- (496.09,144.46) -- cycle ;
%Shape: Triangle [id:dp2205259078347982]
\draw  [fill={rgb, 255:red, 255; green, 255; blue, 255 }  ,fill opacity=1 ] (504.59,188.49) -- (496.83,188.63) -- (500.66,194.8) -- cycle ;
%Straight Lines [id:da34724862648314303]
\draw    (423.92,160.38) -- (423.27,60.69) ;
%Straight Lines [id:da5486449644838334]
\draw    (463.25,160.88) -- (463.02,59.36) ;
%Straight Lines [id:da039643083030675985]
\draw    (404.75,160.38) -- (404.22,59.63) ;
%Straight Lines [id:da8446098857448882]
\draw    (425.08,263.3) -- (424.43,178.38) ;
%Straight Lines [id:da28184992707621725]
\draw    (463.75,262.88) -- (463.18,176.56) ;
%Straight Lines [id:da27036336589584664]
\draw    (404.25,262.38) -- (404.38,176.79) ;
%Shape: Triangle [id:dp08670033656188558]
\draw  [fill={rgb, 255:red, 255; green, 255; blue, 255 }  ,fill opacity=1 ] (503.23,274.13) -- (499.01,267.9) -- (495.57,274.29) -- cycle ;
%Shape: Triangle [id:dp589784484605331]
\draw  [fill={rgb, 255:red, 255; green, 255; blue, 255 }  ,fill opacity=1 ] (499.56,64.35) -- (495.5,57.38) -- (503.59,57.32) -- cycle ;

% Text Node
\draw (200.77,161.83) node [anchor=north west][inner sep=0.75pt]    {$a$};
% Text Node
\draw (55,147.9) node [anchor=north west][inner sep=0.75pt]  [font=\Large]  {$\mathrm{scr}_{n} :$};
% Text Node
\draw (311.83,152.07) node [anchor=north west][inner sep=0.75pt]  [font=\Large]  {$\longmapsto $};
% Text Node
\draw (143.33,224.57) node [anchor=north west][inner sep=0.75pt]    {$l$};
% Text Node
\draw (142.67,91.9) node [anchor=north west][inner sep=0.75pt]    {$l$};
% Text Node
\draw (160.3,77.93) node [anchor=north west][inner sep=0.75pt]  [font=\normalsize]  {$\overbrace{x\ \ \ \ \ \ \ x}^{n}$};
% Text Node
\draw (177.4,121.97) node [anchor=north west][inner sep=0.75pt]    {$...$};
% Text Node
\draw (178.07,190.63) node [anchor=north west][inner sep=0.75pt]    {$...$};
% Text Node
\draw (222.47,227.1) node [anchor=north west][inner sep=0.75pt]  [font=\normalsize]  {$x$};
% Text Node
\draw (260.53,226.9) node [anchor=north west][inner sep=0.75pt]  [font=\normalsize]  {$x$};
% Text Node
\draw (219.8,94.43) node [anchor=north west][inner sep=0.75pt]  [font=\normalsize]  {$x$};
% Text Node
\draw (258.87,94.9) node [anchor=north west][inner sep=0.75pt]  [font=\normalsize]  {$x$};
% Text Node
\draw (236.4,121.97) node [anchor=north west][inner sep=0.75pt]    {$...$};
% Text Node
\draw (237.07,190.63) node [anchor=north west][inner sep=0.75pt]    {$...$};
% Text Node
\draw (585.4,141.57) node [anchor=north west][inner sep=0.75pt]  [font=\footnotesize]  {$x$};
% Text Node
\draw (487.87,204.78) node [anchor=north west][inner sep=0.75pt]  [font=\scriptsize]  {$\mu $};
% Text Node
\draw (488.4,229.98) node [anchor=north west][inner sep=0.75pt]  [font=\scriptsize]  {$\mu $};
% Text Node
\draw (488,252.57) node [anchor=north west][inner sep=0.75pt]  [font=\scriptsize]  {$\mu $};
% Text Node
\draw (487.6,67.24) node [anchor=north west][inner sep=0.75pt]  [font=\scriptsize]  {$\mu $};
% Text Node
\draw (487.8,89.77) node [anchor=north west][inner sep=0.75pt]  [font=\scriptsize]  {$\mu $};
% Text Node
\draw (487.2,115.37) node [anchor=north west][inner sep=0.75pt]  [font=\scriptsize]  {$\mu $};
% Text Node
\draw (510.2,209.97) node [anchor=north west][inner sep=0.75pt]    {$...$};
% Text Node
\draw (515.4,101.98) node [anchor=north west][inner sep=0.75pt]    {$...$};
% Text Node
\draw (525,83.37) node [anchor=north west][inner sep=0.75pt]  [font=\scriptsize]  {$Q$};
% Text Node
\draw (523.8,239.57) node [anchor=north west][inner sep=0.75pt]  [font=\scriptsize]  {$Q$};
% Text Node
\draw (471.6,164.98) node [anchor=north west][inner sep=0.75pt]    {$a$};
% Text Node
\draw (400.67,40.9) node [anchor=north west][inner sep=0.75pt]    {$l$};
% Text Node
\draw (435.4,70.97) node [anchor=north west][inner sep=0.75pt]    {$...$};
% Text Node
\draw (400.83,268.58) node [anchor=north west][inner sep=0.75pt]    {$l$};
% Text Node
\draw (435.57,212.03) node [anchor=north west][inner sep=0.75pt]    {$...$};
% Text Node
\draw (485.03,279.55) node [anchor=north west][inner sep=0.75pt]  [font=\normalsize]  {$x$};
% Text Node
\draw (484.87,41.4) node [anchor=north west][inner sep=0.75pt]  [font=\normalsize]  {$x$};
% Text Node
\draw (160.8,227.43) node [anchor=north west][inner sep=0.75pt]  [font=\normalsize]  {$\underbrace{x\ \ \ \ \ \ \ x}_{n}$};
% Text Node
\draw (418.8,270.43) node [anchor=north west][inner sep=0.75pt]  [font=\normalsize]  {$\underbrace{x\ \ \ \ \ \ \ x}_{n}$};
% Text Node
\draw (419.8,27.93) node [anchor=north west][inner sep=0.75pt]  [font=\normalsize]  {$\overbrace{x\ \ \ \ \ \ \ x}^{n}$};

\end{tikzpicture}

%% file: diagrams/module-morphism-2555.tex
\begin{tikzpicture}[x=0.75pt,y=0.75pt,yscale=-1,xscale=1,baseline = (current bounding box.center)]
%Straight Lines [id:da7352763592572981]
\draw    (185.8,157.53) -- (184.48,234.45) -- (183.8,272.73) ;
%Straight Lines [id:da4381052155843851]
\draw    (226.2,157.93) -- (184.96,196.94) ;
%Straight Lines [id:da327308490811242]
\draw    (255.4,157.93) -- (185.25,218.73) ;
%Straight Lines [id:da4558025109086812]
\draw    (271.45,167.43) -- (184.48,241.45) ;
%Shape: Triangle [id:dp683939765530144]
\draw  [fill={rgb, 255:red, 255; green, 255; blue, 255 }  ,fill opacity=1 ] (216.71,190.68) -- (224.36,189.41) -- (219.68,183.87) -- cycle ;
%Shape: Triangle [id:dp6746306324914118]
\draw  [fill={rgb, 255:red, 255; green, 255; blue, 255 }  ,fill opacity=1 ] (202.05,181.21) -- (209.71,179.94) -- (205.02,174.4) -- cycle ;
%Shape: Triangle [id:dp5348065473022808]
\draw  [fill={rgb, 255:red, 255; green, 255; blue, 255 }  ,fill opacity=1 ] (233.07,199.97) -- (240.69,198.45) -- (235.82,193.06) -- cycle ;
%Straight Lines [id:da21408076286098732]
\draw    (270.05,136.33) -- (184.6,51.13) ;
%Straight Lines [id:da8772144191532312]
\draw    (253,140.73) -- (184.85,73.13) ;
%Straight Lines [id:da1634863018534074]
\draw    (225.8,141.53) -- (184.45,99.93) ;
%Shape: Triangle [id:dp6058572298781101]
\draw  [fill={rgb, 255:red, 255; green, 255; blue, 255 }  ,fill opacity=1 ] (214.12,102.93) -- (216.49,110.32) -- (221.3,104.88) -- cycle ;
%Shape: Triangle [id:dp22421612920739697]
\draw  [fill={rgb, 255:red, 255; green, 255; blue, 255 }  ,fill opacity=1 ] (201.06,116.42) -- (202.83,123.97) -- (208.05,118.93) -- cycle ;
%Shape: Triangle [id:dp1579722863607239]
\draw  [fill={rgb, 255:red, 255; green, 255; blue, 255 }  ,fill opacity=1 ] (225.72,91.13) -- (227.35,98.72) -- (232.66,93.78) -- cycle ;
%Shape: Rectangle [id:dp9799492607315357]
\draw   (175.2,141.13) -- (262.6,141.13) -- (262.6,157.53) -- (175.2,157.53) -- cycle ;
%Straight Lines [id:da32449895079911606]
\draw    (184.6,26.33) -- (185.05,119.14) -- (185.2,140.93) ;
%Curve Lines [id:da8501073720147855]
\draw    (270.05,136.33) .. controls (269.67,135.86) and (275.4,141.13) .. (279,147.13) .. controls (282.6,153.13) and (278.2,161.53) .. (271.45,167.43) ;
%Shape: Triangle [id:dp5433000565796224]
\draw  [fill={rgb, 255:red, 255; green, 255; blue, 255 }  ,fill opacity=1 ] (188.95,125.36) -- (185.05,119.14) -- (181.09,125.22) -- cycle ;
%Shape: Triangle [id:dp2944296141664239]
\draw  [fill={rgb, 255:red, 255; green, 255; blue, 255 }  ,fill opacity=1 ] (189.59,169.25) -- (181.83,169.39) -- (185.66,175.56) -- cycle ;
% Text Node
\draw (270.4,122.33) node [anchor=north west][inner sep=0.75pt]  [font=\footnotesize]  {$x$};
% Text Node
\draw (167.2,258.33) node [anchor=north west][inner sep=0.75pt]  [font=\footnotesize]  {$K_{2}$};
% Text Node
\draw (169.87,186.53) node [anchor=north west][inner sep=0.75pt]  [font=\scriptsize]  {$\mu _{2}^{\dagger }$};
% Text Node
\draw (170.4,213.73) node [anchor=north west][inner sep=0.75pt]  [font=\scriptsize]  {$\mu _{2}^{\dagger }$};
% Text Node
\draw (170,235.33) node [anchor=north west][inner sep=0.75pt]  [font=\scriptsize]  {$\mu _{2}^{\dagger }$};
% Text Node
\draw (169.6,48) node [anchor=north west][inner sep=0.75pt]  [font=\scriptsize]  {$\mu _{1}$};
% Text Node
\draw (169.8,70.53) node [anchor=north west][inner sep=0.75pt]  [font=\scriptsize]  {$\mu _{1}$};
% Text Node
\draw (169.2,96.13) node [anchor=north west][inner sep=0.75pt]  [font=\scriptsize]  {$\mu _{1}$};
% Text Node
\draw (195.2,190.73) node [anchor=north west][inner sep=0.75pt]    {$...$};
% Text Node
\draw (200.4,82.73) node [anchor=north west][inner sep=0.75pt]    {$...$};
% Text Node
\draw (210,64.13) node [anchor=north west][inner sep=0.75pt]  [font=\scriptsize]  {$Q$};
% Text Node
\draw (208.8,220.33) node [anchor=north west][inner sep=0.75pt]  [font=\scriptsize]  {$Q$};
% Text Node
\draw (167.8,26.53) node [anchor=north west][inner sep=0.75pt]  [font=\footnotesize]  {$K_{1}$};
% Text Node
\draw (215.6,145.73) node [anchor=north west][inner sep=0.75pt]    {$a$};
\end{tikzpicture}

%% file: diagrams/edge-mode-2810.tex
\begin{tikzpicture}[
    x=1cm,
    y=1cm,
    pairbox/.style={
        draw,
        ellipse,
        minimum width=1.25cm,
        minimum height=0.62cm,
        inner sep=0pt
    },
    redstring/.style={
        red!80!black,
        dashed,
        dash pattern=on 4pt off 4pt,
        line width=0.9pt,
        line cap=round
    },
    lab/.style={
        font=\small,
        inner sep=1pt
    }
]

% parameters
\def\dx{1.55}      % horizontal distance
\def\ytop{1.85}    % top height
\def\xin{0.25}     % left/right label offset inside ellipse
\def\ybelow{-0.62} % label below ellipse

% centers
\foreach \i in {0,...,4} {
    \coordinate (C\i) at ({\i*\dx},0);
}

% ellipses
\foreach \i in {0,...,4} {
    \node[pairbox] at (C\i) {};
}

% inside labels: replace a_i,b_i by your symbols
\node[lab] at ($(C0)+(-\xin,0)$) {$\bullet$};
\node[lab] at ($(C0)+(\xin,0)$)  {$\bullet$};

\node[lab] at ($(C1)+(-\xin,0)$) {$\bullet$};
\node[lab] at ($(C1)+(\xin,0)$)  {$\bullet$};

\node[lab] at ($(C2)+(-\xin,0)$) {$\bullet$};
\node[lab] at ($(C2)+(\xin,0)$)  {$\bullet$};

\node[lab] at ($(C3)+(-\xin,0)$) {$\bullet$};
\node[lab] at ($(C3)+(\xin,0)$)  {$\bullet$};

\node[lab] at ($(C4)+(-\xin,0)$) {$\bullet$};
\node[lab] at ($(C4)+(\xin,0)$)  {$\bullet$};

% bottom labels
\node[lab] at ($(C0)+(-\xin,\ybelow)$) {$N$};
\node[lab] at ($(C0)+(\xin,\ybelow)$)  {$M^{*}$};

\foreach \i in {1,...,4} {
    \node[lab] at ($(C\i)+(-\xin,\ybelow)$) {$M$};
    \node[lab] at ($(C\i)+(\xin,\ybelow)$)  {$M^{*}$};
}

% red dashed vertical line
\draw[redstring]
    ($(C0)+(-\xin,0)$)
    --
    ($(C0)+(-\xin,\ytop)$);

% red dashed caps between adjacent ellipses
\foreach \i [evaluate=\i as \j using int(\i+1)] in {0,...,3} {

    % lower cup
    \draw[redstring]
        ($(C\i)+(\xin,0)$)
        .. controls
        ($(C\i)+(\xin,1.05)$)
        and
        ($(C\j)+(-\xin,1.05)$)
        ..
        ($(C\j)+(-\xin,0)$);

    % upper cap
    \draw[redstring]
        ($(C\i)+(\xin,\ytop)$)
        .. controls
        ($(C\i)+(\xin,0.85)$)
        and
        ($(C\j)+(-\xin,0.85)$)
        ..
        ($(C\j)+(-\xin,\ytop)$);
}

\end{tikzpicture}

%% file: diagrams/bulk-boundary-action-3473.tex
\begin{tikzcd}
	{Z_1(\C^\rev)} &&&& {(\CMdual)^{\rev}} \\
	\\
	\\
	{\DHR(\Loc^{\bulk}_{\bullet})} &&&& {\DHR(\Loc^{\bdy}_{\bullet})}
	\arrow["{\text{the canonical action}, K } ", from=1-1, to=1-5]
	\arrow["{\ReaDHR^{\bulk}}"', from=1-1, to=4-1]
	\arrow["{\ReaDHR^{\bdy}}", from=1-5, to=4-5]
	\arrow["U", from=4-1, to=4-5]
\end{tikzcd}

%% file: diagrams/canonical-action-3492.tex
\begin{tikzcd}
K_{z,\sigma}(m\otimes y)=(m\otimes y)\otimes z \arrow[r, "\alpha"] \arrow[d, "\beta_{m,y}"'] & m\otimes (y\otimes z) \arrow[d, "\mathrm{id}_m\otimes \sigma_{z,y}"] \\
K_{z,\sigma}(m)\otimes y=(m\otimes z)\otimes y \arrow[r, "\alpha"'] & m\otimes (z\otimes y)
\end{tikzcd}

%% file: diagrams/bending-isomorphism-3635.tex
\tikzset{every picture/.style={line width=0.75pt}} %set default line width to 0.75pt

\begin{tikzpicture}[x=0.75pt,y=0.75pt,yscale=-1,xscale=1]
%uncomment if require: \path (0,331); %set diagram left start at 0, and has height of 331

%Shape: Rectangle [id:dp2292328179477332]
\draw   (53.87,191.2) -- (141.27,191.2) -- (141.27,207.6) -- (53.87,207.6) -- cycle ;
%Straight Lines [id:da292567828925994]
\draw    (91.25,191.67) -- (90.6,152.01) ;
%Straight Lines [id:da055215715408650334]
\draw    (132,191.14) -- (131.35,151.48) ;
%Straight Lines [id:da4097016292599651]
\draw    (64.2,191.24) -- (63.55,151.58) ;
%Straight Lines [id:da17464771029702963]
\draw    (94.3,258.52) -- (94.67,207.73) ;
%Shape: Rectangle [id:dp07941473248567854]
\draw   (53.87,135.2) -- (141.27,135.2) -- (141.27,151.6) -- (53.87,151.6) -- cycle ;
%Shape: Rectangle [id:dp2204427637909876]
\draw   (286.87,190.7) -- (374.27,190.7) -- (374.27,207.1) -- (286.87,207.1) -- cycle ;
%Straight Lines [id:da20672672582080365]
\draw    (324.25,191.17) -- (323.6,151.51) ;
%Straight Lines [id:da0718567138390439]
\draw    (365,190.64) -- (364.35,150.98) ;
%Straight Lines [id:da49438692121257477]
\draw    (297.2,190.74) -- (296.55,151.08) ;
%Straight Lines [id:da14516070718313667]
\draw    (327.3,258.02) -- (327.67,207.23) ;
%Shape: Rectangle [id:dp44396625595859074]
\draw   (286.87,134.7) -- (374.27,134.7) -- (374.27,151.1) -- (286.87,151.1) -- cycle ;
%Straight Lines [id:da1725384029945778]
\draw    (271.25,256.34) -- (270.59,99.06) ;
%Shape: Rectangle [id:dp34108102221854597]
\draw   (497.37,189.2) -- (584.77,189.2) -- (584.77,205.6) -- (497.37,205.6) -- cycle ;
%Straight Lines [id:da4495560377246216]
\draw    (534.75,189.67) -- (534.1,150.01) ;
%Straight Lines [id:da48436324289991695]
\draw    (575.5,189.14) -- (574.85,149.48) ;
%Straight Lines [id:da2947202648247895]
\draw    (507.7,189.24) -- (507.05,149.58) ;
%Straight Lines [id:da06004974254191309]
\draw    (523.8,256.52) -- (524.17,205.73) ;
%Shape: Rectangle [id:dp05148677120431966]
\draw   (497.37,133.2) -- (584.77,133.2) -- (584.77,149.6) -- (497.37,149.6) -- cycle ;
%Straight Lines [id:da628983335929163]
\draw    (533.25,133.67) -- (532.6,80.25) ;
%Straight Lines [id:da06438929093028367]
\draw    (573.5,133.45) -- (572.85,80.03) ;
%Straight Lines [id:da7194810622730917]
\draw    (507.7,133.1) -- (507.05,79.68) ;
%Straight Lines [id:da4247332297512577]
\draw    (561.3,256.52) -- (561.67,205.73) ;
%Straight Lines [id:da6289304952156267]
\draw    (100.75,135.17) -- (100.1,81.75) ;
%Straight Lines [id:da8582931012720256]
\draw    (134.5,134.45) -- (133.85,81.03) ;
%Straight Lines [id:da6226054926608796]
\draw    (82.2,134.6) -- (81.55,81.18) ;
%Straight Lines [id:da39482636625100564]
\draw    (62.5,134.45) -- (61.85,81.03) ;
%Straight Lines [id:da5417737639389525]
\draw    (335.25,134.67) -- (334.6,81.25) ;
%Straight Lines [id:da19053618032919817]
\draw    (369,133.95) -- (368.35,80.53) ;
%Straight Lines [id:da1236104415437389]
\draw    (316.7,134.1) -- (316.05,80.68) ;
%Straight Lines [id:da7169039508609485]
\draw    (297,133.95) -- (296.75,98.84) ;
%Shape: Arc [id:dp19631955146966396]
\draw  [draw opacity=0] (270.59,99.06) .. controls (270.59,99.01) and (270.59,98.96) .. (270.59,98.91) .. controls (270.59,91.69) and (276.44,85.83) .. (283.67,85.83) .. controls (290.87,85.83) and (296.71,91.65) .. (296.75,98.84) -- (283.67,98.91) -- cycle ; \draw   (270.59,99.06) .. controls (270.59,99.01) and (270.59,98.96) .. (270.59,98.91) .. controls (270.59,91.69) and (276.44,85.83) .. (283.67,85.83) .. controls (290.87,85.83) and (296.71,91.65) .. (296.75,98.84) ;

% Text Node
\draw (76.95,236.22) node [anchor=north west][inner sep=0.75pt]  [font=\normalsize]  {$K$};
% Text Node
\draw (53,157.87) node [anchor=north west][inner sep=0.75pt]    {$l$};
% Text Node
\draw (79.13,159.49) node [anchor=north west][inner sep=0.75pt]  [font=\footnotesize]  {$x$};
% Text Node
\draw (119.2,159.32) node [anchor=north west][inner sep=0.75pt]  [font=\footnotesize]  {$x$};
% Text Node
\draw (102.73,162.8) node [anchor=north west][inner sep=0.75pt]    {$...$};
% Text Node
\draw (179.5,152.23) node [anchor=north west][inner sep=0.75pt]  [font=\Large]  {$\longmapsto $};
% Text Node
\draw (309.95,235.72) node [anchor=north west][inner sep=0.75pt]  [font=\normalsize]  {$K$};
% Text Node
\draw (286,157.37) node [anchor=north west][inner sep=0.75pt]    {$l$};
% Text Node
\draw (312.13,158.99) node [anchor=north west][inner sep=0.75pt]  [font=\footnotesize]  {$x$};
% Text Node
\draw (352.2,158.82) node [anchor=north west][inner sep=0.75pt]  [font=\footnotesize]  {$x$};
% Text Node
\draw (335.73,162.3) node [anchor=north west][inner sep=0.75pt]    {$...$};
% Text Node
\draw (249.2,234.06) node [anchor=north west][inner sep=0.75pt]  [font=\normalsize]  {$L^{\lor }$};
% Text Node
\draw (419,147.23) node [anchor=north west][inner sep=0.75pt]  [font=\huge]  {$=$};
% Text Node
\draw (502.45,234.22) node [anchor=north west][inner sep=0.75pt]  [font=\normalsize]  {$L^{\lor }$};
% Text Node
\draw (496.5,155.87) node [anchor=north west][inner sep=0.75pt]    {$l$};
% Text Node
\draw (522.63,157.49) node [anchor=north west][inner sep=0.75pt]  [font=\footnotesize]  {$x$};
% Text Node
\draw (562.7,157.32) node [anchor=north west][inner sep=0.75pt]  [font=\footnotesize]  {$x$};
% Text Node
\draw (546.23,160.8) node [anchor=north west][inner sep=0.75pt]    {$...$};
% Text Node
\draw (503.5,59.8) node [anchor=north west][inner sep=0.75pt]    {$l$};
% Text Node
\draw (528.13,61.94) node [anchor=north west][inner sep=0.75pt]  [font=\normalsize]  {$x$};
% Text Node
\draw (567.7,62.21) node [anchor=north west][inner sep=0.75pt]  [font=\normalsize]  {$x$};
% Text Node
\draw (549.73,96.38) node [anchor=north west][inner sep=0.75pt]    {$...$};
% Text Node
\draw (543.95,234.22) node [anchor=north west][inner sep=0.75pt]  [font=\normalsize]  {$K$};
% Text Node
\draw (78,61.3) node [anchor=north west][inner sep=0.75pt]    {$l$};
% Text Node
\draw (95.63,63.44) node [anchor=north west][inner sep=0.75pt]  [font=\normalsize]  {$x$};
% Text Node
\draw (128.7,63.21) node [anchor=north west][inner sep=0.75pt]  [font=\normalsize]  {$x$};
% Text Node
\draw (110.23,94.88) node [anchor=north west][inner sep=0.75pt]    {$...$};
% Text Node
\draw (56.7,59.56) node [anchor=north west][inner sep=0.75pt]  [font=\normalsize]  {$L$};
% Text Node
\draw (312.5,60.8) node [anchor=north west][inner sep=0.75pt]    {$l$};
% Text Node
\draw (330.13,62.94) node [anchor=north west][inner sep=0.75pt]  [font=\normalsize]  {$x$};
% Text Node
\draw (363.2,62.71) node [anchor=north west][inner sep=0.75pt]  [font=\normalsize]  {$x$};
% Text Node
\draw (344.73,94.38) node [anchor=north west][inner sep=0.75pt]    {$...$};
\end{tikzpicture}

%% file: diagrams/module-equivalence-3819.tex
\begin{tikzcd}
	{(\CMdual)^{\rev} \times \MQ^{\op}} && {(\CMdual)^{\op} \times \MQ^{\op}} && {\MQ^{\op}} \\
	\\
	\\
	{\DHR(\Loc^{\bdy}_\bullet) \times \BCond} &&&& \BCond
	\arrow["{\text{Bending}}", from=1-1, to=1-3]
	\arrow["{\ReaDHR \times \ReaBCond}", from=1-1, to=4-1]
	\arrow["{\text{Canonical action}}", from=1-3, to=1-5]
	\arrow["\ReaBCond", from=1-5, to=4-5]
	\arrow["{\boxtimes_{\Loc^{\bdy}}}", from=4-1, to=4-5]
    \arrow[phantom, from=1-1, to=4-1, "\rotatebox{90}{$\sim$}"', xshift=-1.5mm]
    \arrow[phantom, from=1-5, to=4-5, "\rotatebox{90}{$\sim$}"', xshift=-1.5mm]
\end{tikzcd}

%% file: diagrams/phase-description-3843.tex
\begin{tikzpicture}[x=0.75pt,y=0.75pt,yscale=-1,xscale=1]

%Shape: Rectangle [id:dp8035033636374661]
\draw  [draw opacity=0][fill={rgb, 255:red, 155; green, 155; blue, 155 }  ,fill opacity=0.28 ] (187,98.33) -- (378.33,98.33) -- (378.33,210.69) -- (187,210.69) -- cycle ;
%Straight Lines [id:da8331466086229126]
\draw [color={rgb, 255:red, 128; green, 128; blue, 128 }  ,draw opacity=1 ][line width=2.25]    (187,98.33) -- (188,203.02) ;
%Straight Lines [id:da12321469494819381]
\draw [line width=2.25]    (195,210.69) -- (378.33,210.69) ;
%Shape: Square [id:dp03242459059933889]
\draw  [fill={rgb, 255:red, 255; green, 255; blue, 255 }  ,fill opacity=1 ] (180.67,203.35) -- (195.49,203.35) -- (195.49,218.18) -- (180.67,218.18) -- cycle ;

% Text Node
\draw (134,70.4) node [anchor=north west][inner sep=0.75pt]    {$\mathrm{SymTFT}_{\mathrm{bdy}} \simeq (\CMdual)^{\rev}$};
% Text Node
\draw (290.67,219.07) node [anchor=north west][inner sep=0.75pt]    {$\mathrm{TopDef} \simeq (\QCQ)^{\rev,\op}$};
% Text Node
\draw (107.33,222.07) node [anchor=north west][inner sep=0.75pt]    {$\BCond \simeq \mathcal{M}^{\op}_{Q}$};
% Text Node
\draw (212.67,147.73) node [anchor=north west][inner sep=0.75pt]    {$\mathrm{SymTFT}_{\mathrm{bulk}} \simeq Z_1(\mathcal{C}^\rev)$};

\end{tikzpicture}

%% file: diagrams/enriched-associativity-3974.tex
\begin{tikzcd}
	{^{\cA}\cL(y,z)\ot ^{\cA}\cL(x,y) \ot ^{\cA} \cL(w,x)} && {^{\cA}\cL(y,z)\ot ^{\cA}\cL(w,y)} \\
	{^{\cA}\cL(x,z)\ot ^{\cA} \cL(w,x)} && {^{\cA}\cL(w,z)}
	\arrow["{1\otimes \circ}", from=1-1, to=1-3]
	\arrow["{\circ \otimes 1}"', from=1-1, to=2-1]
	\arrow["\circ", from=1-3, to=2-3]
	\arrow["\circ"', from=2-1, to=2-3]
\end{tikzcd}

%% file: diagrams/enriched-identity-3986.tex
\begin{tikzcd}
	{\mathbbm{1}_{\mathcal{A}} \otimes {^{\mathcal{A}}\mathcal{L}(x,y)}} & & {^{\mathcal{A}}\mathcal{L}(x,y)} & & {{^{\mathcal{A}}\mathcal{L}(x,y)} \otimes \mathbbm{1}_{\mathcal{A}}} \\
	{{^{\mathcal{A}}\mathcal{L}(y,y)} \otimes {^{\mathcal{A}}\mathcal{L}(x,y)}} & & & & {{^{\mathcal{A}}\mathcal{L}(x,y)} \otimes {^{\mathcal{A}}\mathcal{L}(x,x)}}
	\arrow[from=1-1, to=1-3]
	\arrow[from=1-5, to=1-3]
	\arrow["{1_y\otimes 1}"', from=1-1, to=2-1]
	\arrow["{1\otimes 1_x}", from=1-5, to=2-5]
	\arrow["\circ", from=2-1, to=1-3]
	\arrow["\circ"', from=2-5, to=1-3]
\end{tikzcd}

%% file: diagrams/enriched-functor-4044.tex
\begin{tikzcd}
	{\mathbbm{1}_\cB} && {\hat{F}(\mathbbm{1}_\cA)} \\
	{^{\cB}\cM(F(x),F(x))} && {\hat{F}(^{\cA}\cL(x,x))}
	\arrow[from=1-1, to=1-3]
	\arrow["{1_{F(x)}}"', from=1-1, to=2-1]
	\arrow["{\hat{F}(1_x)}", from=1-3, to=2-3]
	\arrow["{F_{x,x}}", from=2-3, to=2-1]
\end{tikzcd}

%% file: diagrams/functor-identity-4056.tex
\begin{tikzcd}
	{\hat{F}({}^{\mathcal{A}}\mathcal{L}(y,z))\otimes \hat{F}({}^{\mathcal{A}}\mathcal{L}(x,y))} & {\hat{F}({}^{\mathcal{A}}\mathcal{L}(y,z)\otimes {}^{\mathcal{A}}\mathcal{L}(x,y))} & {\hat{F}({}^{\mathcal{A}}\mathcal{L}(x,z))} \\
	{{}^{\mathcal{B}}\mathcal{M}(F(y),F(z))\otimes {}^{\mathcal{B}}\mathcal{M}(F(x),F(y))} & & {{}^{\mathcal{B}}\mathcal{M}(F(x),F(z))}
	\arrow[from=1-1, to=1-2]
	\arrow["{\hat{F}(\circ)}", from=1-2, to=1-3]
	\arrow["{F_{y,z}\otimes F_{x,y}}"', from=1-1, to=2-1]
	\arrow["{F_{x,z}}"', from=1-3, to=2-3]
	\arrow["\circ"', from=2-1, to=2-3].
\end{tikzcd}

%% file: diagrams/underlying-functor-4070.tex
\begin{tikzcd}
	{\mathbbm{1}_{\mathcal{B}}} & {\hat{F}(\mathbbm{1}_{\mathcal{A}})} & {\hat{F}({}^{\mathcal{A}}\mathcal{L}(x,y))} & {{}^{\mathcal{B}}\mathcal{M}(F(x),F(y))}
	\arrow[from=1-1, to=1-2]
	\arrow["{\hat{F}(f)}", from=1-2, to=1-3]
	\arrow["{F_{x,y}}", from=1-3, to=1-4]
	\arrow["{\underline{F}(f)}"', bend right, from=1-1, to=1-4].
\end{tikzcd}

%% file: diagrams/enriched-transformation-4096.tex
\begin{tikzcd}
	{\hat{F}(^{\cA}\cL(x,y))} && {\mathbbm{1}_\cB \ot \hat{F}(^{\cA}\cL(x,y))} && {^{\cB}\cM(F(y),G(y))\ot\ ^{\cB}\cM(F(x),F(y))} \\
	&&&& {^{\cB}\cM(F(x),G(y))} \\
	{\hat{G}(^{\cA}\cL(x,y))} && {\hat{G}(^{\cA}\cL(x,y))\ot \mathbbm{1}_\cB} && {^{\cB}\cM(G(x),G(y))\ot\ ^{\cB}\cM(F(x),G(x))}
	\arrow[from=1-1, to=1-3]
	\arrow["{\hat{\xi}}"', from=1-1, to=3-1]
	\arrow["{\xi_y\ot F_{x,y}}", from=1-3, to=1-5]
	\arrow["\circ", from=1-5, to=2-5]
	\arrow[from=3-1, to=3-3]
	\arrow["{G_{x,y}\ot \xi_x}"', from=3-3, to=3-5]
	\arrow["\circ"', from=3-5, to=2-5].
\end{tikzcd}

%% file: bib.bib
@article{Penneys_Add_2025,
  author        = {{Jones}, Corey and {Naaijkens}, Pieter and {Penneys}, David},
  title         = {{Holography for bulk-boundary local topological order}},
  journal       = {arXiv e-prints},
  year          = 2025,
  month         = jun,
  eid           = {arXiv:2506.19969},
  pages         = {arXiv:2506.19969},
  doi           = {10.48550/arXiv.2506.19969},
  archiveprefix = {arXiv},
  eprint        = {2506.19969},
  primaryclass  = {math-ph},
  adsurl        = {https://ui.adsabs.harvard.edu/abs/2025arXiv250619969J}
}

@article{Chen_Wen_2010_LocalUnitary,
   title={Local unitary transformation, long-range quantum entanglement, wave function renormalization, and topological order},
   volume={82},
   ISSN={1550-235X},
   url={http://dx.doi.org/10.1103/PhysRevB.82.155138},
   DOI={10.1103/physrevb.82.155138},
   eprint={1004.3835},
   archivePrefix={arXiv},
   primaryClass={cond-mat.str-el},
   number={15},
   journal={Physical Review B},
   publisher={American Physical Society (APS)},
   author={Chen, Xie and Gu, Zheng-Cheng and Wen, Xiao-Gang},
   year={2010},
   month=oct }

@ARTICLE{EB1,
       author = {{Shi}, Bowen and {Kato}, Kohtaro and {Kim}, Isaac H.},
        title = "{Fusion rules from entanglement}",
      journal = {Annals of Physics},
         year = 2020,
        month = jul,
       volume = {418},
          eid = {168164},
        pages = {168164},
          doi = {10.1016/j.aop.2020.168164},
archivePrefix = {arXiv},
       eprint = {1906.09376},
 primaryClass = {cond-mat.str-el},
       adsurl = {https://ui.adsabs.harvard.edu/abs/2020AnPhy.41868164S}
}

@ARTICLE{Jones_Naaijkens_Penneys_Wallick_Izumi_2025,
       author = {{Jones}, Corey and {Naaijkens}, Pieter and {Penneys}, David and {Wallick}, Daniel and {Izumi}, Masaki},
        title = "{Local topological order and boundary algebras}",
      journal = {Forum of Mathematics, Sigma},
         year = 2025,
        month = aug,
       volume = {13},
          eid = {e135},
        pages = {e135},
          doi = {10.1017/fms.2025.16},
archivePrefix = {arXiv},
       eprint = {2307.12552},
 primaryClass = {math-ph},
      adsurl = {https://ui.adsabs.harvard.edu/abs/2023arXiv230712552J}
}

@ARTICLE{EB2,
       author = {{Huang}, Jin-Long and {McGreevy}, John and {Shi}, Bowen},
        title = "{Knots and entanglement}",
      journal = {SciPost Physics},
         year = 2023,
        month = jun,
       volume = {14},
       number = {6},
          eid = {141},
        pages = {141},
          doi = {10.21468/SciPostPhys.14.6.141},
archivePrefix = {arXiv},
       eprint = {2112.08398},
 primaryClass = {hep-th},
       adsurl = {https://ui.adsabs.harvard.edu/abs/2023ScPP...14..141H}
}

@book{Murphy90,
title = {C*-Algebras and Operator Theory},
publisher = {Academic Press},
address = {San Diego},
year = {1990},
isbn = {978-0-08-092496-0},
doi = {https://doi.org/10.1016/B978-0-08-092496-0.50004-1},
url = {https://www.sciencedirect.com/science/article/pii/B9780080924960500041},
author = {Gerard J. Murphy}
}

@misc{Gaiotto_Johnson-Freyd_2025_CondensationsInHigherCats,
      title={Condensations in higher categories}, 
      author={Davide Gaiotto and Theo Johnson-Freyd},
      year={2025},
      eprint={1905.09566},
      archivePrefix={arXiv},
      primaryClass={math.CT},
      url={https://arxiv.org/abs/1905.09566}, 
}

@misc{Etingof_2000_IsocategoricalGroups,
      title={Isocategorical groups}, 
      author={Pavel Etingof and Shlomo Gelaki},
      year={2000},
      eprint={math/0007196},
      archivePrefix={arXiv},
      primaryClass={math.QA},
      url={https://arxiv.org/abs/math/0007196}, 
}

@article{Chen_Gu_Wen_2011_SPTMPS,
   title={Classification of gapped symmetric phases in one-dimensional spin systems},
   volume={83},
   ISSN={1550-235X},
   url={http://dx.doi.org/10.1103/PhysRevB.83.035107},
   DOI={10.1103/physrevb.83.035107},
   eprint={1008.3745},
   archivePrefix={arXiv},
   primaryClass={cond-mat.str-el},
   number={3},
   journal={Physical Review B},
   publisher={American Physical Society (APS)},
   author={Chen, Xie and Gu, Zheng-Cheng and Wen, Xiao-Gang},
   year={2011},
   month=jan }

@article{Cirac_2011_1DPhaseMPS,
   title={Classifying quantum phases using matrix product states and projected entangled pair states},
   volume={84},
   ISSN={1550-235X},
   url={http://dx.doi.org/10.1103/PhysRevB.84.165139},
   DOI={10.1103/physrevb.84.165139},
   eprint={1010.3732},
   archivePrefix={arXiv},
   primaryClass={cond-mat.str-el},
   number={16},
   journal={Physical Review B},
   publisher={American Physical Society (APS)},
   author={Schuch, Norbert and Pérez-García, David and Cirac, Ignacio},
   year={2011},
   month=oct }

@article{Kapustin_Sopenko_2021_SPTClassification,
   title={A classification of invertible phases of bosonic quantum lattice systems in one dimension},
   volume={62},
   ISSN={1089-7658},
   url={http://dx.doi.org/10.1063/5.0055996},
   DOI={10.1063/5.0055996},
   eprint={2012.15491},
   archivePrefix={arXiv},
   primaryClass={quant-ph},
   number={8},
   journal={Journal of Mathematical Physics},
   publisher={AIP Publishing},
   author={Kapustin, Anton and Sopenko, Nikita and Yang, Bowen},
   year={2021},
   month=aug }

@article{Naaijkens_2011_ToricCode,
   title={Localized endomorphisms in {K}itaev's toric code on the plane},
   volume={23},
   ISSN={1793-6659},
   url={http://dx.doi.org/10.1142/S0129055X1100431X},
   DOI={10.1142/s0129055x1100431x},
   eprint={1012.3857},
   archivePrefix={arXiv},
   primaryClass={math-ph},
   number={04},
   journal={Reviews in Mathematical Physics},
   publisher={World Scientific Pub Co Pte Ltd},
   author={Naaijkens, Pieter},
   year={2011},
   month=may, pages={347--373} }

@article{Naaijkens_Cha_2019_StabilityCharges,
  title={On the Stability of Charges in Infinite Quantum Spin Systems},
  volume={373},
  ISSN={1432-0916},
  url={http://dx.doi.org/10.1007/s00220-019-03630-1},
  DOI={10.1007/s00220-019-03630-1},
  eprint={1804.03203},
  archivePrefix={arXiv},
  primaryClass={math-ph},
  number={1},
  journal={Communications in Mathematical Physics},
  publisher={Springer Science and Business Media LLC},
  author={Cha, Matthew and Naaijkens, Pieter and Nachtergaele, Bruno},
  year={2019},
  month=dec, pages={219--264} }

@article{Naaijkens_2025_SectorsQD,
   title={The Category of Anyon Sectors for Non-Abelian Quantum Double Models},
   volume={407},
   url={https://doi.org/10.1007/s00220-025-05492-2},
   doi={10.1007/s00220-025-05492-2},
   eprint={2503.15611},
   archivePrefix={arXiv},
   primaryClass={math-ph},
   number={1},
   journal={Communications in Mathematical Physics},
   author={Bols, Alex and Hamdan, Mahdie and Naaijkens, Pieter and Vadnerkar, Siddharth},
   year={2025}
}

@incollection{ogata_2021_classificationgappedgroundstate,
   title={Classification of gapped ground state phases in quantum spin systems},
   booktitle={International Congress of Mathematicians},
   pages={4142--4161},
   doi={10.4171/icm2022/29},
   eprint={2110.04675},
   archivePrefix={arXiv},
   primaryClass={math-ph},
   url={https://doi.org/10.4171/icm2022/29},
   author={Ogata, Yoshiko},
   year={2023}
}

@article{Ogata_2022_BraidedCStarTensorCat,
   title={A derivation of braided {C}*-tensor categories from gapped ground states satisfying the approximate Haag duality},
   volume={63},
   ISSN={1089-7658},
   url={http://dx.doi.org/10.1063/5.0061785},
   DOI={10.1063/5.0061785},
   eprint={2106.15741},
   archivePrefix={arXiv},
   primaryClass={math-ph},
   number={1},
   journal={Journal of Mathematical Physics},
   publisher={AIP Publishing},
   author={Ogata, Yoshiko},
   year={2022},
   month=jan }

@article{Ogata_2021_ClassificationSPT,
   title={Classification of symmetry protected topological phases in quantum spin chains},
   volume={2020},
   url={https://doi.org/10.4310/cdm.2020.v2020.n1.a2},
   doi={10.4310/cdm.2020.v2020.n1.a2},
   eprint={2110.04671},
   archivePrefix={arXiv},
   primaryClass={math-ph},
   number={1},
   journal={Current Developments in Mathematics},
   author={Ogata, Yoshiko},
   year={2020},
   pages={41--104}
}

@misc{Ogata_2024_SET,
      title={Anyonic symmetry fractionalization in SET phases}, 
      author={Jose Garre Rubio and Yoshiko Ogata},
      year={2024},
      eprint={2411.01210},
      archivePrefix={arXiv},
      primaryClass={math-ph},
      url={https://arxiv.org/abs/2411.01210}, 
}

@misc{Kawagoe_Vadnerkar_2025_SET,
      title={An operator algebraic approach to symmetry defects and fractionalization}, 
      author={Kyle Kawagoe and Siddharth Vadnerkar and Daniel Wallick},
      year={2025},
      eprint={2410.23380},
      archivePrefix={arXiv},
      primaryClass={math-ph},
      url={https://arxiv.org/abs/2410.23380}, 
}

@article{Ogata_2025_MixedState,
   title={Mixed State Topological Order: Operator Algebraic Approach},
   volume={406},
   url={https://doi.org/10.1007/s00220-025-05475-3},
   doi={10.1007/s00220-025-05475-3},
   eprint={2501.02398},
   archivePrefix={arXiv},
   primaryClass={math-ph},
   number={12},
   journal={Communications in Mathematical Physics},
   author={Ogata, Yoshiko},
   year={2025}
}

@misc{Kong_2015_BbyBulk,
      title={Boundary-bulk relation for topological orders as the functor mapping higher categories to their centers}, 
      author={Liang Kong and Xiao-Gang Wen and Hao Zheng},
      year={2015},
      eprint={1502.01690},
      archivePrefix={arXiv},
      primaryClass={cond-mat.str-el},
      url={https://arxiv.org/abs/1502.01690}, 
}

@article{Kong_2017_BdyBulk,
    author = "Kong, Liang and Wen, Xiao-Gang and Zheng, Hao",
    title = "{Boundary-bulk relation in topological orders}",
    eprint = "1702.00673",
    archivePrefix = "arXiv",
    primaryClass = "cond-mat.str-el",
    doi = "10.1016/j.nuclphysb.2017.06.023",
    journal = "Nucl. Phys. B",
    volume = "922",
    pages = "62--76",
    year = "2017"
}

@article{Kong_2022_QL1,
   title={Categories of quantum liquids {I}},
   volume={2022},
   ISSN={1029-8479},
   url={http://dx.doi.org/10.1007/JHEP08(2022)070},
   DOI={10.1007/jhep08(2022)070},
   eprint={2011.02859},
   archivePrefix={arXiv},
   primaryClass={hep-th},
   number={8},
   journal={Journal of High Energy Physics},
   publisher={Springer Science and Business Media LLC},
   author={Kong, Liang and Zheng, Hao},
   year={2022},
   month=aug }

@article{Kong_2024_QL2,
   title={Categories of Quantum Liquids {II}},
   volume={405},
   ISSN={1432-0916},
   url={http://dx.doi.org/10.1007/s00220-024-05078-4},
   DOI={10.1007/s00220-024-05078-4},
   eprint={2107.03858},
   archivePrefix={arXiv},
   primaryClass={math.CT},
   number={9},
   journal={Communications in Mathematical Physics},
   publisher={Springer Science and Business Media LLC},
   author={Kong, Liang and Zheng, Hao},
   year={2024},
   month=aug }

@article{Johnson_Freyd_2022_OnClassificationTO,
   title={On the Classification of Topological Orders},
   volume={393},
   ISSN={1432-0916},
   url={http://dx.doi.org/10.1007/s00220-022-04380-3},
   DOI={10.1007/s00220-022-04380-3},
   eprint={2003.06663},
   archivePrefix={arXiv},
   primaryClass={math.CT},
   number={2},
   journal={Communications in Mathematical Physics},
   publisher={Springer Science and Business Media LLC},
   author={Johnson-Freyd, Theo},
   year={2022},
   month=apr, pages={989--1033} }

@article{Lan_2024_CatSET,
   title={Category of {SET} orders},
   volume={2024},
   ISSN={1029-8479},
   url={http://dx.doi.org/10.1007/JHEP11(2024)111},
   DOI={10.1007/jhep11(2024)111},
   eprint={2312.15958},
   archivePrefix={arXiv},
   primaryClass={cond-mat.str-el},
   number={11},
   journal={Journal of High Energy Physics},
   publisher={Springer Science and Business Media LLC},
   author={Lan, Tian and Yue, Gen and Wang, Longye},
   year={2024},
   month=nov }

@article{Kong_2020_GaplessEdge,
   title={A mathematical theory of gapless edges of 2d topological orders. Part {I}},
   volume={2020},
   ISSN={1029-8479},
   url={http://dx.doi.org/10.1007/JHEP02(2020)150},
   DOI={10.1007/jhep02(2020)150},
   eprint={1905.04924},
   archivePrefix={arXiv},
   primaryClass={cond-mat.str-el},
   number={2},
   journal={Journal of High Energy Physics},
   publisher={Springer Science and Business Media LLC},
   author={Kong, Liang and Zheng, Hao},
   year={2020},
   month=feb }

@article{Kong_2018_GaplessEnrichedCat,
   title={Gapless edges of 2d topological orders and enriched monoidal categories},
   volume={927},
   ISSN={0550-3213},
   url={http://dx.doi.org/10.1016/j.nuclphysb.2017.12.007},
   DOI={10.1016/j.nuclphysb.2017.12.007},
   eprint={1705.01087},
   archivePrefix={arXiv},
   primaryClass={cond-mat.str-el},
   journal={Nuclear Physics B},
   publisher={Elsevier BV},
   author={Kong, Liang and Zheng, Hao},
   year={2018},
   month=feb, pages={140--165} }

@article{Chen_2020_TopologicalPhaseTransition,
   title={Topological phase transition on the edge of two-dimensional topological order},
   volume={102},
   ISSN={2469-9969},
   url={http://dx.doi.org/10.1103/PhysRevB.102.045139},
   DOI={10.1103/physrevb.102.045139},
   eprint={1903.12334},
   archivePrefix={arXiv},
   primaryClass={cond-mat.str-el},
   number={4},
   journal={Physical Review B},
   publisher={American Physical Society (APS)},
   author={Chen, Wei-Qiang and Jian, Chao-Ming and Kong, Liang and You, Yi-Zhuang and Zheng, Hao},
   year={2020},
   month=jul }

@article{KitaevKong_2012_DomainWall,
   title={Models for Gapped Boundaries and Domain Walls},
   volume={313},
   ISSN={1432-0916},
   url={http://dx.doi.org/10.1007/s00220-012-1500-5},
   DOI={10.1007/s00220-012-1500-5},
   eprint={1104.5047},
   archivePrefix={arXiv},
   primaryClass={cond-mat.str-el},
   number={2},
   journal={Communications in Mathematical Physics},
   publisher={Springer Science and Business Media LLC},
   author={Kitaev, Alexei and Kong, Liang},
   year={2012},
   month=jun, pages={351--373} }

@article{Kong_2022_1DEnrichedCat,
   title={One dimensional gapped quantum phases and enriched fusion categories},
   volume={2022},
   ISSN={1029-8479},
   url={http://dx.doi.org/10.1007/JHEP03(2022)022},
   DOI={10.1007/jhep03(2022)022},
   eprint={2108.08835},
   archivePrefix={arXiv},
   primaryClass={cond-mat.str-el},
   number={3},
   journal={Journal of High Energy Physics},
   publisher={Springer Science and Business Media LLC},
   author={Kong, Liang and Wen, Xiao-Gang and Zheng, Hao},
   year={2022},
   month=mar }

@article{Kong_Yuan_2024_EnrichedMonoidalCat,
  title={Enriched monoidal categories {I}: Centers},
  volume={16},
  ISSN={1664-073X},
  url={http://dx.doi.org/10.4171/qt/217},
  DOI={10.4171/qt/217},
  eprint={2104.03121},
  archivePrefix={arXiv},
  primaryClass={math.CT},
  number={2},
  journal={Quantum Topology},
  publisher={European Mathematical Society - EMS - Publishing House GmbH},
  author={Kong, Liang and Yuan, Wei and Zhang, Zhi-Hao and Zheng, Hao},
  year={2024},
  month=may, pages={343--417} }

@misc{Kong_2022_Invitation,
      title={An invitation to topological orders and category theory}, 
      author={Liang Kong and Zhi-Hao Zhang},
      year={2022},
      eprint={2205.05565},
      archivePrefix={arXiv},
      primaryClass={cond-mat.str-el},
      url={https://arxiv.org/abs/2205.05565}, 
}

@article{Xu_2024_1DPhaseAbelianSym,
   title={Categorical descriptions of one-dimensional gapped phases with Abelian onsite symmetries},
   volume={110},
   ISSN={2469-9969},
   url={http://dx.doi.org/10.1103/PhysRevB.110.155106},
   DOI={10.1103/physrevb.110.155106},
   eprint={2205.09656},
   archivePrefix={arXiv},
   primaryClass={cond-mat.str-el},
   number={15},
   journal={Physical Review B},
   publisher={American Physical Society (APS)},
   author={Xu, Rongge and Zhang, Zhi-Hao},
   year={2024},
   month=oct }

@article{LanZhou_2024_QuantumCurrent,
   title={Quantum current and holographic categorical symmetry},
   volume={16},
   ISSN={2542-4653},
   url={http://dx.doi.org/10.21468/SciPostPhys.16.2.053},
   DOI={10.21468/scipostphys.16.2.053},
   eprint={2305.12917},
   archivePrefix={arXiv},
   primaryClass={cond-mat.str-el},
   number={2},
   journal={SciPost Physics},
   publisher={Stichting SciPost},
   author={Lan, Tian and Zhou, Jing-Ren},
   year={2024},
   month=feb }

@article{Jones_2024_DHR,
   title={{DHR} bimodules of quasi-local algebras and symmetric quantum cellular automata},
   volume={15},
   url={https://doi.org/10.4171/qt/216},
   doi={10.4171/qt/216},
   eprint={2304.00068},
   archivePrefix={arXiv},
   primaryClass={math-ph},
   number={3},
   journal={Quantum Topology},
   author={Jones, Corey},
   year={2024},
   pages={633--686}
}

@misc{Hataishi_2025_DHRAbstractChain,
  title         = {On the structure of {DHR} bimodules of abstract spin chains},
  author        = {Lucas Hataishi and David Jaklitsch and Corey Jones and Makoto Yamashita},
  year          = {2025},
  eprint        = {2504.06094},
  archiveprefix = {arXiv},
  primaryclass  = {math.QA},
  url           = {https://arxiv.org/abs/2504.06094}
}

@article{Chen_2024_QSysCompletion,
   title={Q-system completion for {C*} 2-categories},
   volume={283},
   url={https://doi.org/10.1016/j.jfa.2022.109524},
   doi={10.1016/j.jfa.2022.109524},
   eprint={2105.12010},
   archivePrefix={arXiv},
   primaryClass={math.OA},
   number={3},
   journal={Journal of Functional Analysis},
   author={Chen, Quan and Hernández Palomares, Roberto and Jones, Corey and Penneys, David},
   year={2022},
   pages={109524}
}

@article{Chen_2024_InductiveLimitAF,
   title={K-theoretic Classification of Inductive Limit Actions of Fusion Categories on {AF}-algebras},
   volume={405},
   url={https://doi.org/10.1007/s00220-024-04969-w},
   doi={10.1007/s00220-024-04969-w},
   eprint={2207.11854},
   archivePrefix={arXiv},
   primaryClass={math.OA},
   number={3},
   journal={Communications in Mathematical Physics},
   author={Chen, Quan and Hernández Palomares, Roberto and Jones, Corey},
   year={2024}
}

@article{Ost03,
   title={Module categories, weak Hopf algebras and modular invariants},
   volume={8},
   url={https://doi.org/10.1007/s00031-003-0515-6},
   doi={10.1007/s00031-003-0515-6},
   eprint={math/0111139},
   archivePrefix={arXiv},
   primaryClass={math.QA},
   number={2},
   journal={Transformation Groups},
   author={Ostrik, Victor},
   year={2003},
   pages={177--206}
}

@misc{MYLG25,
      title={Non-invertible SPTs: an on-site realization of (1+1)d anomaly-free fusion category symmetry}, 
      author={Chenqi Meng and Xinping Yang and Tian Lan and Zhengcheng Gu},
      year={2025},
      eprint={2412.20546},
      archivePrefix={arXiv},
      primaryClass={cond-mat.str-el},
      url={https://arxiv.org/abs/2412.20546}, 
}

@book{Naaijkens2017_book,
  author    = {Pieter Naaijkens},
  title     = {Quantum Spin Systems on Infinite Lattices: A Concise Introduction},
  series    = {Lecture Notes in Physics},
  volume    = {933},
  publisher = {Springer},
  year      = {2017},
  doi       = {10.1007/978-3-319-51458-1},
  eprint    = {1311.2717},
  archivePrefix = {arXiv},
  primaryClass  = {math-ph}
}

@book{Effros_1981_DimensionsAC,
    isbn = {0821816977},
    doi = {10.1090/cbms/046},
    url = {https://doi.org/10.1090/cbms/046},
    language = {eng},
    lccn = {81001582},
    publisher = {Published for the Conference Board of the Mathematical Sciences by the American Mathematical Society},
    series = {Regional conference series in mathematics; no. 46},
    title = {Dimensions and C*-algebras},
    year = {1981},
    author = {Effros, Edward G.},
    address = {Providence, R.I},
    booktitle = {Dimensions and C*-algebras},
}

@book{EGNO_2015,
  address   = {Providence, Rhode Island},
  series    = {Mathematical surveys and monographs},
  title     = {Tensor categories},
  isbn      = {978-1-4704-2024-6},
  language  = {en},
  number    = {volume 205},
  publisher = {American Mathematical Society},
  editor    = {Etingof, P. I. and Gelaki, Shlomo and Nikshych, Dmitri and Ostrik, Victor},
  year      = {2015}
}

@book{Blackadar_OperatorAlg_2017,
  address    = {Berlin ; New York},
  title      = {Operator algebras: theory of {C}*-algebras and von {Neumann} algebras},
  isbn       = {978-3-540-28486-4},
  shorttitle = {Operator algebras},
  language   = {en},
  publisher  = {Springer},
  author     = {Blackadar, Bruce},
  year       = {2017},
  note       = {Corrected version}
}

@misc{Wen_lecture_2016,
	title = {Lectures on topological order:  {Long} range entanglement  and topological excitations},
	url = {https://boulderschool.yale.edu/sites/default/files/files/16Boulder.pdf},
	language = {en},
	author = {Wen, Xiao-Gang},
	year = {2016},
}

@article{Giorgetti_2024,
   title={Separable algebras in multitensor {C}$ ^* $-categories are unitarizable},
   volume={9},
   ISSN={2473-6988},
   url={http://dx.doi.org/10.3934/math.2024555},
   DOI={10.3934/math.2024555},
   eprint={2312.12019},
   archivePrefix={arXiv},
   primaryClass={math.OA},
   number={5},
   journal={AIMS Mathematics},
   publisher={American Institute of Mathematical Sciences (AIMS)},
   author={Giorgetti, Luca and Yuan, Wei and Zhao, XuRui},
   year={2024},
   pages={11320--11334} }

@article{Giorgetti_2023,
   title={Realization of rigid {C}$^*$-bicategories as bimodules over type {II}$_1$ von Neumann algebras},
   volume={415},
   ISSN={0001-8708},
   url={http://dx.doi.org/10.1016/j.aim.2023.108886},
   DOI={10.1016/j.aim.2023.108886},
   eprint={2010.01072},
   archivePrefix={arXiv},
   primaryClass={math.CT},
   journal={Advances in Mathematics},
   publisher={Elsevier BV},
   author={Giorgetti, Luca and Yuan, Wei},
   year={2023},
   month=feb, pages={108886} }

@article{Seifnashri_2024,
   title={Cluster State as a Noninvertible Symmetry-Protected Topological Phase},
   volume={133},
   ISSN={1079-7114},
   url={http://dx.doi.org/10.1103/PhysRevLett.133.116601},
   DOI={10.1103/physrevlett.133.116601},
   eprint={2404.01369},
   archivePrefix={arXiv},
   primaryClass={cond-mat.str-el},
   number={11},
   journal={Physical Review Letters},
   publisher={American Physical Society (APS)},
   author={Seifnashri, Sahand and Shao, Shu-Heng},
   year={2024},
   month=sep }

@article{Bra72,
   author={Bratteli, Ola},
   title={Inductive limits of finite dimensional {$C^*$}-algebras},
   journal={Transactions of the American Mathematical Society},
   volume={171},
   pages={195--234},
   year={1972},
   publisher={American Mathematical Society},
   doi={10.1090/S0002-9947-1972-0312282-2}
}

@article{Witten_1988_QFTJones,
    author = "Witten, Edward",
    editor = "Mitra, Asoke N.",
    title = "{Quantum Field Theory and the Jones Polynomial}",
    reportNumber = "IASSNS-HEP-88-33",
    doi = "10.1007/BF01217730",
    journal = "Commun. Math. Phys.",
    volume = "121",
    pages = "351--399",
    year = "1989"
}

@article{Levin_2005,
   title={String-net condensation:A physical mechanism for topological phases},
   volume={71},
   ISSN={1550-235X},
   url={http://dx.doi.org/10.1103/PhysRevB.71.045110},
   DOI={10.1103/physrevb.71.045110},
   eprint={cond-mat/0404617},
   archivePrefix={arXiv},
   primaryClass={cond-mat.str-el},
   number={4},
   journal={Physical Review B},
   publisher={American Physical Society (APS)},
   author={Levin, Michael A. and Wen, Xiao-Gang},
   year={2005},
   month=Jan }

@article{Ji_2020_CategoricalSym,
   title={Categorical symmetry and noninvertible anomaly in symmetry-breaking and topological phase transitions},
   volume={2},
   ISSN={2643-1564},
   url={http://dx.doi.org/10.1103/PhysRevResearch.2.033417},
   DOI={10.1103/physrevresearch.2.033417},
   eprint={1912.13492},
   archivePrefix={arXiv},
   primaryClass={cond-mat.str-el},
   number={3},
   journal={Physical Review Research},
   publisher={American Physical Society (APS)},
   author={Ji, Wenjie and Wen, Xiao-Gang},
   year={2020},
   month=sep }

@article{Kong_2020Classi,
   title={Classification of topological phases with finite internal symmetries in all dimensions},
   volume={2020},
   ISSN={1029-8479},
   url={http://dx.doi.org/10.1007/JHEP09(2020)093},
   DOI={10.1007/jhep09(2020)093},
   eprint={2003.08898},
   archivePrefix={arXiv},
   primaryClass={math-ph},
   number={9},
   journal={Journal of High Energy Physics},
   publisher={Springer Science and Business Media LLC},
   author={Kong, Liang and Lan, Tian and Wen, Xiao-Gang and Zhang, Zhi-Hao and Zheng, Hao},
   year={2020},
   month=sep }

@article{Lichtman_2021,
   title={Bulk anyons as edge symmetries: Boundary phase diagrams of topologically ordered states},
   volume={104},
   ISSN={2469-9969},
   url={http://dx.doi.org/10.1103/PhysRevB.104.075141},
   DOI={10.1103/physrevb.104.075141},
   eprint={2003.04328},
   archivePrefix={arXiv},
   primaryClass={cond-mat.str-el},
   number={7},
   journal={Physical Review B},
   publisher={American Physical Society (APS)},
   author={Lichtman, Tsuf and Thorngren, Ryan and Lindner, Netanel H. and Stern, Ady and Berg, Erez},
   year={2021},
   month=aug }

@article{Chatterjee_2023,
   title={Symmetry as a shadow of topological order and a derivation of topological holographic principle},
   volume={107},
   ISSN={2469-9969},
   url={http://dx.doi.org/10.1103/PhysRevB.107.155136},
   DOI={10.1103/physrevb.107.155136},
   eprint={2203.03596},
   archivePrefix={arXiv},
   primaryClass={cond-mat.str-el},
   number={15},
   journal={Physical Review B},
   publisher={American Physical Society (APS)},
   author={Chatterjee, Arkya and Wen, Xiao-Gang},
   year={2023},
   month=apr }

@article{Moradi_2023,
   title={Topological holography: Towards a unification of Landau and  beyond-Landau physics},
   volume={6},
   ISSN={2666-9366},
   url={http://dx.doi.org/10.21468/SciPostPhysCore.6.4.066},
   DOI={10.21468/scipostphyscore.6.4.066},
   eprint={2207.10712},
   archivePrefix={arXiv},
   primaryClass={cond-mat.str-el},
   number={4},
   journal={SciPost Physics Core},
   publisher={Stichting SciPost},
   author={Moradi, Heidar and Moosavian, Seyed Faroogh and Tiwari, Apoorv},
   year={2023},
   month=oct }

@article{freed2024,
title = "Topological symmetry in quantum field theory",
author = "Freed, Daniel S. and Moore, Gregory W. and Constantin Teleman",
note = "Publisher Copyright: {\textcopyright} 2024 European Mathematical Society.",
year = "2024",
month = oct,
day = "28",
doi = "10.4171/qt/223",
eprint = "2209.07471",
archivePrefix = "arXiv",
primaryClass = "hep-th",
language = "English (US)",
volume = "15",
pages = "779--869",
journal = "Quantum Topology",
issn = "1663-487X",
publisher = "European Mathematical Society Publishing House",
number = "3-4",
}

@article{Bhardwaj2024,
  title = {Categorical Landau Paradigm for Gapped Phases},
  author = {Bhardwaj, Lakshya and Bottini, Lea E. and Pajer, Daniel and Sch\"afer-Nameki, Sakura},
  journal = {Phys. Rev. Lett.},
  volume = {133},
  issue = {16},
  pages = {161601},
  numpages = {6},
  year = {2024},
  month = {Oct},
  publisher = {American Physical Society},
  doi = {10.1103/PhysRevLett.133.161601},
  eprint = {2310.03786},
  archivePrefix = {arXiv},
  primaryClass = {cond-mat.str-el},
  url = {https://link.aps.org/doi/10.1103/PhysRevLett.133.161601}
}

@article{YWL_DefectsIn2+1TO,
  title = {Condensation completion and defects in $2+1\text{D}$ topological orders},
  author = {Yue, Gen and Wang, Longye and Lan, Tian},
  journal = {Phys. Rev. B},
  volume = {112},
  issue = {10},
  pages = {104106},
  numpages = {28},
  year = {2025},
  month = {Sep},
  publisher = {American Physical Society},
  doi = {10.1103/fgnx-t5bj},
  eprint = {2402.19253},
  archivePrefix = {arXiv},
  primaryClass = {cond-mat.str-el},
  url = {https://link.aps.org/doi/10.1103/fgnx-t5bj}
}

@Article{KLW+2005.14178,
  author        = {Kong, Liang and Lan, Tian and Wen, Xiao-Gang and Zhang, Zhi-Hao and Zheng, Hao},
  journal       = {Physical Review Research},
  title         = {Algebraic higher symmetry and categorical symmetry: a holographic and entanglement view of symmetry},
  year          = {2020},
  issn          = {2643-1564},
  month         = oct,
  number        = {4},
  pages         = {043086},
  volume        = {2},
  abbr          = {PRR},
  archiveprefix = {arXiv},
  arxivid       = {2005.14178},
  copyright     = {arXiv.org perpetual, non-exclusive license},
  doi           = {10.1103/physrevresearch.2.043086},
  eprint        = {2005.14178},
  publisher     = {American Physical Society ({APS})},
  selected      = {true},
}

@Article{Lan2412.07198,
  author        = {Lan, Tian},
  journal       = {Communications in Mathematical Physics},
  title         = {Tube Category, Tensor Renormalization and Topological Holography},
  year          = {2025},
  issn          = {1432-0916},
  month         = jul,
  number        = {8},
  pages         = {196},
  volume        = {406},
  archiveprefix = {arXiv},
  copyright     = {arXiv.org perpetual, non-exclusive license},
  creationdate  = {2025-07-16T17:35:11},
  doi           = {10.1007/s00220-025-05383-6},
  eprint        = {2412.07198},
  primaryclass  = {math-ph},
  publisher     = {Springer Science and Business Media LLC},
  selected      = {true},
}

@ARTICLE{2026arXiv260506661Y,
       author = {{Yue}, Gen and {Bai}, Ansi and {Wu}, Linqian and {Lan}, Tian},
        title = "{Pro-Tensor Network}",
      journal = {arXiv e-prints},
         year = 2026,
        month = may,
          eid = {arXiv:2605.06661},
        pages = {arXiv:2605.06661},
          doi = {10.48550/arXiv.2605.06661},
archivePrefix = {arXiv},
       eprint = {2605.06661},
 primaryClass = {cond-mat.str-el},
       adsurl = {https://ui.adsabs.harvard.edu/abs/2026arXiv260506661Y}
}

@ARTICLE{2021Gaiotto,
       author = {{Gaiotto}, Davide and {Kulp}, Justin},
        title = "{Orbifold groupoids}",
      journal = {Journal of High Energy Physics},
         year = 2021,
        month = feb,
       volume = {2021},
       number = {2},
          eid = {132},
        pages = {132},
          doi = {10.1007/JHEP02(2021)132},
archivePrefix = {arXiv},
       eprint = {2008.05960},
 primaryClass = {hep-th},
       adsurl = {https://ui.adsabs.harvard.edu/abs/2021JHEP...02..132G}
}

@ARTICLE{2023SymTFTfromString,
       author = {{Apruzzi}, Fabio and {Bonetti}, Federico and {Garc{\'\i}a Etxebarria}, I{\~n}aki and {Hosseini}, Saghar S. and {Sch{\"a}fer-Nameki}, Sakura},
        title = "{Symmetry TFTs from String Theory}",
      journal = {Communications in Mathematical Physics},
         year = 2023,
        month = aug,
       volume = {402},
       number = {1},
         pages = {895--949},
          doi = {10.1007/s00220-023-04737-2},
archivePrefix = {arXiv},
       eprint = {2112.02092},
 primaryClass = {hep-th},
       adsurl = {https://ui.adsabs.harvard.edu/abs/2023CMaPh.402..895A}
}
